%% file: main.tex
\documentclass[acmsmall]{acmart}

\AtBeginDocument{%
  }

\setcopyright{cc}
\setcctype{by}
\acmDOI{10.1145/3808312}
\acmYear{2026}
\acmJournal{PACMPL}
\acmVolume{10}
\acmNumber{PLDI}
\acmArticle{234}
\acmMonth{6}
\acmSubmissionID{pldi26main-p463-p}
\received{2025-11-13}
\received[accepted]{2026-04-03}

\setcitestyle{numbers,sort&compress}

\usepackage[noend]{algpseudocode}
\usepackage{algorithm}
\usepackage{url}
\usepackage{subcaption}
\usepackage{comment}
\usepackage{multirow}
\usepackage{wrapfig}
\usepackage{enumitem}
\usepackage{hyperref}
\usepackage{xspace}
\usepackage{bbding}
\usepackage{quiver}
\usepackage{graphics}
\usepackage{mdframed}
\usepackage{graphicx}
\usepackage[rightcaption]{sidecap}
\usepackage[noend]{algpseudocode}
\usepackage{amsmath}
\usepackage{mathtools}  
\usepackage{makecell}
\usepackage{amssymb}
\usepackage{pifont}
\usepackage[utf8]{inputenc}

\usepackage{pgfplots}
\usepackage{tikz} 
\usetikzlibrary{arrows.meta}
\usepgflibrary{shapes.multipart} 
\usepgflibrary[shapes.multipart] 
\usetikzlibrary{shapes.multipart} 
\usetikzlibrary[shapes.multipart] 
\usetikzlibrary{shapes,positioning,calc,backgrounds}
\tikzset{
  bigstage/.style = {rectangle, rounded corners=2pt,
                     text height=1.3ex,
                     text depth=.3ex,
                     inner xsep=2pt,
                     inner ysep=2.5pt,
                     minimum width=1.8cm,
                     draw=none,
                     font=\scriptsize,
                     align=center,
                     anchor=center},
  smallstage/.style = {rectangle, rounded corners=2pt,
                       text height=1.3ex,
                       text depth=.3ex,
                       inner xsep=2pt,
                       inner ysep=2.5pt,
                       minimum width=1.45cm,
                       draw=none,
                       font=\scriptsize,
                       align=center,
                       anchor=center},
  smallarrow/.style = {-{Latex[length=4pt,width=4pt]}, line width=0.7pt},
}
\definecolor{lightblue}{HTML}{CFDEE7}
\definecolor{lightpink}{HTML}{F6E7E5}
\definecolor{lightgreen}{HTML}{DBE7E7}
\definecolor{lightorange}{HTML}{FAD8B3}
\colorlet{lightgray}{gray!20}

\usepackage{listings}
\usepackage[frozencache,cachedir=.]{minted}
\definecolor{codebg}{rgb}{0.95,0.95,0.96}
\usepackage{caption}
\usepackage{xcolor}
\usepackage{amsmath}
\usepackage[commandnameprefix=ifneeded]{changes}
\setdeletedmarkup{\textcolor{red}{\sout{#1}}}
\setaddedmarkup{\textcolor{blue}{#1}}
\usepackage{dsfont}
\usepackage{tabularx}
\usepackage{chapterbib}
\usepackage[normalem]{ulem} 

\makeatletter
\patchcmd{\Changes@output}
  {%
    \IfStrEq{#1}{replaced}{{\Changes@Markup@added{#3}}\allowbreak\Changes@Markup@deleted{#4}}{}%
  }{%
    \IfStrEq{#1}{replaced}{{\Changes@Markup@deleted{#4}}\allowbreak\Changes@Markup@added{#3}}{}%
  }{}{}
\makeatother

\lstdefinestyle{pythonstyle}{
  language=Python,
  basicstyle=\ttfamily\footnotesize,
  keywordstyle=\color{blue},
  stringstyle=\color{green!40!black},
  commentstyle=\color{gray},
  frame=single,
  showstringspaces=false,
  breaklines=true,
  tabsize=2
}



\input{macros}

\input{abstract}

\begin{document}

\title{Optimal Predicate Pushdown Synthesis}

\author{Robert Zhang}
\orcid{0009-0001-8853-5813}
\affiliation{%
  \institution{University of Texas at Austin}
  \city{Austin}
  \country{USA}
}
\email{robertz@cs.utexas.edu}

\author{Eric Hayden Campbell}
\orcid{0000-0001-5954-2136}
\affiliation{%
  \institution{University of Texas at Austin}
  \city{Austin}
  \country{USA}
}
\email{eric.hayden.campbell@utexas.edu}

\author{Dixin Tang}
\orcid{0000-0002-3316-6651}
\affiliation{%
  \institution{University of Texas at Austin}
  \city{Austin}
  \country{USA}
}
\email{dixin@utexas.edu}

\author{Işıl Dillig}
\orcid{0000-0001-8006-1230}
\affiliation{%
  \institution{University of Texas at Austin}
  \city{Austin}
  \country{USA}
}
\email{isil@cs.utexas.edu}

\maketitle

\input{sections/intro-new}
\input{sections/problem}

\input{sections/methodology}
\input{sections/synthesis}
\input{sections/eval}
\input{sections/related}

\input{sections/concl}

\input{sections/availability}

\begin{acks}
We thank the anonymous reviewers for their thoughtful feedback and suggestions. This work was conducted in a research group supported by NSF awards CCF-1918889, CNS-2120696, CCF-2210831, and CCF-2319471, CCF-2422130, CCF-2403211, CCF-2505865, CCF-2326576, and CCF-2403211, as well as a DARPA award under agreement HR00112590133 and a gift from Amazon.
\end{acks}

\bibliography{main}

\newpage
\appendix
\input{sections/appendix-proofs}

\end{document}

%% file: macros.tex
\newcommand{\irule}[2]%
   {\mkern-2mu\displaystyle\frac{#1}{\vphantom{,}#2}\mkern-2mu}

\newcommand{\bfpara}[1]{\vspace{3.7pt} \noindent {\bf \emph{#1}}}

\newcommand{\numBenchmarks}{150\xspace}
\newcommand{\numUDF}{26\xspace}

\newcommand{\ablnb}{{\sc Pusharoo-NoBounds}\xspace}

\newcommand{\ablspacer}{{\sc Pusharoo-Spacer}\xspace}
\newcommand{\ableld}{{\sc Pusharoo-Eldarica}\xspace}
\newcommand{\ablnr}{{\sc Pusharoo-NoRepair}\xspace}
\newcommand{\abltp}{{\sc Pusharoo-TwoPhase}\xspace}
\newcommand{\ablenum}{{\sc Pusharoo-NoPrune}\xspace}

\newcommand{\true}{\mathsf{true}}

\newcommand{\Lift}{\mathsf{Lift}}
\newcommand{\toolname}{{\sc Pusharoo}\xspace}
\newcommand{\magicpush}{{\sc MagicPush}\xspace}
\newcommand{\bisim}{\psi}
\newcommand{\filter}{\mathsf{filter}}

\newcommand{\init}{\textsf{Init}\xspace}
\newcommand{\sync}{\textsf{Sync}\xspace}
\newcommand{\stutter}{\textsf{Stutter}\xspace}
\newcommand{\final}{\textsf{Final}\xspace}

\newcommand{\supsharp}{\mathord{\smash[t]{\sharp}}}

\newcommand{\orig}[1]{\textcolor{blue!60!black}{#1}}
\newcommand{\opt}[1]{\textcolor{red!60!black}{#1}}

\makeatletter
\algrenewcommand\ALG@beginalgorithmic{\small} 
\algrenewcommand\alglinenumber[1]{\small #1:}
\makeatother

%% file: abstract.tex
\begin{abstract}

Predicate pushdown is a long-standing {performance optimization that filters}
data as early as possible in a
{computational workflow}. In modern data pipelines, this transformation is especially important because much of the computation occurs inside \emph{user-defined functions (UDFs)} written in general-purpose languages such as Python and Scala. These UDFs capture rich domain logic and complex aggregations and are among the most expensive operations in a pipeline. Moving filters ahead of such UDFs can yield substantial performance gains, but doing so requires
\emph{semantic} reasoning.
This paper introduces a general semantic foundation for predicate pushdown over stateful fold-based  computations.

We view pushdown as a correspondence between two programs that process different subsets of input data, with correctness witnessed by a \emph{bisimulation invariant} relating their internal states. Building on this foundation, we develop a sound and relatively complete framework for verification, alongside a synthesis algorithm that automatically constructs \emph{optimal pushdown decompositions} by finding the strongest admissible pre-filters and weakest residual post-filters.  We implement this approach in a tool called \toolname and evaluate it on 150 real-world \texttt{pandas} and Spark data-processing pipelines. Our evaluation shows that \toolname is significantly more expressive than prior work, producing optimal pushdown transformations with a median synthesis time of {1.6} seconds per benchmark. Furthermore, our experiments demonstrate that the discovered pushdown optimizations speed up end-to-end execution by an average of 2.4$\times$ and up to two orders of magnitude.

\end{abstract}

%% file: sections/intro-new.tex
\section{Introduction} 
\label{sec:intro}

Program optimizations that reorder computation are pervasive in modern data-processing systems. One of the most fundamental among them is \emph{predicate pushdown}, which seeks to move filtering conditions earlier in a {data pipeline} so that irrelevant data is discarded before expensive {operations} are applied~[\citenum{database08} (Ch.~16.2.3), \citenum{ullman90}]. As illustrated in Figure~\ref{fig:combined-pipeline}, this transformation replaces a pipeline that filters results \emph{after} a computation with an equivalent one that filters inputs (partially or fully) \emph{before} the computation, thereby reducing downstream computation cost.

\begin{wrapfigure}{r}[0pt]{4.5cm}
  \centering
  \captionsetup{skip=2pt}
  \begin{subfigure}[t]{0.49\linewidth}
    \centering
    \begin{tikzpicture}[node distance=0.3cm,outer sep=0pt,baseline=(current bounding box.north)]
      \node[inner sep=2pt] (in)  {\scriptsize Input data};
      \node[bigstage, opacity=0] (phantom) [below=of in] {Phantom};
      \node[bigstage, fill=lightblue] (udf) [below=of phantom] {Computation};
      \node[bigstage, fill=lightorange] (filter) [below=of udf] {Original filter $P$};
      \node[inner sep=1.5pt] (out) [below=of filter] {\scriptsize Query result};

      \draw[smallarrow] (in)  -- (udf);
      \draw[smallarrow] (udf) -- (filter);
      \draw[smallarrow] (filter) -- (out);
    \end{tikzpicture}
  \end{subfigure}%
  \hfill
  \begin{subfigure}[t]{0.49\linewidth}
    \centering
    \begin{tikzpicture}[node distance=0.3cm,outer sep=0pt,baseline=(current bounding box.north)]
      \node[inner sep=2pt] (in) {\scriptsize Input data};
      \node[bigstage, fill=lightgreen] (filter) [below=of in] {Input filter $Q$};
      \node[bigstage, fill=lightblue] (udf) [below=of filter] {Computation};
      \node[bigstage, fill=lightpink] (residual) [below=of udf] {Residual filter $P'$};
      \node[inner sep=1.5pt] (out) [below=of residual] {\scriptsize Query result};

      \draw[smallarrow] (in)  -- (filter);
      \draw[smallarrow] (filter)  -- (udf);
      \draw[smallarrow] (udf) -- (residual);
      \draw[smallarrow] (residual) -- (out);
    \end{tikzpicture}
  \end{subfigure}
  \caption{Original (left) and optimized (right) pipelines}
  \label{fig:combined-pipeline}
  \vspace{-0.1in}
\end{wrapfigure}

This optimization is particularly important in modern data pipelines, where much of the computation occurs inside \emph{user-defined functions (UDFs)} written in general-purpose programming languages such as Python and Scala. These functions let developers express rich domain logic and complex aggregations, but they are also among the \emph{most expensive} operations in a pipeline. Moving filters ahead of such computations can therefore yield substantial performance gains by reducing the amount of data processed by a UDF. The challenge, however, is that many UDFs are
fold-like computations that maintain complex internal state and involve nested control flow,
requiring rich semantic reasoning for correctness. Recent work~\cite{MagicPush}
has begun to explore predicate pushdown for restricted classes of UDFs, but a robust, general account of predicate pushdown for complex, stateful UDFs has yet to be systematically studied.

This paper aims to fill exactly this gap.
As a starting point, we develop a unified semantic foundation of predicate pushdown for stateful UDFs, which we model as pure, deterministic folds whose observable behavior is captured by the evolution of an internal accumulator state. Within this framework, a pushdown transformation is interpreted as a semantic correspondence between two state-transforming programs that operate on different subsets of input data. Correctness then amounts to {proving a relational property, namely,} a \emph{bisimulation invariant}~\cite{park81}, that links the executions of the original and transformed programs. This formulation accommodates the full generality of stateful UDFs, whose internal states may evolve differently {before and after pushdown,} but must converge to the same final result. From this foundation, we then develop a sound and relatively complete verification framework for candidate pushdown optimizations.

Building on this verification framework, we then introduce a new synthesis algorithm for automatically discovering \emph{optimal pushdown decompositions} of a post-computation filter~$P$. Specifically, the algorithm splits $P$ into a pre-computation filter~$Q$ that is as logically strong as possible, and a \emph{residual} filter~$P'$ that is as weak as possible, within a fixed predicate universe. This notion of optimality is important because it maximizes early pruning of input data while also minimizing redundant post-UDF validation.
However, a key challenge in solving this {optimal synthesis} problem is that the algorithm must jointly infer three tightly-coupled proof artifacts: a (strongest) input-side filter $Q$, a (weakest) residual filter $P'$, and a bisimulation invariant that witnesses their correctness. {Even} ignoring the optimality requirement, this problem lies beyond the reach of existing methods. For instance, unlike many predicate synthesis tasks that reduce to solving Constrained Horn Clauses (CHCs), the verification conditions connecting these unknown predicates are inherently \emph{non-Horn}, making them incompatible with off-the-shelf CHC solvers. {On top of this}, the optimality requirement further complicates the search, as $Q$ and $P'$ must be optimized in opposite directions.

Our  synthesis method solves these challenges through three complementary algorithmic innovations. First, we {structurally decompose the search space}
into tractable sub-problems \emph{without} compromising optimality. Second, we
obtain sound under- and over-approximations for any valid bisimulation invariant and use these symbolic bounds as \emph{unrealizability certificates} during synthesis. Finally, we incorporate domain-specific \emph{root cause analysis} for diagnosing failed synthesis attempts and use this to inform \emph{predicate repair}.  Collectively, these techniques effectively navigate the search space in a principled way, ultimately guaranteeing both correctness and optimality.

We implemented our approach in a tool called \toolname and evaluated it on 150 real-world benchmarks drawn from \texttt{pandas} and Apache Spark workloads.
\toolname successfully synthesizes correct pushdown transformations for all benchmarks and substantially outperforms prior work, which can  handle less than 15\%
and often yields suboptimal decompositions. Furthermore, the algorithm runs efficiently in practice, with a median synthesis time of 1.6 seconds. Finally, the optimizations enabled by our method result in significant speed-ups of real-world data analytics pipelines, yielding an average speed-up of $2.4\times$ and up to two orders of magnitude in some cases.

To summarize, this paper makes the following key contributions:
\begin{itemize}[leftmargin=*]
\item We introduce a unified semantic formulation of predicate pushdown that subsumes and extends prior definitions.
\item We present a sound and relatively complete methodology for verifying the correctness of candidate pushdown transformations for stateful computations.
\item We develop a novel synthesis algorithm that can generate \emph{optimal} and \emph{provably correct} pushdown optimizations. Notably, our method simultaneously searches for an optimal pushdown optimization together with its  certificate of correctness.
\item We empirically evaluate our method on 150 real-world benchmarks and demonstrate its effectiveness in terms of (1) expressiveness, (2) optimization time, and (3) the end-to-end performance improvements it enables.
\end{itemize}

%% file: sections/problem.tex
\section{Background}\label{sec:background}

In this section, we review the two main forms of predicate pushdown studied in prior work and introduce the class of user-defined functions (UDFs) to which our formulation applies. Since UDFs serve as the building blocks of modern data analytics, we first discuss how they are structured and used in practice. Throughout the paper, we use the terms \emph{UDF} and \emph{computation} interchangeably.

\subsection{User-Defined Functions in Data Pipelines}\label{sec:udfs-in-pipelines}

In this paper, 
we consider a general class of UDFs that can be expressed as as \emph{state-transforming folds}, i.e., computations that traverse a sequence of input rows while updating an internal state.  Each UDF operates over a \emph{dataframe}, which we model as a finite sequences of tuples, sometimes referred to as \emph{rows}.  We assume that each UDF can be expressed in terms of a left fold:
\begin{equation}\label{eq:fold}
F(x) = \mathsf{fold}(x, I, f) = f(\ldots f(f(I, r_1), r_2), \ldots, r_n)
\end{equation}
where $x = [r_1, \ldots, r_n]$ is the input dataframe, $I$ is the initial state, and $f : A \times R \to A$ is a pure, deterministic update function called the \emph{accumulator}. 
The \emph{internal state} of type~$A$ is initialized to~$I$ and updated row by row by applying~$f$. We adopt the convention that $\textsf{fold}([], I, f)$ yields a distinguished undefined value $\bot$, and we assume that UDF evaluation is pure with respect to the pipeline. That is, external side effects, such as I/O and mutation of global state, are out of scope unless they are explicitly reified in the
internal state. However, we do not impose   restrictions on the structure of the internal state, which may be a scalar, tuple, list, or even a dataframe. Thus, our formalism can model a wide variety of UDFs, including:
\begin{itemize}[leftmargin=*]
  \item \textbf{Row-wise UDFs} correspond to folds where the input dataframe consists of a single row, and the accumulator~$f$ applies a stateless transformation to that row.
  
  \item \textbf{Aggregation UDFs} summarize the entire input into a single output value. These functions typically maintain a running summary (e.g., sum, maximum, or top-$k$ elements), which~$f$ updates incrementally across the sequence of rows.
  
  \item \textbf{Table-valued UDFs} produce a variable number of output rows. In our model, these are captured by using a list- or dataframe-valued internal state,
  where~$f$ appends new rows.
\end{itemize}



\begin{figure}[t]
  \centering
  \begin{subfigure}{0.49\linewidth}
    \centering
    \begin{minted}[
      fontsize=\footnotesize,
      frame=lines,
      framesep=2mm
    ]{python}
# Filter to premium items
premium = df[df['category'] == 'premium']
# Compute discounted prices for selected items
discounted = premium.apply(
  lambda r: r['price'] * 0.9, axis=1)
# Retain those w/ discounted price at/above $900
filtered = discounted[discounted >= 900]
    \end{minted}
    \captionsetup{skip=2pt}
    \caption{Original pipeline}
    \label{fig:rowwise-original}
  \end{subfigure}\hfill
  \begin{subfigure}{0.49\linewidth}
    \centering
    \begin{minted}[
      fontsize=\footnotesize,
      frame=lines,
      framesep=2mm
    ]{python}
# Filter to premium items at/above $1000
prefiltered = df[
  (df['category'] == 'premium') &
  (df['price'] >= 1000)]
# Apply discounting UDF after filtering
discounted = prefiltered.apply(
  lambda r: r['price'] * 0.9, axis=1)
    \end{minted}
    \captionsetup{skip=2pt}
    \caption{Exact pushdown}
    \label{fig:rowwise-exact-pushdown}
  \end{subfigure}
  \captionsetup{skip=2pt}
  \caption{{(a) A pipeline that retrieves premium items with a discounted price at or above \$900 (b) Exact pushdown: the post-processing filter is pushed through the discounting UDF entirely.}}
  \vspace{-0.1in}
\end{figure}

\vspace{0.07in}
While this fold-based view captures the core behavior of individual UDFs, such functions rarely appear in isolation. In practice, they are embedded within larger data-processing workflows that combine multiple relational operators (such as joins, group-bys, and filters) with user-defined computations. We refer to such workflows as \emph{pipelines} and provide two simple examples below. 





\begin{example}
Figure~\ref{fig:rowwise-original} shows a \texttt{pandas} pipeline that (1) selects only \texttt{premium} items based on the \texttt{category} attribute, (2) applies a row-wise UDF that computes a 10\% discount by multiplying the \texttt{price} field by~0.9, and (3) retains only those items whose discounted price is at least~\$900.
\end{example}
\vspace{-0.5em}



\begin{example}
Figures~\ref{fig:top2-udaf} and~\ref{fig:udaf-original} show a pipeline involving a user-defined aggregation. Unlike the previous row-wise UDF, \texttt{top2} maintains state across rows to compute the top two values in a group. The pipeline in Figure~\ref{fig:udaf-original} applies this UDF per \texttt{team}, extracting the two highest \texttt{score}s for each team, then applies a filter to keep only teams whose top scores both exceed~90.0.
\end{example}


\subsection{Two Flavors of Predicate Pushdown Optimizations}\label{sec:flavors}

In this section, we describe two different types of pushdown optimizations considered in prior work~\cite{MagicPush,sparser,ullman90}, namely, \emph{exact} (\emph{equivalent}) pushdown and \emph{partial} (\emph{superset}) pushdown. Throughout this section, we use the letter $P$ to denote a predicate over the output of the UDF and $Q$ to denote a predicate over \emph{an input row}.


\paragraph{Exact Pushdown.} Given a computation $F$ followed by a predicate $P$, an \emph{exact (equivalent) pushdown} optimization can be applied if the post-computation predicate $P$ can be entirely eliminated. Formally, exact pushdown seeks a \emph{pre-computation} predicate $Q$ satisfying the following condition: 
\[
\forall x.\;\bigl( y = F(x) \wedge z = F(\filter_Q(x)) \bigr)
\;\Longrightarrow\;
\bigl((P(y) \Rightarrow y = z)\; \wedge\; (\neg P(y) \Rightarrow z = \bot)\bigr),
\]
where $\filter_Q(x)$ retains only elements of $x$ satisfying $Q$, and $\bot$ denotes the absence of any result, treated as failing any filter (i.e., $P(\bot)$ is \textsf{false}). Thus $Q$ is a valid pushdown iff applying it before $F$ yields the same outcome as applying $P$ after $F$, the classical form in relational query optimization.


\begin{example}
Recall the pipeline from Figure~\ref{fig:rowwise-original}, which computes discounted prices using a row-wise UDF and then filters the results to retain only {items priced} at or above \$900. Let $F$ denote the discounting UDF and $P(a) = (a \ge 900)$ be the post-UDF filter. Since $F(r) = r[`\mathsf{price}\textrm'] \cdot 0.9$, the condition $P(F(r))$ holds iff $r[`\mathsf{price}\textrm'] \ge 1000$. Thus, we can push the predicate through $F$ exactly by defining $Q(r) = (r[`\mathsf{price}\textrm'] \ge 1000)$, as shown in Figure~\ref{fig:rowwise-exact-pushdown}.  This transformation satisfies the definition of exact pushdown: it guarantees that the output is identical to the original, since rows that would have failed $P$ are excluded and those that would have passed are included.
\end{example}



\begin{figure}[t]
  \centering
  \begin{subfigure}[t]{0.49\linewidth}
    \centering
    \begin{minted}[
      fontsize=\footnotesize,
      frame=lines,
      framesep=2mm
    ]{python}
# User-defined aggregation function
def top2(x):
  fst = snd = float('-inf')
  for r in x:
    if r > fst: snd, fst = fst, r
    elif r > snd: snd = r
  return [fst, snd]
    \end{minted}
    \captionsetup{skip=2pt}
    \caption{UDF \texttt{top2}}
    \label{fig:top2-udaf}
  \end{subfigure}\hfill
  \begin{subfigure}[t]{0.49\linewidth}
    \centering
    \begin{minted}[
      fontsize=\footnotesize,
      frame=lines,
      framesep=2mm
    ]{python}
top2s = (df
  .groupby('team')['score']
  .agg(top2))
# P: keep only teams whose two highest scores
#    both exceed 90.0
filtered = top2s[top2s.apply(lambda a:
  a[0] > 90.0 and a[1] > 90.0)]
    \end{minted}
    \captionsetup{skip=2pt}
    \caption{Original pipeline}
    \label{fig:udaf-original}
  \end{subfigure}
  \par\vspace{0.5em}\par
  \begin{subfigure}[t]{0.49\linewidth}
    \centering
    \begin{minted}[
      fontsize=\footnotesize,
      frame=lines,
      framesep=2mm
    ]{python}
# Filter out low scores before aggregation
prefiltered = df[df['score'] > 90.0]
top2s = (prefiltered
  .groupby('team')['score']
  .agg(top2))
# Apply full predicate after aggregation
filtered = top2s[top2s.apply(lambda a:
  a[0] > 90.0 and a[1] > 90.0)]
    \end{minted}
    \captionsetup{skip=2pt}
    \caption{Partial pushdown}
    \label{fig:udaf-partial-pushdown}
  \end{subfigure}\hfill
  \begin{subfigure}[t]{0.49\linewidth}
    \centering
    \begin{minted}[
      fontsize=\footnotesize,
      frame=lines,
      framesep=2mm
    ]{python}
prefiltered = df[df['score'] > 90.0] # Q
top2s = (prefiltered
  .groupby('team')['score']
  .agg(top2))
# P': need only check whether a team has seen
#     any second-highest score
filtered = top2s[top2s.apply(lambda a:
  a[1] != float('-inf'))]
    \end{minted}
    \captionsetup{skip=2pt}
    \caption{{Split} pushdown}
    \label{fig:udaf-split-pushdown}
  \end{subfigure}

  \captionsetup{skip=2pt}
  \caption{(a) A UDF that computes the two highest scores for each team. (b) A pipeline using the UDF to retrieve teams whose two highest scores exceed 90.0. (c) Partial pushdown: the input filter eliminates rows that cannot help satisfy the post-computation filter. (d) {Split} pushdown: the post-processing predicate is weakened to $P'$.}
\end{figure}

\paragraph{Partial Pushdown.}
When exact pushdown is not feasible, prior work~\cite{MagicPush,sparser} has also considered a related optimization called \emph{partial (superset) pushdown}. This variant performs filtering before executing $F$, but also applies the original filter $P$ afterwards. In other words, it may conservatively preserve additional data that
would cause $P$ to fail,
and then filters those ``false positives'' out using $P$ after $F$ is applied. Formally, partial pushdown is possible if there exists a predicate $Q$ such that:
    \[
        \forall x.\;\bigl( y = F(x) \wedge z = F(\filter_Q(x)) \bigr)
        \;\Longrightarrow\;
        \bigl((P(y) \Rightarrow P(z))
        \;\wedge\;
        (P(z) \Rightarrow y = z)\bigr)
    \]
Intuitively, $Q$ is a conservative over-approximation of $P$: the first clause $P(y) \Rightarrow P(z)$ ensures \emph{soundness} (no valid result is lost), and the second $P(z) \Rightarrow y = z$ ensures \emph{consistency} (any accepted output matches the original).

\begin{example}
Consider again the pipeline from Figures~\ref{fig:top2-udaf} and~\ref{fig:udaf-original}, which uses a UDF to compute the top two scores for each team. 
Let $F$ denote  the \texttt{top2} function  and $P(a) = (a[0] > 90.0 \wedge a[1] > 90.0)$ be the post-UDF filter from Figure~\ref{fig:udaf-original}. This predicate requires that both the highest and second-highest scores of a team exceed the threshold. 
For this example, exact pushdown is not possible because
any row‐wise test $Q(r)$ can only inspect one score at a time and \emph{cannot} decide whether a given row will end up as the highest or second highest without having access to the entire dataframe. 
However, \emph{partial} pushdown is possible using the predicate $Q(r) = (r[`\mathsf{score}\textrm'] > 90.0)$, which eliminates only those rows that cannot contribute to satisfying the post-filter $P$.
Since a team may not satisfy $P$ after aggregation (e.g., if only one score exceeds the threshold),
the partially pushed-down pipeline re-applies $P$, as shown  in Figure~\ref{fig:udaf-partial-pushdown}.

\end{example}

\section{Problem Statement}\label{sec:problem}

We now introduce a unified formulation of predicate pushdown that generalizes both exact and partial pushdown. This formulation  enables a single algorithm that can synthesize a wide range of pushdown strategies, including those that fall strictly outside the exact and partial cases.


First, given a predicate $P$ and element $a$, we define a function $\Lift(P, a)$ that returns the result of the computation only when $P$ is satisfied and yields $\bot$ otherwise:
\[
\Lift(P, a) = a \;\text{if } P(a),\; \text{else } \bot
\]
Intuitively, $\Lift$ decides whether to keep or discard an internal state based on $P$, with $\bot$ playing the role for internal states that the empty sequence plays for dataframes. We can now define a unified formulation of predicate pushdown as follows:

\begin{definition}{\bf (Generalized Predicate Pushdown)}
\label{def:generalized-pushdown}
Let $F : [R] \to A$ be a computation and $P : A \to \mathsf{Bool}$ a predicate on its output. Generalized predicate pushdown seeks to find a pair of predicates 
$
Q : R \to \mathsf{Bool}
\ \text{and} \
P' : A \to \mathsf{Bool}
$
such that the following equivalence holds:
\begin{equation}\label{eq:unified-lifted}
\forall x.\quad
\Lift(P, F(x))
\;=\;
\Lift(P', F(\filter_Q(x)))
\end{equation}
\end{definition}
Intuitively, the input-side filter $Q$ \emph{pre-filters} the input to $F$, removing rows that can never contribute to satisfying $P$, and $P'$ is a \emph{residual} predicate over $F$'s output that may be strictly weaker than $P$ and  captures \emph{exactly} the information not preserved by $Q$. We refer to any pair $(Q, P')$ satisfying Equation~\eqref{eq:unified-lifted} as a \emph{solution} to the \emph{generalized} predicate pushdown problem.
As the following theorem shows, this formulation subsumes both exact and partial pushdown as special cases, recovering them by choosing $P'$ to be $\true$
or $P$, respectively.\footnote{Proofs of all theorems are provided in Appendix~\ref{sec:proofs}.}

\begin{proposition}\label{thm:exact-partial-cases}
Let $F$ be a computation on inputs $x$ and $P$ a predicate. For any input-side filter $Q$:
\begin{enumerate}[leftmargin=*]
  \item $Q$ satisfies the \emph{exact pushdown} condition for $(F,P)$ precisely when $(Q, \textsf{true})$
        satisfies Equation~\eqref{eq:unified-lifted}
  \item $Q$ satisfies the \emph{partial pushdown} condition for $(F,P)$ precisely when 
        $(Q, P)$ satisfies Equation~\eqref{eq:unified-lifted}.
\end{enumerate}
\end{proposition}

\begin{example}
The pipeline from Figure~\ref{fig:udaf-split-pushdown} illustrates a case where our formulation admits a decomposition that is different  from both exact and partial pushdown. Recall that the post-UDF filter $P$ in this example is $a[0] > 90.0 \land a[1] > 90.0$ and that \emph{partial} pushdown can be applied with $Q(r) = (r[`\mathsf{score}\textrm'] > 90.0)$. Using the same $Q$, our formulation enables a strictly weaker residual predicate, namely $P'(a) = (a[1] \neq -\infty)$, which removes redundant checks. In particular, the original predicate $P$ explicitly verifies both scores against the 90.0 threshold. However, after applying  $Q$, any defined top scores must already be above this threshold. Thus, the check performed by $P$ after filtering is partially redundant, and the weaker predicate $P'$ avoids this redundancy.

\end{example}

As illustrated in this example, some pipelines admit a decomposition in which the input filter~$Q$ rules out some, but not all, irrelevant rows, and the residual predicate~$P'$ still enforces part of the original post-computation filter. We refer to such cases as instances of \textbf{split pushdown}, reflecting that responsibility for enforcing the filter is divided between~$Q$ and~$P'$.

\paragraph{Optimal Decomposition.}  
While Equation~\eqref{eq:unified-lifted} admits sound decompositions that go beyond exact and partial pushdown, it also  admits trivial solutions: for instance, taking $Q \equiv \true$ and $P' \equiv P$ always satisfies Equation~\eqref{eq:unified-lifted}. To avoid such vacuous solutions and realize the practical benefits of predicate pushdown, we additionally seek decompositions that maximize input reduction and minimize post-processing effort. Intuitively, we want the input filter $Q$ to be as \emph{strong} as possible so that irrelevant rows are excluded early, but we want  the residual $P'$ to be as \emph{weak} as possible so that the post-UDF check is minimal. We formalize these goals as follows:

\begin{definition}{\bf (Optimal Predicate Pushdown)}
\label{def:optimal-pushdown}
Let $F : [R] \to A$ be a computation and $P : A \to \mathsf{Bool}$ be a predicate on its output. An optimal solution to the generalized predicate pushdown problem is a pair $(Q^*, P^*)$, where $Q^* : R \to \mathsf{Bool}$ and $P^* : A \to \mathsf{Bool}$, such that
\begin{equation}\label{eq:def-pushdown}
\forall x.\quad \Lift(P, F(x)) = \Lift(P^*, F(\filter_{Q^*}(x)))
\end{equation}
and $(Q^*, P^*)$ satisfies the following optimality conditions:
\begin{itemize}[leftmargin=*]
\item \textbf{Strongest input filter:}Among all solutions to Equation~\eqref{eq:def-pushdown}, $Q^*$ is maximal under implication. That is, there is no other pair $(Q', P')$ satisfying Equation~\eqref{eq:def-pushdown} such that $Q' \Rightarrow Q^*$ and $Q^* \not\Rightarrow Q'$.
\item \textbf{Weakest residual check:} Fixing the input filter $Q^*$, among all solutions to Equation~\eqref{eq:def-pushdown} with $Q^*$, the residual $P^*$ is minimal under implication. That is, there is no other pair $(Q^*, P')$ satisfying Equation~\eqref{eq:def-pushdown} such that $P^* \Rightarrow P'$ and $P' \not\Rightarrow P^*$.
\end{itemize}
\end{definition}



\paragraph{Remark 1: Existence of optimal solutions} In general, an optimal decomposition need not exist as a finitely representable predicate: one can have a strictly improving sequence of feasible decompositions with no strongest input filter (or weakest residual) expressible in the language.
Accordingly, throughout the paper, optimality is understood relative to fixed finite spaces of candidate predicates, within which an optimal decomposition always exists.

\paragraph{Remark 2: Impact on performance}  While optimality is defined in terms of logical strength and weakness, these properties have clear practical consequences in terms of performance. 
A stronger pre-computation filter~$Q$ excludes more rows upfront,
directly reducing UDF work. Minimizing the residual predicate is more nuanced: logical weakness does not, by itself, imply cheaper evaluation. In our setting, however, we restrict $P'$ to mostly be a conjunction of a subset of the conjuncts in $P$, so minimizing $P'$ mainly amounts to dropping conjuncts already enforced by $Q$ and thus avoiding redundant post-computation checks.


\begin{figure}[t]
\centering
\resizebox{\columnwidth}{!}{%
\begin{tikzpicture}[
  stage/.style={
    draw=blue!70!black, thick, dashed, rounded corners=6pt,
    inner sep=5pt, align=center, minimum height=1cm
  },
  stepnum/.style={
    circle, fill=blue!70!black, text=white, font=\bfseries\footnotesize,
    inner sep=1.5pt, minimum size=10pt
  },
  arr/.style={-{Stealth[length=4pt]}, thick, blue!70!black}
]
\node[stage] (A) {\underline{\smash[b]{Original pipeline}}\\[1pt]$P_0 \circ \Pi_{\mathsf{post}} \circ F \circ \Pi_{\mathsf{pre}}$};
\node[stage, right=0.7cm of A] (B) {\underline{\smash[b]{Before UDF pushdown}}\\[1pt]$\Pi_{\mathsf{post}} \circ P \circ F \circ \Pi_{\mathsf{pre}}$};
\node[stage, right=0.7cm of B] (C) {\underline{\smash[b]{After UDF pushdown}}\\[1pt]$\Pi_{\mathsf{post}} \circ P' \circ F \circ Q \circ \Pi_{\mathsf{pre}}$};
\node[stage, right=0.7cm of C] (D) {\underline{\smash[b]{Final pipeline}}\\[1pt]$\Pi_{\mathsf{post}} \circ P' \circ F \circ \Pi_{\mathsf{pre}} \circ Q'$};
\draw[arr] (A.east) -- (B.west);
\draw[arr] (B.east) -- (C.west);
\draw[arr] (C.east) -- (D.west);
\node[stepnum, yshift=9pt] at ($(A.east)!0.5!(B.west)$) {1};
\node[stepnum, yshift=9pt] at ($(B.east)!0.5!(C.west)$) {2};
\node[stepnum, yshift=9pt] at ($(C.east)!0.5!(D.west)$) {3};
\end{tikzpicture}%
}
\caption{Illustration of how our problem formulation fits into end-to-end pipeline optimization. Steps 1 and 3 can be performed by existing relational optimizers, whereas step 2 is enabled by our formulation.}\label{fig:pipeline}
\end{figure}

\paragraph{Remark 3: Using UDF pushdown for pipeline optimization.} UDFs typically sit inside larger pipelines of joins, projections, and filters. Figure~\ref{fig:pipeline} shows how our generalized pushdown formulation composes with standard relational rewrites: starting from $P_0 \circ \Pi_{\mathsf{post}} \circ F \circ \Pi_{\mathsf{pre}}$, a host optimizer first pushes $P_0$ left through $\Pi_{\mathsf{post}}$ using UDF-agnostic rules~\cite{ullman90}, yielding a predicate $P$ at the UDF boundary (``Before UDF pushdown''). Our formulation then rewrites $P \circ F$ as $P' \circ F \circ Q$ with $(Q,P')$ satisfying Equation~\eqref{eq:unified-lifted} (``After UDF pushdown''). The UDF no longer blocks subsequent rewrites: $Q$ and any additional predicates can then be pushed right through $\Pi_{\mathsf{pre}}$ as $Q'$ (``Final pipeline'').


%% file: sections/methodology.tex
\section{Verification Methodology} \label{sec:methodology}


In this section, we present a verification method for checking the correctness of a candidate pushdown optimization.  
The central challenge in verifying pushdown transformations is that the optimized and original executions no longer process the same sequence of inputs: the pushed-down version skips rows filtered by~$Q$, while the original processes them all. As a result, their internal states may differ during execution.
To reason about such executions, we introduce  a \emph{bisimulation invariant}~$\bisim(a_1, a_2)$ that relates the internal state~$a_1$ of the original computation to the internal state~$a_2$ of the optimized one.
Intuitively, $\bisim$ captures a relational correspondence between these two internal states: it allows temporary differences caused by skipped rows, but requires that the executions eventually agree on the same final output.   

Figure~\ref{fig:verification-conditions} formalizes this reasoning principle through four verification conditions (VCs) that together ensure the soundness of a pushdown transformation: \emph{initialization}, \emph{synchronized-step preservation}, \emph{stutter-step preservation}, and \emph{final-state agreement}.  
These conditions correspond to the base case, inductive steps, and postcondition of an inductive proof, and they enable local reasoning about symbolic accumulator states without explicitly tracking full input traces. We explain these VCs in more detail below.

\begin{figure}[t]
\centering
\begin{mdframed}[linecolor=black, backgroundcolor=gray!8, roundcorner=5pt, linewidth=0.5pt, innertopmargin=1em, innerbottommargin=1em, innerleftmargin=1em, innerrightmargin=1em]

\textbf{Verification Conditions for Generalized Pushdown.} Let $\bisim$ denote the bisimulation invariant relating the internal states of the original and optimized executions. The pushdown pair $(Q, P')$ is sound if the following hold for all $a_1, a_2 \in A$ and $r \in R$:

\medskip

\noindent
\textbf{Initialization (\textsf{Init}):} \quad 
$ 
\bisim(I, I)
$

\noindent
\textbf{Synchronized-step preservation (\textsf{Sync}):}
\[
\bisim(a_1, a_2) \land Q(r) \land a_1' = f(a_1, r) \land a_2' = f(a_2, r)
\;\Rightarrow\;
\bisim(a_1', a_2')
\]

\noindent
\textbf{Stutter-step preservation (\textsf{Stutter}):}
\[
\bisim(a_1, a_2) \land \lnot Q(r) \land a_1' = f(a_1, r) \land a_2' = a_2
\;\Rightarrow\;
\bisim(a_1', a_2')
\]

\noindent
\textbf{Final-state agreement (\textsf{Final}):}
\[
\bisim(a_1, a_2)
\;\Rightarrow\;
\big(P(a_1) \land P'(a_2) \land a_1 = a_2\big)
\lor
\big(\lnot P(a_1) \land \lnot P'(a_2)\big)
\]

\end{mdframed}
\captionsetup{skip=2pt}
\caption{VCs for establishing  that $(Q, P')$ is a sound pushdown of $P$ through  UDF $F = \mathsf{fold}(x, I, f)$.}
\label{fig:verification-conditions}
\end{figure}


\vspace{-0.3em}
\paragraph{Initialization (\textsf{Init}).} The first {VC} asserts that $\bisim$ holds at the beginning of both executions, when the state is initialized to $I$. This serves as the base case for the inductive proof.

\vspace{-0.3em}
\paragraph{Synchronized-step preservation (\textsf{Sync}).} This condition handles the case where both the original and optimized executions process the same input row $r$ (i.e., when $Q(r)$ holds). In this case, both accumulators advance via the same accumulator $f$, and we must show that their updated internal states remain related by $\bisim$. We refer to this verification condition as ``synchronized-step preservation'' (or \sync for short) because the two executions move in lockstep on this row.

\vspace{-0.3em}
\paragraph{Stutter-step preservation (\textsf{Stutter}).} Our third condition addresses the case where $Q(r)$ is false: the optimized execution skips row $r$, while the original  processes it. Here, the state in the original execution is updated via $f$, while the  state in the optimized version remains unchanged. We must show that $\bisim$ continues to hold despite this asymmetry. The name ``stutter step'' (\stutter for short) reflects the fact that the optimized execution effectively ``stutters'' while the original moves forward.

\vspace{-0.3em}
\paragraph{Final-state agreement (\textsf{Final}).} 
The final condition ensures that the bisimulation invariant $\bisim$ is strong enough to imply agreement between the two executions at the end. Specifically, \emph{agreement} here means either (1) both executions satisfy their respective predicates, and their internal states are identical (i.e., $P(a_1) \land P'(a_2) \land a_1 = a_2$), or (2) both executions fail their respective predicates.

 \begin{definition}{\bf (Certification Witness)}
\label{def:cert-witness}
A predicate $\psi$ is a \emph{certification witness} for a candidate pushdown pair $(Q, P')$ if $\psi$ satisfies the verification conditions in Figure~\ref{fig:verification-conditions}.
\end{definition}

\begin{example}
We now illustrate our verification methodology in the context of the aggregation example from Section~\ref{sec:background}, where the \texttt{top2} UDF computes the two highest scores for each team. The original filter $P$ requires that both the top and second-highest scores exceed 90.0:
\[
P(a) = (a[0] > 90.0 \wedge a[1] > 90.0)
\]
As discussed in Section~\ref{sec:problem}, the optimal pushdown
prunes the input
and preserves behavior using 
\[
Q(r) = (r[`\mathsf{score}\textrm'] > 90.0)
\quad \text{and} \quad
P'(a) = (a[1] \neq -\infty).
\]

To verify the correctness of this pushdown optimization, we need the following bisimulation invariant
$\psi$,
where we use color to distinguish internal states from the original and optimized executions:
\orig{blue} refers to values computed over the full input (no filtering), and 
\opt{red} refers to values computed after filtering by $Q$:
\vspace{-0.8em}

\[
\begin{aligned}
 &\underbrace{
   (\orig{a_1[0]} > 90.0 \Rightarrow \orig{a_1[0]} = \opt{a_2[0]})
   \wedge
   (\opt{a_2[0]} > 90.0 \Rightarrow \orig{a_1[0]} = \opt{a_2[0]})
 }_{\text{(1) top-1 in sync if $>90.0$ in either}}\\
 \mathllap{\wedge}\;&\underbrace{
   (\orig{a_1[1]} > 90.0 \Rightarrow \orig{a_1[1]} = \opt{a_2[1]})
   \wedge
   (\opt{a_2[1]} > 90.0 \Rightarrow \orig{a_1[1]} = \opt{a_2[1]})
 }_{\text{(1) top-2 in sync if $>90.0$ in either}}\\
 \mathllap{\wedge}\;&\underbrace{(\opt{a_2[0]} \neq -\infty \Rightarrow \opt{a_2[0]} > 90.0)}_{\text{(2) optimized top-1 undefined or $>90$}}
 \wedge
 \underbrace{(\opt{a_2[1]} \neq -\infty \Rightarrow \opt{a_2[1]} > 90.0)}_{\text{(2) optimized top-2  undefined or $>90.0$}}
 \wedge
 \underbrace{(\opt{a_2[1]} > 90.0 \Rightarrow \opt{a_2[0]} > 90.0)}_{\text{(3) top-2 $>90.0$ implies top-1 $>90.0$}}
\end{aligned}
\]

Intuitively, this invariant states that (1) whenever one of the two highest scores exceeds 90.0, it stays in sync between both executions (first four conjuncts); 
(2) throughout the optimized execution, each score is either undefined or exceeds 90.0  (next two conjuncts), and 
(3)  if the second-highest score exceeds 90.0, the top score  must too (last conjunct). Crucially, all of these conjuncts are necessary in order to prove the correctness of the pushdown optimization. Hence, as this example illustrates, bisimulation invariants in this setting can be quite complex, allowing for partial disagreement in internal states while still guaranteeing equivalence of the final outcomes.

\end{example}

We conclude this section by noting that our verification methodology is sound and relatively complete in the standard sense of Hoare-style logics~\cite{cook78}: whenever a candidate pushdown pair $(Q, P')$ is correct and the underlying logic is sufficiently expressive, there exists a bisimulation invariant capable of establishing its correctness.

\begin{theorem}[Soundness and Relative Completeness]\label{thm:sound-complete}
 Let $(Q, P')$ be a candidate solution to the generalized pushdown problem
 for computation $F$ and predicate $P$. Then:
\begin{enumerate}[leftmargin=*]
  \item \textbf{Soundness:} If there exists a bisimulation invariant
  $\psi$ that is a certification witness for $(Q,P')$ by Definition~\ref{def:cert-witness}, then $(Q, P')$ is a sound optimization, i.e.,
$ 
  \forall x.\,\Lift(P, F(x)) = \Lift(P', F(\filter_Q(x))).
$

  \item \textbf{Relative completeness:} If the above equivalence holds and the underlying logic is sufficiently expressive, then  there exists an invariant $\psi$ that is a certification witness for $(Q,P')$.
\end{enumerate}
\end{theorem}

%% file: sections/synthesis.tex
\section{Synthesizing Optimal Pushdown Transformations}\label{sec:synthesis}

Given a
computation $F$ and a post-computation predicate $P$,
our goal is to synthesize a triple
$(Q, P', \psi)$
such that (1)
$Q$, $P'$, and $\psi$
satisfy the verification conditions from Figure~\ref{fig:verification-conditions}, and (2) the solution is optimal according to Definition~\ref{def:optimal-pushdown}, meaning $Q$ is as strong as possible, and $P'$ is as weak as possible. If we can find such a triple, then $(Q, P')$ can be used to obtain the best pushdown optimization for the original pipeline, while $\psi$ serves as a \emph{certificate of soundness} of the optimization. Hence, our synthesis method aims to simultaneously identify the best pushdown optimization possible, while also constructing its proof of correctness.

To motivate our solution, we first discuss the three main challenges in computing
$Q$, $P'$, and $\psi$.
First, even ignoring the optimality requirement, this problem cannot be solved using off-the-shelf CHC solvers (as commonly done in program verification and predicate synthesis), because the  VCs in Figure~\ref{fig:verification-conditions} are  \emph{non-Horn} due to the \emph{negated} occurrence of $Q$ in
\stutter. Second, even if we manually supply the ground-truth solution for $Q$ and thereby make the  resulting VCs Horn, we still find that existing solvers struggle to infer  $\bisim$ and $P'$. Finally, the optimality requirement makes an already-difficult inference problem even harder. In this context, $Q$ and $P'$ need to be optimized in \emph{different} directions, creating a dual-criteria search in which we aim to find the strongest admissible $Q$ and the weakest admissible $P'$. 

Our synthesis method directly addresses these challenges through three core innovations.
First, we introduce a \emph{structured decomposition} of the search space that breaks down the joint inference of
$Q$, $P'$, and $\psi$
into a layered pipeline. By enumerating candidate input filters $Q$ in decreasing order of strength, computing the strongest admissible bisimulation $\bisim$ for each, and then deriving the weakest consistent residual $P'$, our method transforms a tangled three-way dependency into a tractable sequence of subproblems
while still
guaranteeing optimality. 
Second, we develop \emph{symbolic bounds on bisimulation invariants} that act as \emph{unrealizability certificates}. These bounds allow the synthesizer to quickly rule out entire families of pushdown candidates. 
Finally, our algorithm incorporates \emph{root-cause-guided repair}, which analyzes the precise reason a verification condition fails  and selectively refines the responsible artifact (i.e., $Q$, $\psi$, or $P'$).




\subsection{Top-Level Algorithm}\label{sec:top-level}

Algorithm~\ref{alg:synthesize} shows our top-level synthesis procedure, {\sc SynthesizeOptimalPushdown}, which seeks optimal instantiations of $Q$ and~$P'$ expressible as conjunctions of atoms drawn from their respective predicate universes $U_Q$ and $U_{P'}$.
We deliberately confine the search to a fixed predicate universe for several reasons. First, real-world UDFs often manipulate complex internal state whose SMT encodings rely on algebraic data types, making automated refinement techniques such as Craig interpolation impractical in our setting.
Moreover, static analysis of the UDF and its surrounding pipeline allows us to \emph{a priori} extract the building blocks from which $Q$, $P'$, and~$\bisim$ can be formed. This includes not only simple literals but also richer predicates such as implications derived from control-flow structure.
Finally, bounding
synthesis to a finite universe yields a well-structured search space that can be systematically explored in decreasing order of logical strength for~$Q$ and increasing order for~$P'$, ensuring tractable search and clear optimality guarantees.

Given these predicate universes, Algorithm~\ref{alg:synthesize} starts by initializing a worklist with the strongest possible input-side filter (the conjunction of all atoms in~$U_Q$) and explores progressively weaker candidates through the main loop (lines 3--11). In each iteration, it dequeues a candidate input-side filter $Q_0$ from the worklist 
and subjects it to a \emph{feasibility check} by invoking \textsc{WeakenViaBounds} on line~4.  If $Q_0$ is deemed infeasible (meaning that no sound pushdown optimization can exist with $Q_0$ as the input-side filter), \textsc{WeakenViaBounds} returns a repaired (weaker) version of $Q_0$, along with symbolic bounds  $(\bisim_{\min}, \bisim_{\max})$   that summarize the feasible region for the target bisimulation invariant.
These bounds act as under- and over-approximations for the true invariant~$\bisim$: That is, $\bisim_{\min}$ collects atoms that \emph{must} appear in any valid invariant, while $\bisim_{\max}$ gathers those that \emph{may}.

\begin{algorithm}[h]
\caption{\textsc{SynthesizeOptimalPushdown}}
\label{alg:synthesize}
\begin{algorithmic}[1]
\Statex \textbf{Input:} UDF $F = \mathsf{fold}(x, I, f)$, post-UDF filter $P$, predicate universes $U_Q, U_{P'}, U_\bisim$
\Statex \textbf{Output:} Optimal pushdown pair $(Q, P')$
\Function{SynthesizeOptimalPushdown}{$F, P, U_Q, U_{P'}, U_\bisim$}
    \State $W \gets \{U_Q\}$ \Comment{Start from strongest filter}
    \While{$\textsc{Dequeue}(W) = \mathsf{Some}(Q_0)$}
        \State $(\mathsf{ok}, Q, \bisim_{\min}, \bisim_{\max}) \gets \textsc{WeakenViaBounds}(Q_0, P, U_\bisim)$
        \If{$\neg \mathsf{ok}$} \textbf{continue}
        \EndIf
        \State $(\mathsf{ok}, \bisim, \mathsf{diagnosis}) \gets \textsc{FindStrongestBisimulation}(Q, \bisim_{\min}, \bisim_{\max})$
        \If{$\neg \mathsf{ok}$}
            $Q' \gets $ \Call{Repair}{$Q$, \textsf{diagnosis}};\ \textsf{Enqueue}$(W, Q')$
        \Else
            \State $(\mathsf{ok}, P') \gets \textsc{FindResidual}(\bisim, U_{P'}, P)$
            \If{$\mathsf{ok}$} \Return $(Q, P')$ \EndIf
        \EndIf
        \ForAll{$q \in Q$}  \textsf{Enqueue}$(W, Q \setminus \{q\})$ \EndFor
    \EndWhile
    \State \Return \textsf{Fail}
\EndFunction
\Function{Repair}{$Q, \textsf{diagnosis}$}
    \If{$\textsf{diagnosis} = \textsf{Sync}(r)$}
        \Return $ \{\, u \in Q \mid \neg u(r) \,\}$ \Comment{Force $Q(r)=\mathsf{false}$}
    \ElsIf{$\textsf{diagnosis} = \textsf{Stutter}(r)$}
        \Return $\{\, u \in Q \mid u(r) \,\}$ \Comment{Force $Q(r)=\mathsf{true}$}
    \EndIf
\EndFunction
\end{algorithmic}
\end{algorithm}

Once a repaired filter $Q$ and symbolic bounds  are obtained, the algorithm attempts to construct the strongest invariant $\bisim$ that is consistent with the bounds  $(\bisim_{\min}, \bisim_{\max})$ by calling {\sc FindStrongestBisimulation} (line~6). However, since the unrealizability check performed in {\sc WeakenViaBounds} is sound but incomplete, $Q$ \emph{could} still be infeasible, causing {\sc FindStrongestBisimulation} to fail. In this case, {\sc FindStrongestBisimulation} also returns a \emph{diagnosis} that identifies the specific verification condition violated by $Q$. This diagnosis takes one of two forms: 

\begin{itemize}[leftmargin=*]
  \item \textbf{\textsf{Sync}$(r)$:} indicates that the \emph{\sync} VC failed on row $r$. Intuitively, this means that if $Q$ classifies $r$ as selected (i.e., $Q(r)=\mathsf{true}$), then applying the accumulator function $f$ to both executions would produce internal states that no bisimulation within $(\bisim_{\min}, \bisim_{\max})$ can reconcile. The only sound repair is therefore to \emph{exclude} $r$ by enforcing $Q(r)=\mathsf{false}$, implemented by removing all atoms in $Q$ that evaluate to \textsf{true} on $r$.
  
  \item \textbf{\textsf{Stutter}$(r)$:} indicates that the \emph{\stutter} VC failed on row $r$. In this case, $Q$ classifies $r$ as unselected (i.e., $Q(r)=\mathsf{false}$), but skipping $r$ breaks agreement between the two internal states. Because no invariant can restore equivalence after this omission, the repair is to \emph{force inclusion} of $r$ by enforcing $Q(r)=\mathsf{true}$, i.e., by removing all atoms in $Q$ that evaluate to \textsf{false} on $r$.
\end{itemize}

The algorithm then backtracks by enqueuing the repaired filter $Q'$ along with all one-step weakenings of $Q$, each obtained by dropping a single atom.
Otherwise, a valid invariant $\bisim$ is successfully synthesized, and the algorithm proceeds to compute the weakest residual predicate $P'$ by invoking {\sc FindResidual} on line~9. If {\sc FindResidual} succeeds, the algorithm terminates and returns $(Q,P')$; if not, it enqueues all one-step weakenings of $Q$ and backtracks.

\paragraph{Discussion.} A key point about this algorithm is that it can often establish the \emph{infeasibility} of a candidate $Q$ without attempting to synthesize either $\bisim$ or $P'$. Because  each candidate $Q$ could require exploring a vast  space of bisimulation invariants and corresponding residuals, ruling out even a \emph{single} candidate $Q$ can be very beneficial.  Furthermore, this pruning effect is amplified by the repair mechanism (via \textsc{WeakenViaBounds}), which implicitly removes from the search space \emph{many} weakenings of $Q$ that are infeasible for the same reason as $Q$. In other words, the infeasibility of a single $Q$ can propagate to entire families of related filters, significantly reducing the search space (as shown empirically in Section~\ref{sec:eval}).

\begin{theorem}[Guarantees]
\label{thm:optimal-synth-guarantees}
Fix finite predicate universes $U_Q$, $U_{P'}$, and $U_\psi$.
For any computation $F$ and post-filter $P$, \textsc{SynthesizeOptimalPushdown}$(F,P,U_Q,U_{P'},U_\psi)$ satisfies the following properties:

\begin{enumerate}[leftmargin=*]
\item \textbf{Termination.} The procedure terminates.

\item \textbf{Soundness.} If it returns $(Q^*,P^*)$ where $Q^* \subseteq U_Q$ and $P^* \subseteq U_{P'}$, then $(Q^*,P^*)$ satisfies Equation~\eqref{eq:def-pushdown}.

\item \textbf{Optimality within the search space.}
Let $(Q^*,P^*)$ be the pair returned by the procedure. Then:
\begin{enumerate}[leftmargin=*,label=(\alph*)]
\item (\emph{Maximal $Q$}) For any $Q \subseteq U_Q$, $P' \subseteq U_{P'}$, and $\psi \subseteq U_\psi$ such that $\psi$ is a certification witness for $(Q,P')$, $Q$ is not strictly stronger than $Q^*$.
\item (\emph{Minimal $P'$ given $Q^*$}) For any $P' \subseteq U_{P'}$ and $\psi \subseteq U_\psi$ such that $\psi$ is a certification witness for $(Q^*,P')$, $P'$ is not strictly weaker than $P^*$.
\end{enumerate}
\end{enumerate}
\end{theorem}

Please note that Theorem~\ref{thm:optimal-synth-guarantees} provides optimality only relative to the fixed universes $U_Q$, $U_{P'}$, and $U_\psi$.
Accordingly, the procedure may miss valid pushdowns in two ways: (i) there may exist a valid decomposition $(Q_0,P_0)$ with $Q_0 \not\subseteq U_Q$ or $P_0 \not\subseteq U_{P'}$ that is strictly better (stronger $Q_0$ or weaker $P_0$); or (ii) a valid pair $(Q,P')$ with $Q \subseteq U_Q$ and $P' \subseteq U_{P'}$ may require a bisimulation invariant outside $U_\psi$.
However, fixing these finite predicate universes is necessary for termination: Without restricting $U_Q$ and $U_{P'}$, an optimum need not be finitely representable within the predicate language, and without restricting $U_\psi$, searching for an invariant witness does not admit termination guarantees in general.  Additionally, solver timeouts or \textsf{unknown} answers may prevent, in practice, the procedure from finding an optimal solution with the fixed predicate universes.

\subsection{Unrealizability Proofs via Symbolic Bounds on Bisimulation Invariants}

Given a candidate  filter $Q$, Algorithm~\ref{alg:weaken-via-bounds} (\textsc{WeakenViaBounds}) is the engine that decides whether our algorithm should (a) keep $Q$ as is, (b) \emph{repair} it to a weaker predicate, or (c) abandon it completely. To make this decision,
Algorithm~\ref{alg:weaken-via-bounds}
maintains a pair of symbolic bounds $(\bisim_{\min},\bisim_{\max})$ over the invariant universe $U_{\bisim}$ that serve as symbolic lower and upper bounds on the bisimulation invariant. 

\begin{algorithm}[t]
\caption{\textsc{WeakenViaBounds}}
\label{alg:weaken-via-bounds}
\begin{algorithmic}[1]
\Statex \textbf{Input:} Candidate input filter $Q$, post-UDF filter $P$, invariant universe $U_\bisim$
\Statex \textbf{Output:} $(\mathsf{ok}, Q', \bisim_{\min}', \bisim_{\max}')$ where $\mathsf{ok}\in\{\mathsf{true},\mathsf{false}\}$, potentially weakened filter $Q'$, and potentially tightened symbolic bounds $(\bisim_{\min}', \bisim_{\max}')$
\Function{WeakenViaBounds}{$Q, P, U_\bisim$}
    \State $(\bisim_{\min}, \bisim_{\max}) \gets (\varnothing, U_\bisim)$
    \While{\textsf{true}}
        \If{$Q=\varnothing$} \Return $(\mathsf{false}, -, -, -)$ \EndIf
        \State $(\bisim_{\min}, \bisim_{\max}) \gets \Call{RefineBounds}{Q, P, U_\bisim, \bisim_{\min}, \bisim_{\max}}$
        \State ($\mathsf{res}, \mathsf{diagnosis}) \gets \Call{CheckUnrealizable}{Q, \bisim_{\min}, \bisim_{\max}}$
        \If{$\neg \mathsf{res}$}
            \Return $(\mathsf{true}, Q, \bisim_{\min}, \bisim_{\max})$
            \textbf{else} $Q \gets \Call{Repair}{Q, \mathsf{diagnosis}}$
        \EndIf
    \EndWhile
\EndFunction
\end{algorithmic}
\end{algorithm}


\textsc{WeakenViaBounds} initializes $(\bisim_{\min},\bisim_{\max})$ to the trivial bounds $(\varnothing, U_{\bisim}) $ on line 2 and then gradually weakens the candidate filter $Q$ until a termination condition is reached. In each iteration, the procedure alternates between \emph{tightening} the symbolic bounds via \textsc{RefineBounds} (line 5) and querying \textsc{CheckUnrealizable} for an \emph{unrealizability certificate} given the current bounds (line 6). These two procedures play complementary roles. {\sc RefineBounds} addresses the question: \emph{``Given the current $Q$, what is the tightest bound we can infer on the bisimulation invariant?''} On the other hand, {\sc CheckUnrealizable} addresses the question: \emph{``Given the current bounds $(\bisim_{\min}, \bisim_{\max})$, can we prove that no invariant $\bisim$ in this range will suffice to prove the correctness of $Q$?''} Together with the repair mechanism on line 7 (\textsf{false} branch), these two procedures form a virtuous cycle in the following way: First,  the refined bounds may allow us to prove the unrealizability of a filter we previously could not prove to be infeasible. Second, once we detect infeasibility of $Q$, we can weaken it through targeted repair (via the same {\sc Repair} procedure from earlier), which in turn allows us to further refine the bounds on $\bisim$. This alternation between bound tightening and unrealizability checking continues until either (a) $Q$ becomes $\varnothing$ and is deemed hopeless (line 4) or (b) $Q$ appears to be consistent with the bounds (\textsf{true} branch on line 7).  In the remainder of this subsection, we discuss the
sub-procedures
{\sc RefineBounds} and {\sc CheckUnrealizable}.

\begin{algorithm}[h]
\caption{\textsc{RefineBounds}}
\label{alg:refine-bounds}
\begin{algorithmic}[1]
\Statex \textbf{Input:} Candidate filter $Q$, post-filter $P$, invariant universe $U_{\bisim}$, symbolic bounds $(\bisim_{\min}, \bisim_{\max})$
\Statex \textbf{Output:} Refined bounds $(\bisim_{\min}, \bisim_{\max})$
\Function{RefineBounds}{$Q, P, U_{\bisim}, \bisim_{\min}, \bisim_{\max}$}

  \If{$\bisim_{\min} = \varnothing$}   \Comment{Lower bound (computed  once)}
    \State $\phi \gets (P(a_1) \land P(a_2) \land a_1 = a_2) \;\lor\; (\neg P(a_1) \land \neg P(a_2))$
    \State $\bisim_{\min} \gets \textsc{FindWeakestImplicant}(U_{\bisim}, \phi)$
  \EndIf

  \ForAll{$p \in \bisim_{\max} \setminus \bisim_{\min}$} \Comment{Upper bound at base case (computed once)}
    \If{$\neg p(I, I)$}  $\bisim_{\max} \gets \bisim_{\max} \setminus \{p\}$
    \EndIf 
  \EndFor

  \ForAll{$p \in \bisim_{\max} \setminus \bisim_{\min}$} \Comment{Upper bound based on current $Q$}
    \State $\varphi \gets \bisim_{\max}(a_1, a_2) \land Q(r) \land a_1' = f(a_1, r) \land a_2' = f(a_2, r)$
    \If{$\varphi \;\not\Rightarrow\; p(a_1', a_2')$}  $\bisim_{\max} \gets \bisim_{\max} \setminus \{p\}$
    \EndIf
  \EndFor

  \State \Return $(\bisim_{\min}, \bisim_{\max})$
\EndFunction
\end{algorithmic}
\end{algorithm}

\paragraph{Computing symbolic bounds on the bisimulation invariant.}
Algorithm~\ref{alg:refine-bounds} (\textsc{RefineBounds}) maintains the symbolic bounds $(\bisim_{\min},\bisim_{\max})$ within which any valid bisimulation must lie.
The lower bound $\bisim_{\min}$ is derived from \final and is  computed only  \emph{once}, as this VC is independent of  $Q$. 
One subtlety  here is that \final also depends on $P'$, which is completely unknown at this stage. 
To sidestep this dependency, \textsc{RefineBounds} instantiates
occurrences of the unknown residual $P’$ in the VC with the known filter $P$, yielding the formula $\phi$ on line 3 of Algorithm~\ref{alg:refine-bounds}. Intuitively, $P$ represents the strongest possible residual; thus,  if the input-side filter $Q$ is indeed valid, then applying $Q$ in conjunction with the original filter $P$ should already yield a behaviorally equivalent result. 
Hence, substituting $P'$ with $P$ in \final can only \emph{weaken} the requirement and therefore produces a conservative approximation for deriving a lower bound on $\bisim$. In particular, any implicant of   formula $\phi$ from line 3 serves as a valid lower bound, but we compute the \emph{weakest prime implicant}~\cite{prime} to ensure that the bound is as tight as possible.\footnote{See Appendix~\ref{sec:proofs} for the definition of \textsc{FindWeakestImplicant} and its correctness proofs.}

In contrast, the upper bound $\bisim_{\max}$ is refined incrementally as $Q$ evolves because its computation involves \sync, which is dependent on $Q$. 
At initialization, $\bisim_{\max}$ contains all atoms in the universe $U_\bisim$.
Each call to \textsc{RefineBounds} prunes this set in two ways. 
First, the algorithm removes any atom that fails \init.\footnote{This initialization check $p(I,I)$ only needs to be performed on the first invocation of \textsc{RefineBounds}. } Then, it checks whether each remaining atom $p$ is preserved by \sync  under the current $\bisim_{\max}$. Specifically, it instantiates \sync with the full upper bound in the antecedent and the candidate atom $p$ alone in the consequent:
\[
\bisim_{\max}(a_1, a_2) \land Q(r) \land a_1' = f(a_1, r) \land a_2' = f(a_2, r)
\;\Rightarrow\;
p(a_1', a_2').
\]
If this implication fails, the atom $p$ is removed from $\bisim_{\max}$.

It is worth noting that the {\sc RefineBounds} procedure does \emph{not} use \stutter in refining the upper bound in order to ensure soundness.
Specifically, recall that \stutter contains the negation of $Q$ in its antecedent. Since we start from the strongest possible  $Q$ and gradually weaken it until it is sound, the current candidate $Q$ could be logically stronger than the sound solution $Q^*$. This, in turn, means that the antecedent of \stutter could be weaker than when it is instantiated with the ground truth solution $Q^*$. Hence, it would be unsound to  drop predicates that are not implied by the current antecedent. Due to this asymmetry, {\sc RefineBounds} only uses \init and \sync.

To illustrate, consider the \texttt{top2} running example from Section~\ref{sec:background} with $I = (-\infty, -\infty)$.
The \init pass evaluates each atom on $(I,I)$: for instance,
$a_1[0] = -\infty \Rightarrow \mathsf{false}$ (``the top score is always defined'') fails and is discarded, while
$a_1[0] = -\infty \Rightarrow a_1[0] = a_2[0]$ (``if undefined in the original, then also in the optimized'') evaluates to $-\infty = -\infty$ and survives.
The \sync pass then removes atoms incompatible with $Q(r) = (r[\text{`}\mathsf{score}\textrm'] > 90.0)$: the atom
$(a_1[0] > 90.0 \Rightarrow \mathsf{false})$ survives \init but fails \sync, since processing a qualifying row forces the top score above~90.0.


\paragraph{Checking {unrealizability} via symbolic bounds.} 
Algorithm~\ref{alg:qfeasible-approx} (\textsc{CheckUnrealizable}) uses the bounds $(\bisim_{\min},\bisim_{\max})$ to decide whether the current filter $Q$ could  admit a sufficiently  strong bisimulation. 
Rather than synthesizing an explicit invariant, it asks a simpler question: is \emph{any} invariant lying between these bounds compatible {with} the one-step proof obligations? 
To answer this question, the procedure instantiates the two key verification conditions (namely, \sync and \stutter) using the bounds in their most conservative form: $\bisim_{\max}$ is placed in the antecedent (the strongest possible assumption) and $\bisim_{\min}$ in the consequent (the weakest admissible conclusion). 
If \emph{even this} implication fails, then no invariant consistent with the current candidate $Q$ could possibly satisfy the condition.

\begin{algorithm}[t]
\caption{\textsc{CheckUnrealizable}}
\label{alg:qfeasible-approx}
\begin{algorithmic}[1]
\Statex \textbf{Input:} Candidate input filter $Q$, symbolic bounds $(\bisim_{\min}, \bisim_{\max})$
\Statex \textbf{Output:} Unrealizability judgment and root cause diagnosis: $\mathsf{Sync}(r)$ or $\mathsf{Stutter}(r)$
\Function{CheckUnrealizable}{$Q, \bisim_{\min}, \bisim_{\max}$}
    \State $\chi_{\mathsf{sync}} \gets \big ( (\bisim_{\max}(a_1, a_2) \land Q(r) \land a_1' = f(a_1, r) \land a_2' = f(a_2, r)) \;\Rightarrow\; \bisim_{\min}(a_1', a_2') \big )$ 
    
    \If{\textsf{Invalid}($\chi_{\mathsf{sync}})$}
        \State $(a_1, a_2, r) \gets \textsf{Model}(\neg \chi_{\mathsf{sync}})$
        \State \Return $(\mathsf{true}, \mathsf{Sync}(r))$
    \EndIf
    \State $\chi_{\mathsf{stutter}} \gets \big ( (\bisim_{\max}(a_1, a_2) \land \neg Q(r) \land a_1' = f(a_1, r) \land a_2' = a_2) \;\Rightarrow\; \bisim_{\min}(a_1', a_2') \big ) $
    \If{\textsf{Invalid}($\chi_{\mathsf{stutter}})$}
        \State $(a_1, a_2, r) \gets \textsf{Model}(\neg \chi_{\mathsf{stutter}})$
        \State \Return $(\mathsf{true}, \mathsf{Stutter}(r))$
    \EndIf
    \State \Return $(\mathsf{false}, -)$
\EndFunction
\end{algorithmic}
\end{algorithm}

If \sync fails, the procedure returns a witness $\mathsf{Sync}(r)$, where $r$ is a symbolic row that demonstrates the violation. 
Intuitively, this means that if $Q$ were to classify $r$ as selected, the two internal states would be forced into a disagreement that no invariant between $(\bisim_{\min},\bisim_{\max})$ could reconcile. 
The only possible repair is therefore to exclude $r$, i.e., to enforce $Q(r)=\mathsf{false}$. 
Dually, if \stutter fails, the procedure returns $\mathsf{Stutter}(r)$ with a witness $r$. 
In this case, the violation arises because skipping $r$ (i.e., $Q(r)=\mathsf{false}$) would break agreement, and no admissible invariant could compensate. 
In this case, the only repair is to retain $r$ by  enforcing $Q(r)=\mathsf{true}$. 
If both checks succeed, the procedure returns $(\mathsf{false}, -)$, meaning that there is at least one invariant consistent with the current bounds that \emph{could}
establish the correctness of $Q$.
Importantly, {this} is not a guarantee that such an invariant can actually be synthesized but merely that the current filter $Q$ has not yet been refuted by the bounds. 
Thus, \textsc{CheckUnrealizable} acts as a sound but incomplete filter.

Returning to the \texttt{top2} example, suppose $U_Q$ additionally contains $r[\text{`}\mathsf{score}\textrm'] > 95.0$.
The \stutter check finds a unrealizability witness with $a_1 = a_2 = (97.0, -\infty)$ and $r[\text{`}\mathsf{score}\textrm'] = 92.0$: the optimized execution skips $r$ (since $92.0 < 95.0$), but the original processes it, producing $a_1' = (97.0, 92.0)$ with $P(a_1') = \mathsf{true}$ while $P(a_2') = P(a_2) = \mathsf{false}$, violating $\bisim_{\min}$.
The repair forces $Q(r) = \mathsf{true}$ by dropping $r[\text{`}\mathsf{score}\textrm'] > 95.0$ (false on the witness), weakening $Q$ to $r[\text{`}\mathsf{score}\textrm'] > 90.0$.

\subsection{Synthesizing Bisimulation Invariants and Residuals}
Given a candidate  $Q$, recall that the top-level synthesis algorithm must construct two remaining artifacts: first, the strongest bisimulation invariant $\bisim$ consistent with $Q$, and \emph{then} the weakest residual predicate $P'$. This ordering preserves optimality: for any given choice of $P’$, if we cannot prove \final using the strongest possible bisimulation invariant, we certainly also cannot prove it using a weaker one. Thus, our algorithm first infers the strongest possible bisimulation invariant and then uses it to search for the weakest possible residual. 

\begin{algorithm}[t]
\caption{\textsc{FindStrongestBisimulation}}
\label{alg:find-bisim}
\begin{algorithmic}[1]
\Statex \textbf{Input:} Candidate input filter $Q$, symbolic bounds $(\bisim_{\min}, \bisim_{\max})$
\Statex \textbf{Output:} Strongest bisimulation invariant $\bisim$ or root cause diagnosis upon failure

\Function{FindStrongestBisimulation}{$Q, \bisim_{\min}, \bisim_{\max}$}
  \State $\bisim_{\text{cand}} \gets \bisim_{\max}$

  \State \textbf{repeat until} $\bisim_{\text{cand}}$ \textbf{converges}
    \State \hspace{1.2em} $(\mathsf{ok}, \bisim_{\text{cand}}, r) \gets \Call{WeakenViaVC}{Q(r), f(a_2,r), \bisim_{\min}, \bisim_{\text{cand}}}$ \Comment{\sync}
    \State \hspace{1.2em} \textbf{if} $\neg \mathsf{ok}$ \Return $(\mathsf{false}, -, \mathsf{Sync}(r))$
    \State \hspace{1.2em} $(\mathsf{ok}, \bisim_{\text{cand}}, r) \gets \Call{WeakenViaVC}{\neg Q(r), a_2, \bisim_{\min}, \bisim_{\text{cand}}}$ \Comment{\stutter}
    \State \hspace{1.2em} \textbf{if} $\neg \mathsf{ok}$ \Return $(\mathsf{false}, -, \mathsf{Stutter}(r))$
  
  \State \Return $(\mathsf{true}, \bisim_{\text{cand}}, -)$
\EndFunction

\Function{WeakenViaVC}{$G(r), t_2, \bisim_{\min}, \bisim_{\text{cand}}$}
  \ForAll{$p \in \bisim_{\text{cand}}$}
    \State $\chi \gets \big((\bisim_{\text{cand}} \land G(r) \land a_1' = f(a_1,r) \land a_2' = t_2) \Rightarrow p(a_1', a_2')\big)$
    \If{$\mathsf{Invalid}(\chi)$}
      \If{$p \in \bisim_{\min}$}
        $(a_1, a_2, r) \gets \mathsf{Model}(\neg \chi)$;\ \Return $(\mathsf{false}, -, r)$
      \Else\ $\bisim_{\text{cand}} \gets \bisim_{\text{cand}} \setminus \{p\}$ \EndIf
    \EndIf
  \EndFor
  \State \Return $(\mathsf{true}, \bisim_{\text{cand}}, -)$
\EndFunction
\end{algorithmic}
\end{algorithm}

Algorithm~\ref{alg:find-bisim} (\textsc{FindStrongestBisimulation}) computes the strongest bisimulation invariant for a given candidate $Q$. This is the most straightforward part of the  algorithm and follows the familiar Houdini-style weakening loop~\cite{houdini}, with two modifications.
 First, instead of initializing from the full invariant universe $ U_{\bisim} $, it uses the upper bound $\bisim_{\max}$, which already excludes atoms ruled out by earlier analysis. 
Second, it enforces the symbolic lower bound $\bisim_{\min}$ throughout. Whenever $\bisim_{\min}$ is violated, we perform root cause analysis as follows: any attempt to remove an atom in $\bisim_{\min}$ indicates that no admissible invariant exists for the current filter, and the procedure returns
\textsf{false}. In addition, it returns a symbolic row $r$ as a witness to guide  subsequent repair (lines 5 and 7). 

\paragraph{Remark.}  This Houdini-style inference of the bisimulation invariant in Algorithm~\ref{alg:find-bisim} is only possible because of our careful stratification of the search: 
Without a fixed $Q$, we
would not be able to
apply such an algorithm because the constraints governing $\bisim$ change along with $Q$.

\smallskip

\begin{algorithm}[t]
\caption{\textsc{FindResidual}}
\label{alg:find-residual-final}
\begin{algorithmic}[1]
\Require Bisimulation invariant $\bisim$, residual universe $U_{P'}$, post-UDF filter $P$
\Ensure Weakest residual $P'$ such that \final holds, or failure
\Function{FindResidual}{$\bisim, U_{P'}, P$}

  \State \textcolor{blue}{\small \texttt{\# Early check: does $P'=P$ satisfy \final?}}
  \State $\chi \gets \big( \bisim(a_1, a_2) \Rightarrow [(P(a_1)\wedge P(a_2)\wedge a_1 = a_2)\ \lor\ (\neg P(a_1)\wedge \neg P(a_2))] \big)$
  \If{\textsf{Invalid}($\chi$)} \Return $(\mathsf{false}, -)$ \EndIf

  \State $W \gets \{\top\}$; \quad $\mathsf{Visited} \gets \varnothing$ \Comment{Worklist of residual candidates}

  \While{$W \neq \varnothing$}
    \State $P' \gets \textsc{Dequeue}(W)$
    \If{$P' \in \mathsf{Visited}$} \textbf{continue} \EndIf
    \State $\mathsf{Visited} \gets \mathsf{Visited} \cup \{P'\}$

    \State \textcolor{blue}{\small \texttt{\# Check \final under $\bisim$ for current $P'$}}
    \State $\chi \gets \big( \bisim(a_1, a_2) \Rightarrow [(P(a_1)\wedge P'(a_2)\wedge a_1 = a_2)\ \lor\ (\neg P(a_1)\wedge \neg P'(a_2))] \big)$
    \If{\textsf{Valid}($\chi$)} \Return $(\mathsf{true}, P')$ \EndIf
    \State \Call{HandleFailure}{$P, P', \mathsf{Model}(\neg \chi), U_{P'}, W$} \Comment{May enqueue strengthened candidates}
  \EndWhile

  \State \Return $(\mathsf{false}, -)$
\EndFunction


\Function{HandleFailure}{$P, P', (a_1, a_2), U_{P'}, W$} 
  \If{$P(a_1) \land P'(a_2) \land a_1 \neq a_2$}  \Return  \Comment{Mismatch on acceptance}
  \ElsIf{$P(a_1) \land \neg P'(a_2)$}  \Return  \Comment{False rejection}
  \ElsIf{$\neg P(a_1) \land P'(a_2)$} \Comment{Spurious acceptance}
    \ForAll{$p \in U_{P'} \setminus P'$}
      \If{$\neg p(a_2)$}  $P'' \gets P' \land p$; \ \textsf{Enqueue}($W, P''$)
      \EndIf
    \EndFor
  \EndIf
\EndFunction

\end{algorithmic}
\end{algorithm}

Next, Algorithm~\ref{alg:find-residual-final} ({\sc FindResidual}) describes how to compute the weakest residual for a given $Q$ and $\bisim$. At a high level, this algorithm iteratively strengthens the candidate residual until the \final VC becomes valid. 
As an initial optimization,  Algorithm~\ref{alg:find-residual-final} checks whether the original filter $P$ satisfies \final (lines 3--4). If not, $Q$ is too strong, and no residual can succeed, so the procedure returns \textsf{false}. 
Otherwise, the algorithm considers increasingly stronger candidates and returns the current candidate $P'$ as soon as \final becomes valid. The key part of the procedure is how it handles failures by performing root cause analysis in {\sc HandleFailure}. If $\chi$ is violated, line 13 of Algorithm~\ref{alg:find-residual-final} leverages a model $(a_1, a_2)$ of $\neg \chi$ to identify a possible repair strategy. In particular, the algorithm differentiates between three failure modes:

\begin{itemize}[leftmargin=*]
    \item \textbf{Mismatch on acceptance:} If $P(a_1) = P'(a_2) = \mathsf{true}$ and $a_1 \ne a_2$ (line 16 in {\sc HandleFailure}), the original and optimized executions disagree on which outputs they accept. No strengthening of the residual predicate can resolve this case, so the algorithm discards the current candidate.
    \item \textbf{False rejection:} If $P(a_1) = \mathsf{true}$ and $P'(a_2) = \mathsf{false}$ (line 17),  the optimized execution incorrectly rejects an output that the original accepts. As in the previous case, strengthening $P'$ cannot reverse this outcome, so the candidate is again discarded.
    \item \textbf{Spurious acceptance:} If $P(a_1) = \mathsf{false}$ and $P'(a_2) = \mathsf{true}$ (lines 18--20), the optimized execution incorrectly accepts an output that should be rejected. In this case, the algorithm strengthens $P'$ to exclude $a_2$ by conjoining each atom $p \in U_{P'} \setminus P'$ such that $p(a_2) = \mathsf{false}$, and enqueues each strengthened candidate onto the worklist $W$.
\end{itemize}

Importantly, the root cause analysis performed by {\sc HandleFailure}
further
reduces the search space of
the synthesis procedure, as not all failed residual candidates require further exploration. 

%% file: sections/eval.tex
\section{Evaluation} \label{sec:eval}

{
We have implemented our method in a new tool called \toolname, which is written in OCaml and uses the Z3 SMT solver~\cite{Z3}. \toolname leverages a custom, Python-like DSL to represent UDFs in a uniform, framework-agnostic way. This DSL serves as a realistic intermediate representation (IR) for data pipelines written in Python, Scala, or Java, enabling \toolname to support a broad range of data processing frameworks such as \texttt{pandas} and Spark.
To construct suitable predicate universes for each task, \toolname performs lightweight data-flow analyses over the UDF accumulator function~$f$, initializer~$I$, and post-UDF filter~$P$. Concretely, \toolname constructs the predicate universes for $Q$, $P'$, and invariants using a largely two-step procedure: First, it extracts a finite set of base predicates from each benchmark by taking the post-filter predicate and the UDF's control-flow conditions and normalizing them into CNF clauses. Then, it augments this base set of CNF clauses with disjunctive combinations between related atoms, where the notion of ``related'' is determined based on data dependence. Further details about our implementation, including the DSL and the predicate universe construction procedure, can be found in Appendix~\ref{sec:impl-details}.
}

In {the rest of} this section, we describe an experimental evaluation of \toolname that aims to answer the following research questions:
\begin{enumerate}
\item[\textbf{RQ1.}] How does our method compare against prior work~\cite{MagicPush}?
\item[\textbf{RQ2.}] How efficiently can \toolname{} synthesize correct predicate pushdown decompositions, and what is the complexity of the  synthesized predicates and invariants?
{
\item[\textbf{RQ3.}] What impact do pushdown optimizations have on the running time of real-world pipelines? 
}
\item[\textbf{RQ4.}] What is the impact of key design decisions on \toolname's synthesis efficiency?
\item[\textbf{RQ5.}] How sensitive is \toolname to the size of the predicate universe in terms of solution quality and synthesis time?

\end{enumerate}


\bfpara{Benchmarks.} To answer these questions, we evaluated \toolname on \numBenchmarks benchmarks, consisting of \numUDF unique UDFs and 3-9 predicates per UDF. The UDFs used in our evaluation are drawn from two sources: (i) 19 Spark UDFs taken in their entirety from prior work~\cite{ink} and (ii) 7 pandas UDFs crawled from GitHub using the same inclusion criteria as~\cite{ink}, focusing on UDFs that operate over the entire dataframe rather than row-wise maps and that involve at least one of conditional logic, collection processing, or multiple sub-aggregations. These benchmarks cover diverse domains (e.g., finance, machine learning, healthcare), and all are {stateful in the fold sense}: That is, each UDF updates an accumulator state across the input sequence, so they are not row-wise stateless map functions. The pandas UDFs are written in Python, while the Spark UDFs use Scala or Java.

\begin{wraptable}{r}{0.3\textwidth}
\vspace{-0.8em}
\centering
\footnotesize
\setlength{\tabcolsep}{3pt}
\begin{tabular}{lc}
\toprule
\textbf{Feature} & \textbf{Value} \\
\midrule
Total UDFs & 26 \\
\quad \texttt{pandas} UDFs & 7 \\
\quad Spark UDFs & 19 \\
\# w/ conditionals & 23 \\
\# w/ collections & 8 \\
\# w/ tuple accumulator & 18 \\
\makecell[l]{\# w/ cross-dependent\\[-0.6ex]\hphantom{\# }aggregations} & 8 \\
\midrule
Avg AST size & 77 \\
Max AST size & 168 \\
\bottomrule
\end{tabular}
\caption{UDF statistics}
\label{tab:benchmark-summary}
\vspace{-3.5em}
\end{wraptable}

\bfpara{Pre-processing.} Since these benchmarks span several languages, we manually translated them to \toolname's underlying DSL. This translation is largely syntax-directed and preserves the structure of the original computation. In a few cases, our translation required making some semantics-preserving  adaptations, such as replacing timestamps with integers, to align with our DSL.


\bfpara{UDF statistics.} Table~\ref{tab:benchmark-summary} summarizes key structural and syntactic features of the UDFs used in our evaluation. {The vast} majority (85\%) include conditionals, 31\% utilize collections, and 69\% use tuple-valued accumulators to track multiple pieces of state. Furthermore, 8 of the 26 UDFs (31\%) feature cross-dependent aggregations, where the computation of one aggregation depends on the intermediate results of {one or more of others}. These characteristics, along with nontrivial AST sizes (average 77 nodes), highlight the semantic complexity of the benchmarks.

\bfpara{Filters.} To evaluate \toolname's pushdown capabilities, each benchmark must be paired with a post-UDF filtering predicate. For benchmarks that already include a post-UDF filter (as is the case for most \texttt{pandas} cases), we retained the original predicate. To broaden the evaluation space, we also generated additional predicates using a schema designed to reflect common filtering idioms:
\begin{itemize}[leftmargin=*]
  \item \textbf{Scalar-valued outputs:} For UDFs that return a single scalar value~$v$, we generated predicates of the form $\{v~\mathbin{\texttt{op}}~c \mid \texttt{op} \in O, c \in C\}$, where $O = \{=, \neq, >, \geq, <, \leq\}$ and $C$ is a set of constants either drawn from the original pipeline or chosen to reflect typical threshold-style filters. For cases where $v$ is an \texttt{Optional} value, we wrapped the predicate in a \texttt{match} expression, with the \texttt{Some} branch applying the comparison and the \texttt{None} branch returning a fixed Boolean.

  \item \textbf{Tuple-valued outputs:} For UDFs that return a tuple~$a$, we constructed field-level predicates of the form $\{a[i]~\mathbin{\texttt{op}}~c \mid 0 \le i < \texttt{length}(a), \texttt{op} \in O, c \in C\}$, applying the same treatment as in the scalar case for any \texttt{Optional} fields. These serve as building blocks for more complex filters.

  \item \textbf{Compound predicates:} To capture more expressive queries, we composed field-level predicates into compound filters using conjunctions and disjunctions across multiple tuple fields. These reflect common patterns of multi-attribute filtering found in analytic workloads.
\end{itemize}
For each UDF, we generated 10--20 candidate predicates. Since our goal is to evaluate \toolname's optimization capability when pushdown is actually feasible, we retained only those predicates offering genuine optimization opportunities. We first checked pushdown feasibility on small, bounded symbolic dataframes (see Appendix~\ref{sec:bmc} for details) and kept only benchmarks where pushdown was \emph{potentially} feasible. We then manually filtered these cases to ensure that optimization was indeed possible. After this two-stage filtering, each UDF contributed between three and nine valid predicates, with an average of 5.8, yielding a total of 150 benchmarks for evaluation.

\bfpara{Experimental setup.}
All of the experiments are conducted on a machine with an Apple M3 Max chip (14-core CPU) and 36 GB of memory. We use a 10-minute timeout for each benchmark.

\subsection{Comparison with Prior Work}

In this section, we compare \toolname to \magicpush~\cite{MagicPush}, the only existing approach that performs predicate pushdown for UDFs. We first provide some necessary background on \magicpush and then present our empirical findings.

\bfpara{Background.}   \magicpush supports only exact and partial pushdown and verifies each candidate using a small‐model theorem. It performs bounded symbolic execution on a fixed ``mini‐table,'' then relies on structural preconditions to lift the proof to unbounded inputs. These preconditions require that the UDF: (1)
be associative and commutative; (2) admit a ``small‐model'' reduction, i.e., a smaller input yielding the same output; (3) be monotonic with respect to the post‐UDF predicate; and (4) be insensitive to rows removed by the pre‐filter. 
To synthesize pre-UDF filters, \magicpush enumerates disjunctions of predicate conjunctions drawn from the pipeline. Each candidate is checked for \emph{exact} pushdown; if this fails, it is re-evaluated for \emph{partial} pushdown.

\bfpara{Experimental setup.} Since \magicpush's complete implementation is not publicly available, we cannot evaluate it directly.\footnote{The implementation of \magicpush was not publicly released alongside the original paper~\cite{MagicPush}. We contacted the authors at the start of this work and received partial source code. However, essential components of their tool and benchmark data were unavailable and could not be recovered, making a direct empirical comparison infeasible.} 
Instead, we conservatively report an \emph{upper bound} on \magicpush's performance: for any benchmark that meets its structural preconditions and passes the bounded check, we assume it could perform exact or partial pushdown.  
However, this estimate is highly optimistic: since \magicpush's search space grows combinatorially, it may still fail to find a valid pushdown within practical time limits even when these conditions hold.

\begin{table}[t]
\centering
\footnotesize
\setlength{\tabcolsep}{4pt}
\begin{tabular}{lc>{\quad}cccc>{\quad}ccc}
\toprule
\multirow{2}{*}{\textbf{Type}}
  & \multirow{2}{*}{\textbf{\# Benchmarks}}
  & \multicolumn{4}{c}{\quad\toolname}
  & \multicolumn{3}{c}{
    \quad\magicpush (Upper Bound)
    }\\
  &
    & \textbf{Any} & \textbf{Exact} & \textbf{Partial} & \textbf{Split}
    & \textbf{Any} & \textbf{Exact} & \textbf{Partial} \\
\midrule

\texttt{pandas} & 27 & 27 (\textbf{100.0\%}) & 9 (\textbf{33.3\%}) & 12 (\textbf{44.4\%}) & 6 (22.2\%) & 12 (44.4\%) & 5 (18.5\%) & 7 (25.9\%) \\
Spark & 123 & 123 (\textbf{100.0\%}) & 2 (\textbf{1.6\%}) & 40 (\textbf{32.5\%}) & 81 (65.9\%) & 10 (8.1\%) & 0 (0.0\%) & 10 (8.1\%) \\
\midrule
\textbf{Overall} & 150 & 150 (\textbf{100.0\%}) & 11 (\textbf{7.3\%}) & 52 (\textbf{34.7\%}) & 87 (58.0\%) & 22 (14.7\%) & 5 (3.3\%) & 17 (11.3\%) \\
\bottomrule
\end{tabular}
\caption{Pushdown outcomes for \textsc{Pusharoo} vs. \textsc{MagicPush}'s theoretical best-case results. The ``Any'' column reports the number of benchmarks for which each tool successfully applies \emph{some} form of pushdown.
}
\label{tab:pushdown-results}
\vspace{-0.3in}
\end{table}


\bfpara{Main results.} Table~\ref{tab:pushdown-results} summarizes the evaluation results. For \toolname, we distinguish exact pushdowns (no residual check), partial pushdowns (residual identical to the original post-UDF filter), and split pushdowns (residual strictly weaker). For \magicpush, ``Exact'' and ``Partial'' count benchmarks whose structural preconditions and bounded check permit those modes; \magicpush does not support split pushdown.

Two main trends stand out in Table~\ref{tab:pushdown-results}. First, \toolname successfully applies some form of pushdown to \emph{every} benchmark in our suite, while \magicpush applies to only 22 out of 150 cases (14.7\%), underscoring \toolname's broader applicability. Second, \toolname's generalized split pushdown dominates both exact and partial modes across the board,  demonstrating that many UDFs admit nontrivial residuals that go beyond exact and partial pushdown. 

\bfpara{Optimality.}
By design, \toolname always returns the strongest pre‐filter \(Q^*\) and weakest residual \(P^*\) within our chosen predicate universe, so every solution is optimal in that sense. Since proving optimality over all conceivable predicates lies beyond our finite search space, \toolname does not provide a global optimality guarantee.
Hence, to evaluate optimality in practice, we manually examined 10 benchmarks selected by taking two or three examples from each of the four UDF categories from Table~\ref{tab:benchmark-summary}, spanning a range of UDF and post-filter complexities. For each selected benchmark, we manually constructed what we believed to be the optimal pushdown and compared it against the decomposition synthesized by Pusharoo. In all 10 cases, the synthesized result matched the manually identified optimum.
Furthermore, among the 17 cases where \magicpush can theoretically apply a partial pushdown optimization, \toolname finds a weaker (and simpler) residual than the original post-UDF predicate for 8 of these (47.1\%). These findings indicate that, even in cases where \magicpush is applicable, \toolname produces strictly better decompositions.
\medskip


\begin{mdframed}[backgroundcolor=gray!8,linewidth=0.5pt]
\textbf{RQ1 Summary:}  \toolname is more expressive than \magicpush, achieving valid pushdown transformations on \textbf{100\%} of our 150 benchmarks. In contrast, {\magicpush can feasibly be applied to} only \textbf{22/150 (14.7\%)} of these cases, leaving \textbf{85.3\%} of real-world UDFs beyond its theoretical capabilities.
\end{mdframed}

\subsection{Running Time and Predicate Complexity}

In this section, we evaluate \toolname's performance in terms of synthesis time and report on the complexity of the synthesized artifacts. The results of this evaluation are presented in  Table~\ref{tab:main-results-summary}, which groups benchmarks by solution type: exact, partial, or split. For each category, we report the number of benchmarks, average synthesis time, and average sizes of the strongest pre‐filter $Q^*$, the {weakest} residual $P^*$, the original predicate $P$, and the bisimulation invariant $\psi$.

\begin{table}[t]
\centering
\footnotesize
\begin{tabular}{l c c c c c c}
\toprule
\textbf{Type} 
  & \textbf{\# Benchmarks} 
    & \textbf{Avg Runtime (s)} 
      & $\mathbf{Avg~|Q^*|}$ 
        & $\mathbf{Avg~|P^*|}$ 
          & $\mathbf{Avg~|\psi|}$
            & $\mathbf{Avg~|P|}$\\
\midrule

Exact & 11 & 0.26 & 1.18 & N/A & 11.27 & 1.18 \\
Partial & 52 & 1.98 & 2.52 & 4.60 & 19.69 & 5.48 \\
Split & 87 & 5.30 & 2.31 & 3.41 & 36.91 & 5.09 \\
\midrule
\textbf{Overall} & 150 & 3.78 & 2.30 & 3.57 & 29.06 & 4.94 \\
\bottomrule
\end{tabular}
\caption{Breakdown of the types of pushdowns synthesized by \toolname.}
\label{tab:main-results-summary}
\vspace{-0.2in}
\end{table}

Overall, \toolname synthesizes pushdowns in {less than 4} seconds on average. The synthesized pre‐filters have an average size of {2.3} clauses, and residual checks average 3.6 clauses compared to 4.9 for the original predicates. Bisimulation invariants average {29.1} in size, indicating that proving correctness of pushdown optimizations over UDFs is often nontrivial. As expected, split decompositions incur the highest synthesis cost ({5.30}s on average) as they require finding a more complex invariant with average size {36.9}. Exact cases, which eliminate the residual entirely, complete fastest ({0.26}s on average) and have the simplest invariants ({11.3} in size). Partial cases fall in between, taking {less than 2} seconds and requiring invariants of average size {19.7}. This progression reflects the deeper semantic reasoning required as solutions move from exact to partial to split.
Appendix~\ref{sec:case-studies} presents two case studies to  illustrate the kind of pushdowns enabled by \toolname.

\medskip
\begin{mdframed}[backgroundcolor=gray!8, linewidth=0.5pt]
\textbf{RQ2 Summary:}  
\toolname  takes  a median of {1.60}s (average {3.78}s) to synthesize optimal pushdown optimizations. The generated bisimulation invariants have an average size of {29.1}, reflecting the complexity of the relational reasoning required. Finally, on 87 of 150 benchmarks, \toolname produces split solutions that go beyond  exact and partial pushdown.
\end{mdframed}

\subsection{End-To-End Performance}
\label{sec:end-to-end}

Beyond synthesis coverage and efficiency, we also measured the run-time impact of our pushdown optimizations on the 150 \texttt{pandas} and Spark pipelines by comparing end-to-end runtimes before and after optimization. Following empirical benchmarking practices common in database research~\cite{dsb21,tpcds02,leis15}, we generated 100 million input rows for each pipeline using mixed Zipfian and uniform distributions chosen to mirror realistic domain semantics. 
As shown in Table~\ref{tab:end-to-end}, \toolname reduced execution time by 58.8\% on average (median 61.4\%), with some pipelines achieving improvements of up to 99.2\%. These savings correspond to an average speedup of 2.4$\times$, with extreme cases reaching two orders of magnitude.
{As expected, exact pushdown yields the largest speedup on average by completely eliminating the post-UDF filter $P$, followed by split pushdown, which weakens $P$, and finally partial pushdown, which preserves it.}


\begin{wraptable}{r}{0.35\textwidth}
\vspace{-2.3em}
\centering
\footnotesize
\setlength{\tabcolsep}{3pt}
\begin{tabular}{lcccc}
\toprule
\textbf{Type} & \textbf{Min} & \textbf{Avg} & \textbf{Median} & \textbf{Max} \\
\midrule
Exact   & 36.9\% & 70.1\% & 80.8\% & 91.4\% \\
Partial & 24.2\% & 56.1\% & 54.7\% & 84.2\% \\
Split   & 25.1\% & 59.1\% & 61.8\% & 99.2\% \\
\midrule
Overall & 24.2\% & 58.8\% & 61.4\% & 99.2\% \\
\bottomrule
\end{tabular}
\caption{Pipeline runtime reduction.}
\label{tab:end-to-end}
\vspace{-3.3em}
\end{wraptable}
\bfpara{Split vs.\ partial pushdown.}
To quantify the benefit of synthesizing a weakened residual over retaining the original post-UDF filter, we compared split pushdown against partial pushdown on the 87 benchmarks that admit split pushdown. Across these, split pushdown is on average 5.1\% faster, up to 14\%.
The gains are most pronounced for UDFs returning complex types (e.g., large tuples or collections), where redundant predicate evaluation is more expensive.

\medskip

\begin{mdframed}[backgroundcolor=gray!8, linewidth=0.5pt]
\textbf{RQ3 Summary:} \toolname optimizations yield major end-to-end runtime savings, speeding up execution by $2.4\times$ on average and up to two orders of magnitude in some cases.
\end{mdframed}

\subsection{Ablation Studies}

To evaluate the significance of key algorithmic choices in \toolname{}, we conducted ablation studies by disabling or replacing essential components individually. Specifically, we consider:
\begin{itemize}[leftmargin=*]


\item \ablnb, which does not infer bounds on the bisimulation invariant. 


\item \ablnr, which disables all root cause analyses and  targeted repairs. 

\item \abltp, which first generates candidate filters $Q, P'$ and \emph{then} verifies their correctness by attempting to find a bisimulation invariant. 

\item \ablenum, which enumerates $Q$ candidates from strongest to weakest, performing an independent full solve for each without sharing invariant reasoning across candidates.





\item \ablspacer and \ableld, which    replace lines 4--10 of \textsc{SynthesizeOptimalPushdown} (Algorithm~\ref{alg:synthesize}) with a call to \textsc{Spacer}~\cite{Spacer2} and \textsc{Eldarica}~\cite{Eldarica} (two of the best-performing CHC solvers), respectively, to infer both $\psi$ and $P'$.

\end{itemize}

We note that the \ablspacer and \ableld ablations use an off-the-shelf CHC solver to infer \emph{only} the bisimulation invariant~$\bisim$ and residual $P'$, but not the pre-UDF filter $Q$.
By instantiating $Q$ before invoking the solver, all VCs remain in strict Horn form (i.e., each clause has at most one positive head literal) so \textsc{Spacer} and \textsc{Eldarica} can, in principle, solve the resulting Horn clauses. If $Q$ is instead treated as an unknown, the Horn-clause restriction is violated, causing CHC solvers to fail outright.
Additionally, we note that \ablspacer and \ableld do not have optimality guarantees and may not find the weakest residual.

\begin{figure}[t]
\centering
\captionsetup{skip=2pt}

\begin{minipage}[b]{0.47\textwidth}
  \centering
  \vspace{0pt} 
  \includegraphics[scale=0.31]{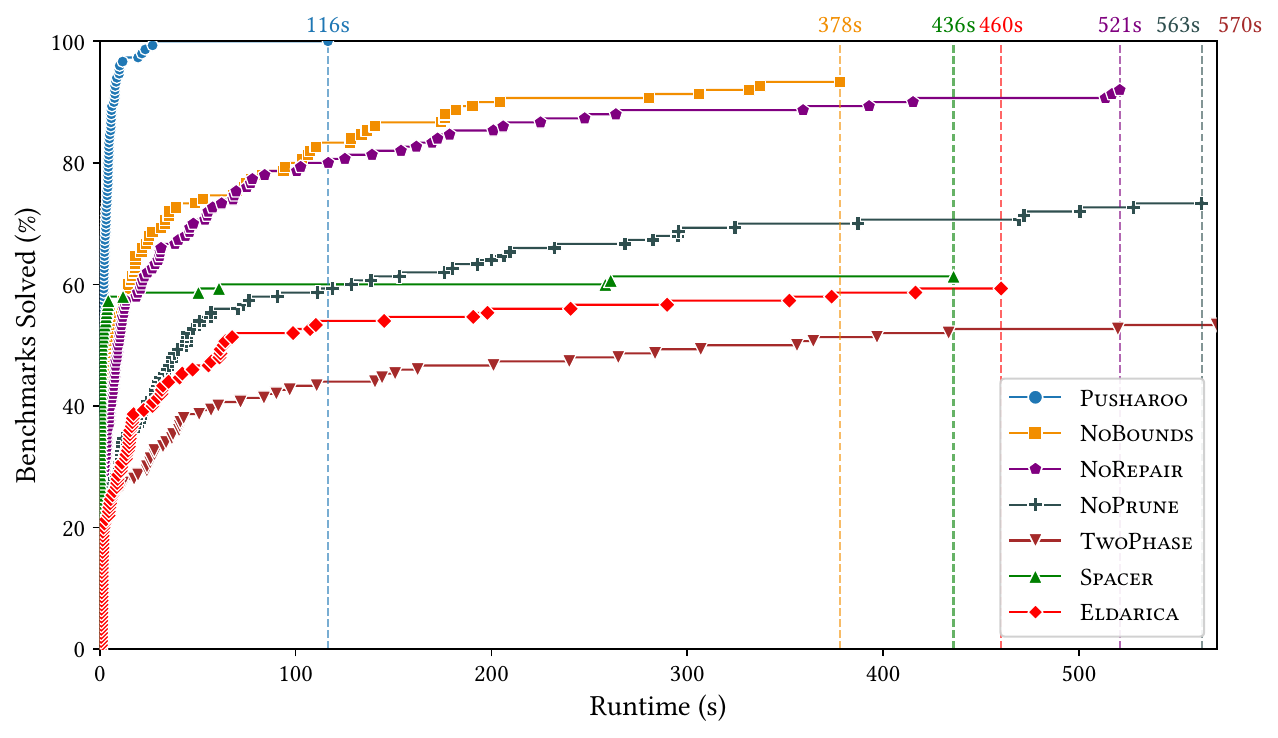}
  \captionof{figure}{CDF plot of \% of benchmarks solved over solving time for \toolname{} and ablations.}
  \label{fig:cdf}
\end{minipage}
\hfill
\begin{minipage}[b]{0.52\textwidth}
  \centering
  \vspace{0pt} 
  \scriptsize
  \begin{tabular}{lccc}
    \toprule
    \textbf{Type} & \textsc{NoBounds} & \textsc{NoRepair} & \textsc{TwoPhase} \\
    \midrule
Exact   & 33.3 / 36.4 / 50.0 & 33.3 / 37.9 / 50.0 & 0.0 / 6.1 / 66.7 \\
Partial & 52.6 / 51.6 / 94.4 & 50.0 / 57.3 / 99.2 & 85.3 / 59.1 / 99.4 \\
Split   & 40.3 / 38.3 / 92.4 & 70.0 / 64.7 / 95.1 & 80.9 / 49.2 / 99.5 \\
\midrule
Overall & 41.1 / 42.8 / 94.4 & 61.7 / 60.1 / 99.2 & 64.8 / 46.8 / 99.5 \\
    \bottomrule
  \end{tabular}
  \vspace{4.5em}
  \captionsetup{skip=-0.1pt}
  \captionof{table}{Search space \% reduction (median / average / max) enabled by \toolname{}.}
  \label{tab:space-reduction}
\end{minipage}
\end{figure}


\bfpara{Solve rate and running time comparison.} To quantify the impact of these ablations on solve rate and running time, we present a \textit{{cumulative distribution function (CDF)} plot} in Figure~\ref{fig:cdf} that shows the percentage of benchmarks solved ($y$-axis) against per-benchmark solving time ($x$-axis). 


Figure~\ref{fig:cdf} shows that \toolname significantly outperforms all ablations. Among all variants, \abltp performs the worst, highlighting the importance of  the tight coupling between synthesis (i.e., finding pushdown predicates) and verification (i.e., finding a bisimulation to prove correctness of optimization). Even the best performing ablation, namely \ablnb, solves fewer benchmarks and takes $9.6\times$ as long to solve the benchmarks that can be solved by both \toolname and \ablnb.



\begin{wraptable}{r}{0.28\textwidth}
\vspace{0.3em}
\centering
\footnotesize
\setlength{\tabcolsep}{3pt}
\begin{tabular}{lcc}
  \toprule
  \textbf{Variant} & \textbf{Solved} & \textbf{Subopt} \\
  \midrule
  \toolname           & 150 & 0  \\
\textsc{Spacer} & 92 & 61 \\
\textsc{Eldarica} & 89 & 51 \\
  \bottomrule
\end{tabular}
\caption{Solution optimality.}
\label{tab:synth-opt}
\vspace{-3.5em}
\end{wraptable}

\bfpara{Optimality comparison.}  In addition to performing significantly worse than \toolname, the CHC ablations also do not come with optimality guarantees. To quantify this aspect, Table~\ref{tab:synth-opt} additionally explores how often the CHC ablations yield a suboptimal decomposition compared to \toolname.  Overall, \ablspacer produces suboptimal residuals on 61 benchmarks (66\% of what it solves), and \ableld on 51 benchmarks (57\% of what it solves). These results demonstrate that replacing our core synthesis procedure with off-the-shelf CHC solvers not only causes many {failures and timeouts}, but also leads to suboptimal solutions for the majority of {the solved} benchmarks.

\bfpara{Search space reduction.} While Figure~\ref{fig:cdf} shows that our three algorithmic innovations (namely, structured search, unrealizability certificates, and targeted repair) have a significant impact on running time, we further directly measure the search space reduction that these ideas enable. As shown in Table~\ref{tab:space-reduction},  all of our algorithmic
contributions
result in a significant reduction in search space, with averages ranging from 43 to 60\% across the three ablations.

\medskip
\begin{mdframed}[backgroundcolor=gray!8, linewidth=0.5pt]
\textbf{RQ4 Summary:}
All of our core algorithmic ideas have a significant impact on synthesis time and success rate. Furthermore, the ablations that replace parts of the algorithm with off-the-shelf CHC solvers  are also significantly worse and yield suboptimal solutions.
\end{mdframed}

\subsection{Sensitivity to Predicate Universe Size}

\begin{wraptable}{r}{0.32\textwidth}
\vspace{-2.8em}
\centering
\footnotesize
\setlength{\tabcolsep}{3pt}
\begin{tabular}{llcc}
\toprule
\textbf{Enlarged} & \textbf{By} & \makecell[c]{\textbf{\% Better}\\[-0.3ex]\textbf{Solns}} & \makecell[c]{\textbf{Avg $\Delta$}\\[-0.3ex]\textbf{Time}} \\
\midrule
$U_Q$    & 10\% & 0\% & +2.0\% \\
         & 20\% & 0\% & +1.6\% \\
         & 30\% & 0\% & +3.5\% \\
\midrule
$U_\psi$ & 10\% & 0\% & +9.7\% \\
         & 20\% & 0\% & +21.4\% \\
         & 30\% & 0\% & +29.2\% \\
\midrule
Both     & 10\% & 0\% & +10.6\% \\
         & 20\% & 0\% & +21.2\% \\
         & 30\% & 0\% & +30.8\% \\
\bottomrule
\end{tabular}
\caption{Sensitivity analysis.}
\label{tab:sensitivity}
\vspace{-2.5em}
\end{wraptable}

To evaluate universe-size sensitivity, we systematically enlarge $U_Q$ and $U_\psi$ with new disjunctions of atoms while leaving $U_{P'}$ untouched: since the goal of weakening $P$ is to simplify the post-UDF check by \emph{dropping} conjuncts, adding disjunctions there does not make sense from the perspective of reducing post-UDF work.
Table~\ref{tab:sensitivity} shows the results. Adding predicates does not improve solution quality on any benchmark, suggesting our universe-construction heuristics already effectively capture relevant predicates. Furthermore, although worst-case complexity is exponential in $|U_Q|$, synthesis time grows modestly, indicating that our algorithmic optimizations are effective in practice.

\medskip
\begin{mdframed}[backgroundcolor=gray!8, linewidth=0.5pt]
\textbf{RQ5 Summary:}
Enlarging the predicate universes  does not improve solution quality, and synthesis time grows modestly rather than exhibiting the exponential blow-up suggested by worst-case complexity.
\end{mdframed}

%% file: sections/related.tex
    \section{Related Work} \label{sec:related}

\bfpara{Predicate pushdown.}
Predicate pushdown is a widely used optimization for improving query performance, supported by nearly all modern data systems including PostgreSQL, MySQL, SQL Server, Oracle, Spark, and Flink~\cite{postgresql,mysql,sql-server,oracle,spark,flink}.  
Prior research primarily targets relational operators, built-in aggregations, or semi-structured data formats~\cite{Larson86IVM,DIP:journals/pvldb/OrrKC19,plaque:pacmmod/LinM24,FlexPushdownDB:journals/pvldb/YangYWLYSAS21,SIA:conf/sigmod/0010ANHW21,sparser,AutomaticMR:pvldb/JahaniCR11}, but remains ineffective for most UDFs.  
To our knowledge, \magicpush~\cite{MagicPush} is the only prior system that performs pushdown through UDFs; however,  it relies on restrictive assumptions that fail for the vast majority of the real-world benchmarks from  Section~\ref{sec:eval}. Additionally, \magicpush can only perform exact and partial pushdown and cannot infer weaker nontrivial residuals.

\bfpara{Optimizing UDFs.}
Many systems attempt to expose UDF semantics to query optimizers by translating imperative code into relational form.  Representative approaches use program synthesis (\textsc{Qbs}~\cite{QBS:conf/pldi/CheungSM13}), static analysis (EqSQL~\cite{EqSQL:sigmod/EmaniRBS16}, Aggify~\cite{Aggify:sigmod/GuptaP020}), quey rewrites~\cite{Decorrelation:icde/Simhadri0CG014}, and compiler-level inlining (Froid~\cite{Froid:pvldb/RamachandraPEHG17}).
Recent work extends these ideas to Spark SQL and RDDs~\cite{zhang2021udf,zhang2023automated}, while \textsc{Dias}~\cite{Dias:journals/pacmmod/BaziotisKM24} applies lightweight rewrites to \texttt{pandas} workloads involving row-wise UDFs.
Other directions include fusing multiple UDFs to avoid redundant computation~\cite{consolidate} and verifying homomorphism properties of aggregations to synthesize merge operators for incremental and parallel execution~\cite{ink}.


\bfpara{Optimal synthesis.}
Our method relates to optimal synthesis techniques that produce correct solutions under logical or quantitative criteria, including \textsc{OptPCSat}~\cite{OptPCSat} for CHC solving, multi-abduction~\cite{multiabduction} for weakest specification inference, \textsc{loud}~\cite{loud} for extremal over and under approximations, and systems such as \textsc{Synapse}~\cite{synapse} and abstract-interpretation-guided approaches~\cite{mell24} that optimize    cost metrics. However, these approaches do not address our setting, which requires synthesizing a pair of predicates with opposing optimality objectives (strongest $Q$ and weakest $P'$).

 \bfpara{Relational verification.}
This paper is related to a long line of work on \emph{relational verification} which aims to reason about the relationship between multiple programs or different executions of the same program. Existing relational verification techniques include relational program logics~\cite{benton04,yang07,sousa16,aguirre17}, product programs~\cite{barthe11},  and methods for regression verification~\cite{lahiri12,regress09,felsing14}
In the database domain, several systems target SQL query equivalence~\cite{qed,chu17,chu18,spes}, and   \textsc{Mediator}~\cite{Mediator} addresses equivalence of database‐driven applications under schema changes. In contrast to prior work, our goal is to simultaneously synthesize and verify the correctness of pushdown optimizations.

%% file: sections/concl.tex
\section{Conclusion} \label{sec:concl}
This paper develops a unified semantic framework for pushdown optimizations, which aim to reduce the amount of data a computation must process. Our framework generalizes and extends prior exact and partial pushdown transformations, and enables them on a broader class of computations, including those with complex internal state. Building on this foundation, we introduce a bisimulation-based verification method and a synthesis algorithm that constructs optimal pushdown transformations together with correctness proofs. The synthesis procedure solves a challenging second-order constraint solving problem via structured decomposition (to guarantee optimality), symbolic bounds on the bisimulation invariant (to prove unrealizability), and
root-cause-guided predicate repair (for goal-directed reasoning).
Our implementation, \toolname, applies these techniques to
real-world \texttt{pandas} and Spark pipelines, synthesizing provably correct pushdown transformations in seconds and significantly outperforming the prior state-of-the-art, leading to substantial end-to-end improvements of up to two orders of magnitude.

%% file: sections/availability.tex
\section*{Data-Availability Statement}
The software artifact accompanying this paper, comprising the implementation of \toolname, the benchmark suite of 150 \texttt{pandas} and Spark pipelines, and scripts to reproduce the evaluation results in Section~\ref{sec:eval}, can be found on Zenodo~\cite{pusharoo_artifact} and GitHub~\cite{pusharoo_github}.

%% file: sections/appendix-proofs.tex
\section{Proofs of Theorems} \label{sec:proofs}

This appendix contains proofs of all theorems from the paper.

\subsection*{Proof of Theorem~\ref{thm:exact-partial-cases}}

Let
\[
F(x) = \mathsf{fold}(x, I, f),
\qquad
y = F(x),
\qquad
z = F(\filter_Q(x)),
\]
and recall from Section~\ref{sec:problem} that the \emph{generalized pushdown condition} is:
\[
\forall x.\quad
\Lift(P,y) = \Lift(P',z),
\qquad\text{where}\quad
\Lift(P,a) =
\begin{cases}
a & P(a)\text{ holds},\\
\bot & \text{otherwise}.
\end{cases}
\]

\begin{enumerate}[leftmargin=*]
  \item \textbf{Exact pushdown $\iff$ $(Q,P'\equiv\true)$ generalized pushdown.}
    \begin{enumerate}[label=(\roman*),leftmargin=*,align=left,widest=ii]
      \item \emph{($\Rightarrow$)} Assume $Q$ satisfies exact pushdown, i.e., for all $x$,
      \[
        (P(y) \Rightarrow y = z) \land (\neg P(y) \Rightarrow z = \bot).
      \]
      Set $P' \equiv \true$, so $\Lift(P',z) = z$.
      \begin{itemize}[leftmargin=*]
        \item If $P(y)$, then $y = z$, so $\Lift(P, y) = y = z = \Lift(P', z)$.
        
        \item If $\neg P(y)$, then $z = \bot$, so $\Lift(P, y) = \bot = \Lift(P', z)$.
      \end{itemize}
      Therefore, $\Lift(P,y) = \Lift(P',z)$ for all $x$.

      \item \emph{($\Leftarrow$)} Assume $(Q,\true)$ satisfies the generalized pushdown condition, i.e., $\Lift(P,y) = \Lift(\true,z) = z$ for all $x$. Then:
      \begin{itemize}[leftmargin=*]
        \item If $P(y)$, then $\Lift(P, y) = y = z$.

        \item If $\neg P(y)$, then $\Lift(P, y) = \bot = z$.
      \end{itemize}
      Hence, the exact pushdown condition holds.
    \end{enumerate}

  \item \textbf{Partial pushdown $\iff$ $(Q,P'\equiv P)$ generalized pushdown.}
    \begin{enumerate}[label=(\roman*),leftmargin=*,align=left,widest=ii]
      \item \emph{($\Rightarrow$)} Assume $Q$ satisfies partial pushdown, i.e., for all $x$,
      \[
        (P(y) \Rightarrow P(z)) \land (P(z) \Rightarrow y = z).
      \]
      Set $P' \equiv P$, so $\Lift(P',z) = \Lift(P,z)$.
      \begin{itemize}[leftmargin=*]
        \item If $P(z)$, then $y = z$, so $\Lift(P, y) = \Lift(P, z) = z = \Lift(P, z) = \Lift(P', z)$.

        \item If $\neg P(z)$, then $\neg P(y)$, so $\Lift(P, y) = \bot = \Lift(P, z) = \Lift(P', z)$.
      \end{itemize}
      Therefore, $\Lift(P,y) = \Lift(P',z)$ for all $x$.

      \item \emph{($\Leftarrow$)} Assume $(Q,P)$ satisfies the generalized pushdown condition, i.e., $\Lift(P,y) = \Lift(P,z)$ for all $x$. Then:
      \begin{itemize}[leftmargin=*]
        \item If $P(z)$, then $\Lift(P, y) = \Lift(P, z) = z \neq \bot$, so $y = z$.
        \item If $\neg P(z)$, then $\Lift(P, y) = \bot$, so $\neg P(y)$.
      \end{itemize}
      Hence, the partial pushdown condition holds.
    \end{enumerate}
\end{enumerate} \qed

\subsection*{Proof of Theorem~\ref{thm:sound-complete} (Soundness and Relative Completeness)}

We reduce correctness of the pushdown transformation to a Hoare‐style proof of a \emph{product program} that simultaneously simulates the original and filtered executions. Define the product program:
\[
\begin{array}{l}
a_1 := I; \quad a_2 := I;\\
\textbf{for each }r \textbf{ in } x\ \{\\
\quad\textbf{if }Q(r)\textbf{ then } \{a_1 := f(a_1,r);\; a_2 := f(a_2,r)\} \\
\quad\textbf{else } \{a_1 := f(a_1,r)\} \\
\}\\
\textbf{assert }\mathsf{Lift}(P,a_1)=\mathsf{Lift}(P',a_2)
\end{array}
\]

We show that our four verification conditions (VCs) correspond exactly to the standard Hoare‐logic proof obligations for establishing the partial‐correctness triple
\[
\{\,a_1=I \wedge a_2=I\,\}\;\;\text{product program}\;\;\{\,\Lift(P,a_1)=\Lift(P',a_2)\,\}.
\]

\paragraph{1. Initialization}  
This is the Hoare precondition: from $a_1=I$ and $a_2=I$, the invariant $\psi(a_1,a_2)$ must hold before the loop.

\paragraph{2. Synchronized‐step preservation}  
Corresponds to the Hoare rule for the ``then'' branch when $Q(r)$ is true.  We require
\[
\{\,\psi(a_1,a_2)\,\}\;a_1 := f(a_1,r);\;a_2 := f(a_2,r)\;\{\,\psi(a_1,a_2)\,\},
\]
which is exactly our synchronized‐step condition.

\paragraph{3. Stutter‐step preservation}
Corresponds to the Hoare rule for the ``else'' branch when $Q(r)$ is false.  We require
\[
\{\,\psi(a_1,a_2)\,\}\;a_1 := f(a_1,r)\;\{\,\psi(a_1,a_2)\,\},
\]
with $a_2$ unchanged, matching our stutter‐step condition.

\paragraph{4. Final‐state agreement}
Corresponds to the Hoare postcondition.  From $\psi(a_1,a_2)$ after the loop, we must show
\[
\Lift(P,a_1)=\Lift(P',a_2),
\]
which is our final‐state agreement condition.

\medskip
By the \emph{soundness} of Hoare logic, if one discharges these four proof obligations, then the asserted triple is valid, and hence
$\forall x.\;\Lift(P,F(x))=\Lift(P',F(\filter_Q(x)))$.

\bigskip\noindent
For \emph{relative completeness}, we appeal to Cook’s theorem~\cite{cook78} for Hoare logic: any valid partial‐correctness triple over a sufficiently expressive assertion logic admits a Hoare‐logic proof using some loop invariant.  Here, the loop invariant is exactly a bisimulation invariant $\psi$.  Hence, whenever
$\forall x.\;\Lift(P,F(x))\allowbreak=\Lift(P',F(\filter_Q(x)))$
holds semantically, there must exist an invariant $\psi$ that satisfies the four VCs above. \qed

\subsection*{Proof of Theorem~\ref{thm:optimal-synth-guarantees} (Guarantees of \textsc{SynthesizeOptimalPushdown})}

We begin by showing that worklist $W$ strictly decreases through the while loop in \textsc{SynthesizeOptimalPushdown}. We refer to candidate input filters in a worklist with the greatest number of atoms as the \emph{largest}. Formally, let the size of the largest filter(s) in a worklist be $\sigma(W) = \max_{Q\in W}\,\lvert Q \rvert$ and the number of largest filters in a worklist be $\kappa(W) = \sum_{Q \in W} \mathds{1}_{\{Q' \mid \sigma(W) = \lvert Q' \rvert\}}(Q)$, where $\mathds{1}$ is an indicator function that returns $1$ if the argument filter is in the set in the subscript and $0$ otherwise.
We then define a standard lexicographic order on worklists: for worklists $W_a$ and $W_b$,
\[
  W_a \prec W_b
  \;\Longleftrightarrow\;
  \sigma(W_a) < \sigma(W_b) \lor
  (\sigma(W_a) = \sigma(W_b) \land \kappa(W_a) < \kappa(W_b)).
\]
$\preceq$ is defined similarly. We prove that the worklist strictly decreases through the while loop according to $\prec$ via the following lemma:

\begin{lemma}[Worklist Strictly Decreases]\label{lem:synth-worklist-decreases}
  Let $W_k$ be the worklist at the beginning of the $k$-th iteration of the while loop. For all $k \ge 0$, $W_{k+1} \prec W_k$.
\end{lemma}

\begin{proof}
  For any $k$, let $Q_k = \textsc{Dequeue}(W_k)$ and $W_k'$ be the remainder of $W_k$. By Theorem~\ref{thm:weaken-correct}, \textsc{WeakenViaBounds} always terminates. If \textsc{WeakenViaBounds}$(Q_k, P, U_\psi)$ returns $(\mathsf{false}, -, -, -)$, then the iteration ends. Since the worklist is a max-heap, $Q_k$ is the largest in $W_k$, i.e., $|Q_k| = \sigma(W_k)$. Suppose $Q_k$ is the sole largest, then $\sigma(W_{k+1}) = \sigma(W_k') < \sigma(W_k)$. Otherwise, we are left with one less largest filter, i.e., $\kappa(W_{k+1}) = \kappa(W_k) - 1$. In either case, we have $W_{k+1} \prec W_k$.
  
  Otherwise, \textsc{WeakenViaBounds} returns $(\mathsf{true}, Q_k', \psi_{\min}, \psi_{\max})$, where, by Theorem~\ref{thm:weaken-correct}, $Q_k' \subseteq Q_k$. The loop iteration then proceeds to execute \textsc{FindStrongestBisimulation},
  which, by Theorem~\ref{thm:find-strongest-bisim-correct}, always terminates. If \textsc{FindStrongestBisimulation}$(Q_k', \psi_{\min}, \psi_{\max})$ returns $(\mathsf{false}, -, \mathsf{diagnosis})$, then \textsc{Repair}$(Q_k', \mathsf{diagnosis})$, by definition, proceeds to terminate and produce $Q_k'' \subset Q_k'$. Thus, $Q_k'' \subset Q_k$. Since only $Q_k''$ and all one-atom weakenings of $Q_k'$ are enqueued, we have $W_{k+1} \prec W_k$ regardless of whether $Q_k$ is the sole largest in $W_k$. A similar reasoning applies when \textsc{FindStrongestBisimulation} returns $(\mathsf{true}, \psi, -)$ and \textsc{FindResidual}$(\psi, U_{P'}, P)$ returns $(\mathsf{false}, -)$, the only difference being that no repair of $Q_k'$ is enqueued.

  Now, the only case left is when \textsc{FindResidual} returns $(\mathsf{true}, P')$, at which point the loop ends.
\end{proof}

We now prove the three guarantees of Theorem~\ref{thm:optimal-synth-guarantees}.

\begin{enumerate}[leftmargin=*]
  \item \textbf{Termination.} Immediate from the finitude of $U_Q$ and Lemma~\ref{lem:synth-worklist-decreases}.

  \item \textbf{Soundness.} The algorithm returns $(Q^*,P^*)$ only at line~10 of Algorithm~\ref{alg:synthesize}. At this return point:
  \textsc{WeakenViaBounds} returned $(\mathsf{true}, Q^*, \psi_{\min}, \psi_{\max})$ with $\psi_{\max}$ satisfying \init (Theorem~\ref{thm:weaken-correct});
  \textsc{FindStrongestBisimulation} returned $(\mathsf{true}, \psi, -)$, where $\psi$ is the strongest invariant between $\psi_{\min}$ and $\psi_{\max}$ satisfying \sync and \stutter under $Q^*$ (Theorem~\ref{thm:find-strongest-bisim-correct}), and since $\psi \subseteq \psi_{\max}$, $\psi$ also satisfies \init;
  \textsc{FindResidual} returned $(\mathsf{true}, P^*)$ with $P^*$ satisfying \final under $\psi$ (Theorem~\ref{thm:residual-correct}).
  Thus, $\psi$ is a certification witness for $(Q^*,P^*)$ per Definition~\ref{def:cert-witness}, and by Theorem~\ref{thm:sound-complete} (soundness), $(Q^*,P^*)$ satisfies Equation~\eqref{eq:def-pushdown}.

    \item \textbf{Optimality within the search space.} We prove (a) and (b) separately.

  \medskip\noindent\emph{Part~(b): Minimal $P'$ given $Q^*$.}
  At the return point, $\psi$ is the strongest invariant between $\psi_{\min}$ and $\psi_{\max}$ satisfying \sync and \stutter under $Q^*$ (Theorem~\ref{thm:find-strongest-bisim-correct}), and $P^*$ is the weakest residual in $U_{P'}$ satisfying \final under $\psi$ (Theorem~\ref{thm:residual-correct}). Now consider any certification witness $\psi'$ for $(Q^*, P')$ with $P' \subseteq U_{P'}$. For $\psi'$ to certify any $P'$ in conjunction with $Q^*$, it must also certify $(Q^*, P)$: otherwise, it would mean that $Q^*$ is too strong for any residual to establish correctness. Hence, $\psi'$ is a valid invariant under $Q^*$ and $P$. By Theorem~\ref{thm:weaken-correct}, $\psi_{\min} \subseteq \psi' \subseteq \psi_{\max}$. Since $\psi'$ also satisfies \sync and \stutter under $Q^*$, and $\psi$ is the strongest such invariant in this range, $\psi' \subseteq \psi$. Consequently, any $P'$ satisfying \final under $\psi'$ also satisfies \final under $\psi$, and since $P^*$ is the weakest such residual, $P'$ cannot be strictly weaker than $P^*$.

  \medskip\noindent\emph{Part~(a): Maximal $Q$.}
  Define property $\mathcal{P}(W_m)$:
  \begin{quote}
    For every state with $W \preceq W_m$ at the beginning of the loop iteration, if there exist $Q^{\supsharp}$ and $Q$ such that $Q^{\supsharp} \subseteq Q \in W$ and some $\psi \subseteq U_\psi$ is a certification witness for $(Q^{\supsharp}, P)$, then \textsc{SynthesizeOptimalPushdown} returns some $(Q^*, P^*)$ with $Q^{\supsharp} \subseteq Q^*$.
  \end{quote}
  We prove $\mathcal{P}(W)$ for all worklists $W$ by well-founded induction.

  \paragraph{Base case.} In this case, $W = \varnothing$, so $\mathcal{P}(W) = \mathcal{P}(\varnothing)$ vacuously holds.

  \paragraph{Inductive step.} Assume $\mathcal{P}(W_n)$ for all $\varnothing \preceq W_n \prec W_k$ and consider a state with $W_k$ at the beginning of the loop iteration. Let $Q_k = \textsc{Dequeue}(W_k)$ and $W_k'$ be the remainder of $W_k$. There are two cases:

  \begin{enumerate}[leftmargin=*]
    \item \emph{There exists $Q_k^{\supsharp} \subseteq Q_k$ such that some $\psi \subseteq U_\psi$ is a certification witness for $(Q_k^{\supsharp}, P)$.}
      Suppose \textsc{WeakenViaBounds}$(Q_k, P, U_\psi)$ returns $(\mathsf{false}, -, -, -)$. By Theorem~\ref{thm:weaken-correct}, this means $Q_k = \varnothing$ or all atoms in $Q_k$ violate \sync or \stutter under any $\psi \subseteq U_\psi$, either of which contradicts the assumption. Thus, \textsc{WeakenViaBounds} has to return $(\mathsf{true}, Q_k', \psi_{\min}, \psi_{\max})$. By Theorem~\ref{thm:weaken-correct}, we know $Q_k' \subseteq Q_k$, all atoms in $Q_k \setminus Q_k'$ violate \sync or \stutter under any $\psi \subseteq U_\psi$, $\psi_{\max}$ satisfies \init, and $\psi_{\min} \subseteq \psi \subseteq \psi_{\max}$ for any $\psi \subseteq U_\psi$ certifying $(Q_k', P)$. Clearly, atoms in $Q_k \setminus Q_k'$ cannot occur in $Q_k^{\supsharp}$, so we know $Q_k^{\supsharp} \subseteq Q_k' = Q_k \setminus (Q_k \setminus Q_k')$.

      \vspace{0.4em}
      Next, suppose \textsc{FindStrongestBisimulation}$(Q_k', \psi_{\min}, \psi_{\max})$ proceeds to return\\$(\mathsf{false}, -, \mathsf{diagnosis})$. By Theorem~\ref{thm:find-strongest-bisim-correct}, this means that $Q_k'$ violates \sync or \stutter under any invariant between $\psi_{\min}$ and $\psi_{\max}$, with \textsf{diagnosis} as a witness. By definition, \textsc{Repair}$(Q_k', \mathsf{diagnosis})$ proceeds to produce $Q_k'' \subset Q_k'$, where $Q_k' \setminus Q_k''$ are the offending atoms. Since $\psi_{\min} \subseteq \psi \subseteq \psi_{\max}$ for any $\psi \subseteq U_\psi$ certifying $(Q_k', P)$, the atoms in $Q_k' \setminus Q_k''$ cannot occur in $Q_k^{\supsharp}$, meaning that $Q_k^{\supsharp} \subseteq Q_k'' = Q_k' \setminus (Q_k' \setminus Q_k'')$. Since $Q_k''$ is subsequently enqueued before the start of the next loop iteration, we know there exists $Q_k^{\supsharp} \subseteq Q_k'' \in W_{k+1}$ with the same certification witness. Since $W_{k+1} \prec W_k$ by Lemma~\ref{lem:synth-worklist-decreases}, the induction hypothesis $\mathcal{P}(W_{k+1})$ applies, yielding some $(Q^*, P^*)$ with $Q_k^{\supsharp} \subseteq Q^*$, and $\mathcal{P}(W_k)$ follows.

      \vspace{0.4em}
      Otherwise, \textsc{FindStrongestBisimulation}$(Q_k', \psi_{\min}, \psi_{\max})$ returns $(\mathsf{true}, \psi', -)$, where, by Theorem~\ref{thm:find-strongest-bisim-correct}, $\psi'$ is the strongest invariant between $\psi_{\min}$ and $\psi_{\max}$ satisfying \sync and \stutter under $Q_k'$. Since $\psi_{\min} \subseteq \psi \subseteq \psi_{\max}$ for any valid $\psi \subseteq U_\psi$, $\psi'$ is the strongest invariant in $U_\psi$. Since $\psi_{\max}$ satisfies \init, $\psi' \subseteq \psi_{\max}$ also satisfies \init. Thus, at this point, we know $Q_k'$ satisfies \init, \sync, and \stutter under $\psi'$.


      \vspace{0.4em}
      Now, suppose \textsc{FindResidual}$(\psi', U_{P'}, P)$ returns $(\mathsf{false}, -)$, which, by Theorem~\ref{thm:refine-bounds-correct}, means that $\psi'$ violates \final under any $P' \subseteq U_{P'}$. Since $\psi'$ is the strongest invariant in $U_\psi$ satisfying \sync and \stutter under $Q_k'$, and no $P' \subseteq U_{P'}$ satisfies \final under $\psi'$, no certification witness for $(Q_k', \cdot)$ exists in the universe. Since $Q_k^{\supsharp} \subseteq Q_k'$, and $(Q_k^{\supsharp}, P)$ does have a certification witness by assumption, we know $Q_k^{\supsharp} \subset Q_k'$, and since all one-atom weakenings of $Q_k'$ proceed to be enqueued, there exists $Q_k' \setminus \{q\} \supseteq Q_k^{\supsharp}$ in $W_{k+1}$ for some $q \in Q_k' \setminus Q_k^{\supsharp}$. Since $W_{k+1} \prec W_k$ by Lemma~\ref{lem:synth-worklist-decreases}, the induction hypothesis $\mathcal{P}(W_{k+1})$ applies, yielding some $(Q^*, P^*)$ with $Q_k^{\supsharp} \subseteq Q^*$, and $\mathcal{P}(W_k)$ follows.

      \vspace{0.4em}
      Otherwise, \textsc{FindResidual} returns $(\mathsf{true}, P')$, where, by Theorem~\ref{thm:residual-correct}, $P'$ is the weakest residual in $U_{P'}$ satisfying \final under $\psi'$.
      Thus, $(Q_k', P')$ is a certified solution witnessed by $\psi'$. Since $Q_k^{\supsharp} \subseteq Q_k'$ and the algorithm returns $(Q^*, P^*) = (Q_k', P')$, we have $Q_k^{\supsharp} \subseteq Q^*$, and $\mathcal{P}(W_k)$ follows.
    \item \emph{No such $Q_k^{\supsharp} \subseteq Q_k$ exists.}
      By assumption, the certifiable $Q^{\supsharp}$ is a subset of some $Q \in W_k'$. Since we only enqueue onto $W_k'$ throughout the rest of the iteration, $Q^{\supsharp} \subseteq Q \in W_{k+1}$. Since $W_{k+1} \prec W_k$ by Lemma~\ref{lem:synth-worklist-decreases}, the induction hypothesis $\mathcal{P}(W_{k+1})$ applies, yielding some $(Q^*, P^*)$ with $Q^{\supsharp} \subseteq Q^*$, and $\mathcal{P}(W_k)$ follows.
  \end{enumerate}

  \paragraph{Conclusion.} Before the while loop begins, $W = \{U_Q\}$. For any $Q \subseteq U_Q$, $P' \subseteq U_{P'}$, and $\psi \subseteq U_\psi$ such that $\psi$ certifies $(Q, P')$: By the reasoning of Part~(b), $\psi$ must also certify $(Q, P)$, so $Q$ serves as $Q^{\supsharp}$ with $Q^{\supsharp} \subseteq U_Q \in W$. By $\mathcal{P}(\{U_Q\})$, the algorithm returns $(Q^*, P^*)$ with $Q \subseteq Q^*$. Hence no certifiable $Q$ is strictly stronger than $Q^*$.
\end{enumerate} \qed

\begin{theorem}[Correctness of \textsc{WeakenViaBounds}]\label{thm:weaken-correct}
Let $Q$ be a candidate input filter, $P$ be a post-UDF filter, and $U_\psi$ be a finite invariant universe. \textsc{WeakenViaBounds}$(Q, P, U_\psi)$ always terminates and returns either
\begin{enumerate}[label=(\alph*),leftmargin=*]
\item $(\textsf{true}, Q', \psi_{\min}, \psi_{\max})$, where $Q' \subseteq Q$, all atoms in $Q \setminus Q'$ violate \sync or \stutter under any $\psi \subseteq U_\psi$, $\psi_{\max}$ satisfies \init, and $\psi_{\min} \subseteq \psi \subseteq \psi_{\max}$ for any $\psi \subseteq U_\psi$ certifying $(Q', P)$, or
\item $(\textsf{false},-,-,-)$, in which case either $Q = \varnothing$ or all atoms in $Q$ violate \sync or \stutter under any $\psi \subseteq U_\psi$.
\end{enumerate}
\end{theorem}

\begin{proof}

We split the argument into termination and partial correctness.

\begin{enumerate}[leftmargin=*]
  \item \textbf{Termination.} The first line initializing $\psi_{\min}$ and $\psi_{\max}$ clearly terminates. Next, we prove that the loop in \textsc{WeakenViaBounds} terminates by showing that $\lvert Q \rvert$ strictly decreases through it. Arbitrarily consider the $k$-th iteration, and let $Q_k$ be the candidate input filter at the beginning. If $Q_k = \varnothing$, then \textsc{WeakenViaBounds} finishes. Otherwise, \textsc{RefineBounds}$(Q_k, P, U_\psi, \psi_{\min}, \psi_{\max})$ proceeds to terminate (by Theorem~\ref{thm:refine-bounds-correct}) and return potentially tightened invariant bounds $(\psi_{\min}', \psi_{\max}')$, followed by \textsc{CheckUnrealizable}$(Q_k, \psi_{\min}', \psi_{\max}')$ terminating (by Theorem~\ref{thm:check-unrealizable-correct}) and deciding whether $Q_k$ is unrealizable within the bounds. If \textsc{CheckUnrealizable} returns \textsf{false}, then \textsc{WeakenViaBounds} finishes; otherwise, it proceeds to repair $Q_k$. By definition, \textsc{Repair} terminates and returns $Q_k' \subset Q_k$ and the iteration ends. Thus, $\lvert Q_{k+1} \rvert = \lvert Q_k' \rvert < \lvert Q_k \rvert$. Thus, by the finitude of $Q$ at the beginning of the loop, the loop terminates, and it immediately follows that \textsc{WeakenViaBounds} terminates.

  \item \textbf{Partial correctness.} We first define property $\mathcal{P}(Q_m)$:
  \begin{quote}
    For every state with $Q_n$ and $(\psi_{\min}, \psi_{\max})$ at the beginning of the loop iteration where $\lvert Q_n \rvert \le \lvert Q_m \rvert$, if all atoms in $Q \setminus Q_n$ violate \sync or \stutter under any $\psi \subseteq U_\psi$,
    $\psi_{\min}$ is either $\varnothing$ or the weakest implicant of
    \[
      \phi
      \;=\;
      (P(a_1) \land P(a_2) \land a_1 = a_2) \lor (\neg P(a_1) \land \neg P(a_2))
    \]
    over $U_\psi$, and $\psi_{\min} \subseteq \psi \subseteq \psi_{\max} \subseteq U_\psi$ for any $\psi \subseteq U_\psi$ certifying $(Q_n, P)$, then \textsc{WeakenViaBounds} returns either $(\textsf{true},\,Q_n',\psi_{\min}', \psi_{\max}')$, where
    \begin{enumerate}[leftmargin=*]
      \item $Q_n' \subseteq Q_n$,
      \item all atoms in $Q \setminus Q_n'$ violate \sync or \stutter under any $\psi \subseteq U_\psi$,
      \item $\psi_{\max}'$ satisfies \init, and
      \item $\psi_{\min} \subseteq \psi_{\min}' \subseteq \psi \subseteq \psi_{\max}' \subseteq \psi_{\max}$ for any $\psi \subseteq U_\psi$ certifying $(Q_n', P)$,
    \end{enumerate}
    or $(\textsf{false},-,-,-)$, in which case either $Q = \varnothing$ or all atoms in $Q$ violate \sync or \stutter under any $\psi \subseteq U_\psi$.
    
  \end{quote}
  
  We prove $\mathcal{P}(Q)$ for all candidate input filters by induction on $\lvert Q \rvert$.

  \paragraph{Base case.} In this case, $Q = \varnothing$, and the loop immediately returns $(\mathsf{false}, -, -, -)$, so $\mathcal{P}(Q) = \mathcal{P}(\varnothing)$ holds.

  \paragraph{Inductive step.} Assume $\mathcal{P}(Q_n)$ for all $0 \le \lvert Q_n \rvert < m$ and consider a state with $Q_k$ and $(\psi_{\min}, \psi_{\max})$ where $\lvert Q_k \rvert = m$ at the beginning of the loop iteration. By assumption,
  $\psi_{\min}$ is either $\varnothing$ or the weakest implicant of $\phi$ over $U_\psi$, and $\psi_{\min} \subseteq \psi \subseteq \psi_{\max} \subseteq U_\psi$ for any $\psi \subseteq U_\psi$ certifying $(Q_k, P)$. Let
  \[
    (\psi_{\min}', \psi_{\max}') = \textsc{RefineBounds}(Q_k, P, U_\psi, \psi_{\min}, \psi_{\max}),
  \]
  and by Theorem~\ref{thm:refine-bounds-correct}, we know $\psi_{\min} \subseteq \psi_{\min}' \subseteq \psi \subseteq \psi_{\max}' \subseteq \psi_{\max} \subseteq U_\psi$ for any $\psi$ where $\psi_{\min} \subseteq \psi \subseteq \psi_{\max}$ that certifies $(Q_k, P)$, $\psi_{\min}'$ is the weakest implicant of $\phi$ over $U_\psi$, and $\psi_{\max}'$ satisfies \init. Since any $\psi \subseteq U_\psi$ certifying $(Q_k, P)$ must be between $\psi_{\min}$ and $\psi_{\max}$, and any such $\psi$ between $\psi_{\min}$ and $\psi_{\max}$ must be between $\psi_{\min}'$ and $\psi_{\max}'$, we know that any $\psi \subseteq U_\psi$ certifying $(Q_k, P)$ must be between $\psi_{\min}'$ and $\psi_{\max}'$. Now, we analyze the two possible outcomes of $\textsc{CheckUnrealizable}(Q_k, \psi_{\min}', \psi_{\max}')$:

  \begin{itemize}[leftmargin=*]
    \item \emph{(\textsf{false}, -).} In this case, \textsc{WeakenViaBounds} returns $(\mathsf{true}, Q_k, \psi_{\min}', \psi_{\max}')$. Since we also know that $Q_k \subseteq Q_k$, all atoms in $Q \setminus Q_k$ violate \sync or \stutter under any $\psi \subseteq U_\psi$, and $\psi_{\max}'$ satisfies \init, $\mathcal{P}(Q_k)$ holds.
    


    \item \emph{(\textsf{true}, \textsf{diagnosis}).} By Theorem~\ref{thm:check-unrealizable-correct}, $Q_k$ violates \sync or \stutter under any $\psi$ between $\psi_{\min}'$ and $\psi_{\max}'$, and \textsf{diagnosis} witnesses the failure. Then, by definition, \textsc{Repair}$(Q_k, \mathsf{diagnosis})$ returns $Q_k' \subset Q_k$ where $Q_k \setminus Q_k'$ are all the offending atoms according to \textsf{diagnosis}. By assumption, all atoms in $Q \setminus Q_k$ violate \sync or \stutter under any $\psi \subseteq U_\psi$, so we know that all atoms in $Q \setminus Q_k' = (Q \setminus Q_k) \cup (Q_k \setminus Q_k')$ also do.

    \vspace{0.4em}
    Let $\psi$ be any $\psi \subseteq U_\psi$ certifying $(Q_k', P)$. Since $\psi_{\min}'$ is the weakest implicant of $\phi$ over $U_\psi$, $\psi_{\min}' \subseteq \psi$. Since $\psi \subseteq \psi_{\max}' \subseteq U_\psi$ for any $\psi \subseteq U_\psi$ certifying $(Q_k, P)$, and now that $Q_k'$ is weaker than $Q_k$, by the structure of \sync, it is still the case for any $\psi \subseteq U_\psi$ certifying $(Q_k', P)$. Thus, $\psi_{\min}' \subseteq \psi \subseteq \psi_{\max}' \subseteq U_\psi$ for any $\psi \subseteq U_\psi$ certifying $(Q_k', P)$.

    \vspace{0.4em}
    Lastly, $\lvert Q_k' \rvert < \lvert Q_k \rvert$, so the induction hypothesis $\mathcal{P}(Q_k')$ applies, and $\mathcal{P}(Q_k)$ follows.


  \end{itemize}

  \paragraph{Conclusion.} Before the while loop begins, $\psi_{\min} = \varnothing$ and $\psi_{\max} = U_\psi$.
  Trivially, $\psi_{\min} = \varnothing \subseteq \psi \subseteq U_\psi = \psi_{\max}$ for any $\psi \subseteq U_\psi$ certifying $(Q, P)$. And vacuously, all atoms in $Q \setminus Q = \varnothing$ violate \sync or \stutter under any $\psi \subseteq U_\psi$. Thus, $\mathcal{P}(Q)$ applies, and the theorem follows.
\end{enumerate}
\end{proof}

\begin{theorem}[Correctness of \textsc{RefineBounds}]\label{thm:refine-bounds-correct}
  Let $Q$ be a candidate input filter, $P$ be a post-UDF filter, $U_\psi$ be the finite invariant universe, and $(\psi_{\min}, \psi_{\max})$ be symbolic bounds.\\\textsc{RefineBounds}$(Q, P, U_\bisim, \psi_{\min}, \psi_{\max})$ always terminates and returns $(\psi_{\min}', \psi_{\max}')$ s.t.
  \begin{enumerate}[label=(\alph*),leftmargin=*]
    \item $\psi_{\min}'$ is the weakest implicant of
    \[
      \phi
      \;=\;
      (P(a_1) \land P(a_2) \land a_1 = a_2) \lor (\neg P(a_1) \land \neg P(a_2))
    \]
    over $U_\psi$ if $\psi_{\min} = \varnothing$ and $\psi_{\min}' = \psi_{\min}$ otherwise; $\psi_{\max}' \subseteq \psi_{\max}$;
    \item $\psi'_{\min} \subseteq \psi \subseteq \psi'_{\max}$ for any $\psi$ where $\psi_{\min} \subseteq \psi \subseteq \psi_{\max}$ that certifies $(Q, P)$; and
    \item $\psi_{\max}'$ satisfies \init if $\psi_{\min}'$ is the weakest implicant of $\phi$ over $U_\psi$.
  \end{enumerate}
\end{theorem}

\begin{proof}

By the assumption that $U_\psi$ contains all necessary atoms, \textsc{FindWeakestImplicant} terminates. Since $\psi_{\max}$ is finite and the two loops only remove atoms from it, they terminate. Thus, \textsc{RefineBounds} terminates. For partial correctness, we first consider the lower bound case:
\begin{enumerate}[leftmargin=*]
  \item \emph{$\psi_{\min} = \varnothing$.}
  In this case, $\psi_{\min}' = \textsc{FindWeakestImplicant}(U_\psi, \phi)$.
  By Lemma~\ref{lem:exist} and the assumption that $U_\psi$ contains all necessary atoms, we know $\psi_{\min}' \neq \varnothing$, so $\psi_{\min} \subseteq \psi_{\min}'$. For any $\psi$ certifying $(Q, P)$, $\psi_{\min} = \varnothing \subset \psi$ and $\psi \Rightarrow \phi$. By definition, \textsc{FindWeakestImplicant} returns the intersection of all subsets of $U_{\psi}$ that entail $\phi$, so it must be contained in $\psi$. Hence,
  \[
    \psi'_{\min}
    \;=\;
    \bigcap\bigl\{\,S\subseteq U_{\psi}\mid S\models\phi\bigr\}
    \;\subseteq\;
    \psi.
  \]

  \item \emph{$\psi_{\min} \neq \varnothing$.}
  In this case, $\psi_{\min}' = \psi_{\min}$, so $\psi_{\min} \subseteq \psi_{\min}'$ and, for any $\psi$ where $\psi_{\min} \subseteq \psi$ that certifies $(Q, P)$, $\psi_{\min}' \subseteq \psi$.
\end{enumerate}

In the upper bound case, $\psi_{\max}'$ is formed by removing from $\psi_{\max}$ any atom that violates \init or \sync under $Q$, so $\psi'_{\max} \subseteq \psi_{\max}$. For any $\psi \subseteq \psi_{\max}$ that certifies $(Q, P)$, $\psi$ satisfies both \init and \sync and thus does not contain the atoms in $\psi_{\max} \setminus \psi_{\max}'$, so $\psi \subseteq \psi'_{\max}$.

\vspace{0.4em}
Suppose $\psi_{\min}'$ is the weakest implicant of $\phi$ over $U_\psi$. By Lemma~\ref{lem:prune}, $\phi \Rightarrow \psi_{\min}'$. Suppose for contradiction that $\neg \psi_{\min}'(I, I)$, which means that
\[
  \phi(I,I)
  \;=\;
  (P(I) \land P(I) \land I = I) \lor (\neg P(I) \land \neg P(I))
\]
is invalid. However, it is clearly a tautology: if $P(I)$ holds, then the left disjunct holds; if it does not, then the right disjunct holds. Thus, we have a contradiction, so we know that $\psi_{\min}'(I, I)$ has to hold, and there is no need to check whether atoms in $\psi_{\max}$ that are also $\psi_{\min}'$ violate \init. Hence, $\psi_{\max}'$ satisfies \init.
\end{proof}

\begin{theorem}[{Correctness} of \textsc{CheckUnrealizable}]\label{thm:check-unrealizable-correct}
Let $Q$ be a candidate input filter and $(\psi_{\min},\allowbreak\psi_{\max})$ be symbolic bounds. \textsc{CheckUnrealizable}$(Q, \psi_{\min}, \psi_{\max})$ always terminates, and
\begin{enumerate}[label=(\alph*),leftmargin=*]
  \item if it returns $(\textsf{true}, \sync(r))$, then $Q$ violates \sync under any $\psi$ between $\psi_{\min}$ and $\psi_{\max}$ on row $r$;
  \item if it returns $(\textsf{true}, \stutter(r))$, then $Q$ violates \stutter under any $\psi$ between $\psi_{\min}$ and $\psi_{\max}$ on $r$.
\end{enumerate}
\end{theorem}

\begin{proof}

Termination of \textsc{CheckUnrealizable} is immediate as it only involves straight-line code. We consider each failure mode in turn:
\begin{enumerate}[label=(\alph*),leftmargin=*]
  \item \emph{\sync failure.} In this case, the check
  \[
    \chi_{\mathsf{sync}}\;:\; 
    \psi_{\max}(a_1, a_2) \land Q(r) \land a_1' = f(a_1, r) \land a_2' = f(a_2, r)
    \;\Longrightarrow\;
    \psi_{\min}(a_1', a_2')
  \]
  is invalid, i.e., there exists a model $(a_1, a_2, r)$ for which $\psi_{\max}(a_1, a_2)$ and $Q(r)$ hold but $\psi_{\min}(a_1', a_2')$ fails. Now, for any $\psi$ where $\psi_{\min} \subseteq \psi \subseteq \psi_{\max}$, $\psi(a_1, a_2)$ holds since $\psi_{\max}(a_1, a_2)$ does. However, since $\psi(a_1',a_2')$ implies $\psi_{\min}(a_1',a_2')$, we know $\psi(a_1', a_2')$ fails, so $Q$ violates \sync under such $\psi$ on row $r$.

  \item \emph{\stutter failure.} In this case, the check
  \[
    \chi_{\mathsf{stutter}}\;:\;
    \psi_{\max}(a_1, a_2) \land \neg Q(r) \land a_1' = f(a_1, r) \land a_2' = a_2
    \;\Longrightarrow\;
    \psi_{\min}(a_1', a_2')
  \]
  is invalid for some $(a_1, a_2, r)$. By the same reasoning, $Q$ violates \stutter under any $\psi$ where $\psi_{\min} \subseteq \psi \subseteq \psi_{\max}$ on row $r$.
\end{enumerate}
\end{proof}

\begin{theorem}[Correctness of \textsc{FindStrongestBisimulation}]\label{thm:find-strongest-bisim-correct}
Let $Q$ be a candidate input filter and $(\psi_{\min},\psi_{\max})$ be symbolic bounds. \textsc{FindStrongestBisimulation}$(Q,\psi_{\min},\psi_{\max})$ always terminates and returns either
\begin{enumerate}[label=(\alph*),leftmargin=*]
\item $(\textsf{true}, \psi, -)$, where $\psi$ is the strongest invariant between $\psi_{\min}$ and $\psi_{\max}$ satisfying \sync and \stutter under $Q$, or
\item $(\textsf{false}, -, \textsf{diagnosis})$, in which case $Q$ violates \sync or \stutter under any $\psi$ between $\psi_{\min}$ and $\psi_{\max}$, and \textsf{diagnosis} witnesses this failure.
\end{enumerate}
\end{theorem}

\begin{proof}

We split the argument into termination and partial correctness.

\begin{enumerate}[leftmargin=*]
  \item \textbf{Termination.} The first line initializing $\psi_\text{cand}$ to $\psi_{\max}$ clearly terminates. Next, we prove that the loop in \textsc{FindStrongestBisimulation} terminates by showing that $\lvert \psi_\text{cand} \rvert$ strictly decreases. Consider an arbitrary iteration of the loop. \textsc{WeakenViaVC} terminates as it loops over each atom in the finite $\psi_\text{cand}$. Suppose either \textsc{WeakenViaVC} call returns $(\mathsf{false}, -, r)$, then \textsc{FindStrongestBisimulation} finishes. Suppose both \textsc{WeakenViaVC} call return $(\mathsf{true}, \psi_\text{cand}, -)$, then the loop ends as $\psi_\text{cand}$ has converged. Now, the only remaining case is if either call returns $(\mathsf{true}, \psi_\text{cand}', -)$ where $\psi_\text{cand}' \subset \psi_\text{cand}$, in which case $\lvert \psi_\text{cand}' \rvert < \lvert \psi_\text{cand} \rvert$. Thus, by the finitude of $\psi_\text{cand} = \psi_{\max}$ at beginning of the loop, the loop terminates. After the loop, the final line of the \textsc{FindStrongestBisimulation} returns from it, so the algorithm terminates.

  \item \textbf{Partial correctness.} We first define property $\mathcal{P}(\psi_m)$:
  \begin{quote}
    For every state with $\psi_\text{cand}$ where $\lvert \psi_\text{cand} \rvert \le \lvert \psi_m \rvert$ at the beginning of the loop iteration, if $\psi_{\min} \subseteq \psi_\text{cand} \subseteq \psi_{\max}$ and all atoms in $\psi_{\max} \setminus \psi_\text{cand}$ violate \sync or \stutter under $Q$, \textsc{FindStrongestBisimulation}$(Q, \psi_{\min}, \psi_{\max})$ returns either $(\mathsf{true}, \psi_\text{cand}', -)$, where $\psi_{\min} \subseteq \psi_\text{cand}' \subseteq \psi_\text{cand}$, $\psi_\text{cand}'$ satisfies \sync and \stutter under $Q$, and all atoms in $\psi_{\max} \setminus \psi_\text{cand}'$ violate either of them, or $(\mathsf{false}, -, \mathsf{diagnosis})$, in which case $Q$ violates \sync or \stutter under any $\psi$ between $\psi_{\min}$ and $\psi_{\max}$, and \textsf{diagnosis} witnesses this failure.
  \end{quote}
  We prove $\mathcal{P}(\psi_\text{cand})$ for all invariants by induction on $\lvert \psi_\text{cand} \rvert$.

  \paragraph{Base case.} In this case, $\psi_\text{cand} = \varnothing \subset \psi_{\min}$, so $\mathcal{P}(\psi_\text{cand}) = \mathcal{P}(\varnothing)$ vacuously holds.
  

  \paragraph{Inductive step.} Assume $\mathcal{P}(\psi_n)$ for all $0 \le \lvert \psi_n \rvert < m$ and consider a state with $\psi_\text{cand}$ where $\lvert \psi_\text{cand} \rvert = m$ at the beginning of the loop iteration. By assumption, we know $\psi_{\min} \subseteq \psi_\text{cand} \subseteq \psi_{\max}$ and all atoms in $\psi_{\max} \setminus \psi_\text{cand}$ violate \sync or \stutter under $Q$. We go over the three possible scenarios:
  \begin{itemize}[leftmargin=*]
    \item \emph{Either \textsc{WeakenViaVC} call returns $(\textsf{false}, -, r)$.} In this case, by definition of \textsc{WeakenViaVC}, some atom $p$ in both $\psi_\text{cand}$ and $\psi_{\min}$ violates \sync or \stutter under $Q$, so we know that $Q$ violates \sync or \stutter under any $\psi$ between $\psi_{\min}$ and $\psi_{\max}$. \textsc{FindStrongestBisimulation} proceeds to return $(\mathsf{false}, -, \mathsf{diagnosis})$, where \textsf{diagnosis} witnesses the failure, and $\mathcal{P}(\psi_\text{cand})$ follows.
    
    \item \emph{Both \textsc{WeakenViaVC} calls return $(\textsf{true}, \psi_\text{cand}, -)$.} By definition of \textsc{WeakenViaVC}, this means that no atom in $\psi_\text{cand}$ violates \sync or \stutter under $Q$. Thus, $\psi_\text{cand}$ converges, the loop ends, and \text{FindStrongestBisimulation} returns $(\mathsf{true}, \psi_\text{cand}, -)$, where $\psi_{\min} \subseteq \psi_\text{cand} \subseteq \psi_\text{cand}$, all atoms in $\psi_\text{cand}$ satisfies \sync and \stutter under $Q$, and all atoms in $\psi_{\max} \setminus \psi_\text{cand}$ violate either of them. Thus, $\mathcal{P}(\psi_\text{cand})$ holds.

    \item \emph{Either \textsc{WeakenViaVC} call returns $(\textsf{true}, \psi_\text{cand}', -)$ where $\psi_\text{cand}' \subset \psi_\text{cand}$.} In this case, by definition of \textsc{WeakenViaVC} and $\psi_\text{cand} \subseteq \psi_{\max}$, $\psi_{\min} \subseteq \psi_\text{cand}' \subseteq \psi_{\max}$ and all atoms in $\psi_\text{cand} \setminus \psi_\text{cand}'$ violate \sync or \stutter under $Q$. Since all atoms in $\psi_{\max} \setminus \psi_\text{cand}$ also do, we know all atoms in $\psi_{\max} \setminus \psi_\text{cand}' = (\psi_{\max} \setminus \psi_\text{cand}) \cup (\psi_\text{cand} \setminus \psi_\text{cand}')$ do, too. And since $\lvert \psi_\text{cand}' \rvert < \lvert \psi_\text{cand} \rvert$, the induction hypothesis $\mathcal{P}(\psi_\text{cand}')$ applies, and $\mathcal{P}(\psi_\text{cand})$ follows.
    

  \end{itemize}

  \paragraph{Conclusion.} Before the loop begins, $\psi_\text{cand} = \psi_{\max}$. Clearly, $\psi_{\min} \subseteq \psi_{\max} \subseteq \psi_{\max}$ and all atoms in $\psi_{\max} \setminus \psi_{\max} = \varnothing$ violate \sync or \stutter under $Q$. Thus, $\mathcal{P}(\psi_\text{cand}) = \mathcal{P}(\psi_{\max})$ applies. Suppose \textsc{FindStrongestBisimulation} returns $(\mathsf{true}, \psi_\text{cand}', -)$ where $\psi_{\min} \subseteq \psi_\text{cand}' \subseteq \psi_\text{cand} = \psi_{\max}$, $\psi_\text{cand}'$ satisfies \sync and \stutter under $Q$, and all atoms in $\psi_{\max} \setminus \psi_\text{cand}'$ violate \sync or \stutter under $Q$. Suppose for contradiction that there exists an invariant $\psi_\text{cand}''$ s.t. $\psi_{\min} \subseteq \psi_\text{cand}' \subset \psi_\text{cand}'' \subseteq \psi_{\max}$ (i.e., $\psi_\text{cand}'' \Rightarrow \psi_\text{cand}'$) that satisfies \sync and \stutter under $Q$. This means that some atom $p$ in $\psi_{\max} \setminus \psi_\text{cand}'$ also exists in $\psi_\text{cand}''$. However, since all atoms in $\psi_{\max} \setminus \psi_\text{cand}'$ violate \sync or \stutter under $Q$, $\psi_\text{cand}''$ cannot satisfy \sync and \stutter under $Q$, so we have a contradiction. Therefore, in this case, $\psi_\text{cand}'$ is the strongest invariant between $\psi_{\min}$ and $\psi_{\max}$ satisfying \sync and \stutter under $Q$, and the theorem follows.
\end{enumerate}
\end{proof}

\begin{theorem}[Correctness of \textsc{FindResidual}]\label{thm:residual-correct}
  Let $\psi$ be a bisimulation invariant, $U_{P'}$ be the finite residual universe, and $P$ be a post-UDF filter. \textsc{FindResidual}$(\psi, U_{P'}, P)$ always terminates and returns either
  \begin{enumerate}[label=(\alph*),leftmargin=*]
    \item $(\textsf{true}, P')$, where $P'$ is the weakest residual in $U_{P'}$ satisfying \final under $\psi$, or
    \item $(\textsf{false}, -)$, in which case $\psi$ violates \final under any $P' \subseteq U_{P'}$.
  \end{enumerate}
\end{theorem}

\begin{proof}

We begin by showing that worklist $W$ strictly decreases through the while loop in \textsc{FindResidual}. We refer to candidate residuals in a worklist with the least number of atoms as the \emph{smallest}.
Let $\sigma(W) = \max_{P' \in W}{\lvert U_{P'} \setminus P' \rvert}$ and $\kappa(W) = \sum_{P' \in W} \mathds{1}_{\{P'' \mid \sigma(W) = \lvert U_{P'} \setminus P'' \rvert\}}(P')$. We then define a standard lexicographic order on worklists: for worklists $W_a$ and $W_b$,
\[
  W_a \prec W_b
  \;\Longleftrightarrow\;
  \sigma(W_a) < \sigma(W_b) \lor
  (\sigma(W_a) = \sigma(W_b) \land \kappa(W_a) < \kappa(W_b)).
\]
$\preceq$ is defined similarly. We prove that the worklist strictly decreases according to $\prec$ via the following lemma:

\begin{lemma}[Worklist Strictly Decreases]\label{lem:residual-worklist-decreases}
  Let $W_k$ be the worklist at the beginning of the $k$-th iteration of the while loop. For all $k \ge 0$, $W_{k+1} \prec W_k$.
\end{lemma}

\begin{proof}
  For any $k$, let $P_k' = \textsc{Dequeue}(W_k)$ and $W_k'$ be the remainder of $W_k$. Since the worklist is a min-heap, $P_k'$ is (one of) the smallest, i.e., $\lvert U_{P'} \setminus P_k' \rvert = \sigma(W_k)$. If $P_k' \in \mathsf{Visited}$, then the iteration ends with $W_{k+1} = W_k'$. If $P_k'$ is the sole smallest in $W_k$, we have $\sigma(W_{k+1}) < \sigma(W_k)$; otherwise, we have $\sigma(W_{k+1}) = \sigma(W_k)$ yet $\kappa(W_{k+1}) < \kappa(W_k)$. Thus, in either case, $W_{k+1} \prec W_k$.

  \vspace{0.4em}
  Moving on, if $\chi$ is valid, then \textsc{FindResidual} finishes. Otherwise, the iteration proceeds to call \textsc{HandleFailure}, which, by the finitude of $U_{P'}$, terminates. If \textsc{HandleFailure} returns without enqueuing any candidate residual onto $W_k'$, the iteration ends with $W_{k+1} = W_k'$, and as before, we know $W_{k+1} \prec W_k$. Otherwise, it enqueues $W_k'' = \{ P_k' \land p \mid p \in U_{P'} \setminus P_k' \}$. Clearly, $\sigma(W_k'') < \lvert U_{P'} \setminus P_k' \rvert = \sigma(W_k)$. We consider the two cases:
  \begin{itemize}[leftmargin=*]
    \item \emph{$P_k'$ is the sole smallest in $W_k$.} In this case, $\sigma(W_k') < \lvert U_{P'} \setminus P_k' \rvert = \sigma(W_k)$. Thus, we have $\sigma(W_{k+1}) = \sigma(W_k' \cup W_k'') < \sigma(W_k)$, which means $W_{k+1} \prec W_k$.
    \item \emph{$P_k'$ is not the sole smallest in $W_k$.} That is, there is some $P' \in W_k'$ s.t. $\sigma(W_k') = \lvert U_{P'} \setminus P' \rvert = \lvert U_{P'} \setminus P_k' \rvert = \sigma(W_k)$. Since $\sigma(W_k'') < \sigma(W_k)$, we know $\sigma(W_{k+1}) = \sigma(W_k' \cup W_k'') = \sigma(W_k') = \sigma(W_k)$. Since $P_k' \notin W_{k+1}$, there is one less smallest residual in $W_{k+1}$, i.e., $\kappa(W_{k+1}) < \kappa(W_k)$. Thus, $W_{k+1} \prec W_k$.
  \end{itemize}
\end{proof}

We split Theorem~\ref{thm:residual-correct} into termination and partial correctness.

\begin{enumerate}[leftmargin=*]
  \item \textbf{Termination.} The straight-line code leading up to the while loop clearly terminates. By the finitude of $U_{P'}$ and Lemma~\ref{lem:residual-worklist-decreases}, the while loop terminates, and it follows that \textsc{FindResidual} terminates.

  \item \textbf{Partial correctness.} First, the early rejection check (lines 3--4) amounts to instantiating $P'$ in the correctness requirement (Equation~\eqref{eq:unified-lifted}) with $P$, i.e.,
  \[
    \forall x.\quad\Lift(P,\,F(x))\;=\;\Lift(P,\,F(\filter_Q(x))).
  \]
  If this check fails---i.e., there exists an $x$ such that filtering by $Q$ changes the lifted outcome---then $Q$ is too strong and no residual can succeed under $\psi$. Thus, \textsc{FindResidual} returns \textsf{false} and $\psi$ violates \final under any $P' \subseteq U_{P'}$.
  Otherwise, \textsc{FindResidual} proceeds. We define property $\mathcal{P}(W_m)$:
  \begin{quote}
    For every state with $W \preceq W_m$ at the beginning of the while loop iteration, if $P' \subseteq U_{P'}$ for all $P' \in W$, and there exist $P_1' \in W$ and $P_2'\subseteq U_{P'}$ s.t. $P_1' \notin \mathsf{Visited}$ and $P_1' \cup P_2'$ satisfies \final under $\psi$, then \textsc{FindResidual}$(\psi, U_{P'}, P)$ returns $(\mathsf{true}, P')$, where $P'$ is the weakest such valid residual in $U_{P'}$.
  \end{quote}
  We prove $\mathcal{P}(W)$ for all worklists $W$ by well-founded induction.

  \paragraph{Base case.} In this case, $W = \varnothing$, so $\mathcal{P}(W) = \mathcal{P}(\varnothing)$ vacuously holds.

  \paragraph{Inductive step.} Assume $\mathcal{P}(W_n)$ for all $\varnothing \preceq W_n \prec W_k$ and consider a state with $W_k$ at the beginning of the loop iteration. Let $P_k' = \textsc{Dequeue}(W_k)$ and $W_k'$ be the remainder of $W_k$. By assumption, we know $P' \subseteq U_{P'}$ for all $P' \in W_k' \subset W_k$.

  \vspace{0.4em}
  If $P_k' \in \mathsf{Visited}$, then the iteration ends with $W_{k+1} = W_k'$. By assumption, there must exist $P_1' \in W_k' = W_{k+1}$ and $P_2' \subseteq U_{P'}$ s.t. $P_1' \notin \mathsf{Visited}$ and $P_1' \cup P_2'$ satisfies \final under $\psi$. And since $W_{k+1} \prec W_k$ by Lemma~\ref{lem:residual-worklist-decreases}, the induction hypothesis $\mathcal{P}(W_{k+1})$ applies, and $\mathcal{P}(W_k)$ follows.

  \vspace{0.4em}
  If $\chi$ (\final) holds, then \textsc{FindResidual} returns $(\mathsf{true}, P_k')$. In this case, $P_1' = P_k' \notin \mathsf{Visited}$ and $P_2' = \varnothing \subseteq U_{P'}$. By assumption, $P' \subseteq U_{P'}$ for all $P' \in W_k$, so $P_k' \subseteq U_{P'}$. And since the worklist is a min-heap, $P_k'$ is the weakest residual in $U_{P'}$ satisfying \final under $\psi$. Thus, $\mathcal{P}(W_k)$ holds.

  \vspace{0.4em}
  At this point, we know that $P_k'$ violates \final under $\psi$, i.e., there exists accumulator states $(a_1, a_2)$ s.t. $\psi(a_1, a_2)$ holds yet the consequent fails. We analyze the three possible failure modes that \textsc{HandleFailure} examines:
  \begin{itemize}[leftmargin=*]
    \item \emph{Mismatch on acceptance.} In this case, $P(a_1)$ and $P_k'(a_2)$ hold while $a_1 \neq a_2$, and \textsc{HandleFailure} discards $P_k'$ without enqueuing anything onto $W_k'$. Since further strengthening $P_k'$ can only flip $P'_k(a_2)$ false, in which case the consequent remains to fail, we know there does not exist any $P_2' \subseteq U_{P'}$ s.t. $P_k' \cup P_2'$ satisfies \final under $\psi$. Then by assumption, there must be a $P_1' \in W_k' = W_{k+1}$ for which such a $P_2'$ exists. And since $W_{k+1} \prec W_k$ by Lemma~\ref{lem:residual-worklist-decreases}, the induction hypothesis $\mathcal{P}(W_{k+1})$ applies, and $\mathcal{P}(W_k)$ follows.

    \item \emph{False rejection.} In this case, $P(a_1)$ holds and $P_k'(a_2)$ fails. Since further strengthening $P_k'$ would not flip $P_k'(a_2)$ true, we know there does not exist any $P_2' \subseteq U_{P'}$ s.t. $P_k' \cup P_2'$ satisfies \final under $\psi$, as in the previous case. Then by the same reasoning, $\mathcal{P}(W_k)$ holds.

    \item \emph{Spurious acceptance.} In this case, $P(a_1)$ fails while $P_k'(a_2)$ holds, meaning that there is a chance of strengthening $P_k'$ for it to eventually satisfy \final under $\psi$. \textsc{HandleFailure} proceeds to enqueue all one-atom strengthenings of $P_k'$ that reject $a_2$.
      
    \vspace{0.4em}
    Suppose there exists $P_2' \subseteq U_{P'}$ s.t. $P_k' \cup P_2'$ satisfies \final under $\psi$. Since $P_k'(a_2)$ holds, we know $P_2'$ must reject $a_2$ by containing some atom $p \in U_{P'} \setminus P_k'$ where $\neg p(a_2)$. By definition, \textsc{HandleFailure} enqueues $P_k' \land p$, so there exist $P_k' \land p \in W_{k+1}$ and $P_2' \setminus \{p\} \subseteq U_{P'}$ s.t. $(P_k' \land p) \cup (P_2' \setminus \{p\}) = P_k' \cup P_2'$ satisfies \final under $\psi$.
    Since $W_{k+1} = W_k' \cup \{P_k' \land p \mid p \in U_{P'} \setminus P_k'\}$ and $P' \subseteq U_{P'}$ for all $P' \in W_k'$, we have $P' \subseteq U_{P'}$ for all $P' \in W_{k+1}$.
    Thus, since $W_{k+1} \prec W_k$ by Lemma~\ref{lem:residual-worklist-decreases}, the induction hypothesis $\mathcal{P}(W_{k+1})$ applies, and $\mathcal{P}(W_k)$ follows.

    \vspace{0.4em}
    Otherwise, there does not exist such a $P_2'$ for $P_k'$. By assumption, there must be a $P_1' \in W_k' \subseteq W_{k+1}$ for which such a $P_2'$ exists. As shown above, $P' \subseteq U_{P'}$ for all $P' \in W_{k+1}$. Thus, since $W_{k+1} \prec W_k$ by Lemma~\ref{lem:residual-worklist-decreases}, the induction hypothesis $\mathcal{P}(W_{k+1})$ applies, and $\mathcal{P}(W_k)$ follows.
  \end{itemize}
    




  \paragraph{Conclusion.} Before the while loop begins, $W = \{\top\}$ and $\mathsf{Visited} = \varnothing$. At this point, we already know that, with $P'$ instantiated to $P$, \final holds under $\psi$, so by the assumption that $U_{P'}$ contains all necessary atoms, we know there exists $\top \in W$ and $P \subseteq U_{P'}$ s.t. $\{\top\} \cup P = P$ satisfies \final under $\psi$. Again, by the assumption that $U_{P'}$ contains all necessary atoms, we know $P' \subseteq U_{P'}$ for all $P' \in W$. Thus, $\mathcal{P}(W)$ applies, so we know \textsc{FindResidual} returns $(\mathsf{true}, P')$, where $P'$ is weakest residual in $U_{P'}$ satisfying \final $\psi$, and the theorem follows.

\end{enumerate}
\end{proof}

\subsection*{Implementation and Correctness of \textsc{FindWeakestImplicant}}


\begin{algorithm}[t]
\caption{\textsc{FindWeakestImplicant}}
\label{alg:find-weakest-implicant}
\begin{algorithmic}[1]
\Statex \textbf{Input:} Predicate universe $U$, formula $\Phi$
\Statex \textbf{Output:} Weakest conjunction $\chi \subseteq U$ such that $\chi \Rightarrow \Phi$
\Function{FindWeakestImplicant}{$U, \Phi$}
\State $U \gets \{p \in U \mid \Phi \Rightarrow p \}$ \Comment{Filter out atoms not entailed by $\Phi$ (Lemma~\ref{lem:prune})}
\State $\chi \gets \top$
\While{$\chi \not\Rightarrow \Phi$}
    \State $\mathcal{I} \gets \mathsf{Model}(\chi \not\Rightarrow \Phi)$
    \State $S \gets \{p \in U \mid \neg p(\mathcal{I})\}$ \Comment{By Lemma~\ref{lem:exist}, $S \neq \varnothing$}
    \State $\chi\gets \chi\land \bigwedge S$ \Comment{By Lemma~\ref{lem:necessity}, each $p \in S$ is required}
\EndWhile
\State \Return $\chi$
\EndFunction
\end{algorithmic}
\end{algorithm}

Our implementation of the \textsc{FindWeakestImplicant} algorithm is shown in Algorithm~\ref{alg:find-weakest-implicant}, whose correctness we prove via a series of lemmas. Given a finite universe of atoms $U$ and a formula $\Phi$, this procedure performs an initial pruning step using static entailment checks to remove any atom $p \in U$ such that $\Phi \not\Rightarrow p$, justified by Lemma~\ref{lem:prune}. Let $U'$ be the pruned universe. Next, starting with $\top$ as the candidate implicant $\chi$, the algorithm uses an SMT solver to iteratively find models $\mathcal{I}$ that falsify $\chi$. In each iteration, it identifies all atoms $p \in U'$ such that $p(\mathcal{I}) = \mathsf{false}$, and adds them conjunctively to $\chi$. Provided that $\bigwedge U \Rightarrow \Phi$ (so we know $\bigwedge U' \Rightarrow \Phi$), Lemma~\ref{lem:exist} guarantees that such atoms exist, and Lemma~\ref{lem:necessity} ensures they are semantically necessary. Again by Lemma~\ref{lem:exist}, the loop is guaranteed to terminate, and it continues until the candidate implicant $\chi$ becomes strong enough to entail $\Phi$ (in the worst case, $\chi$ becomes $U'$), at which point it is returned as the weakest valid implicant expressible over $U$.


\begin{lemma}
\label{lem:exist}
Let $U$ be a finite set of atomic predicates such that $\bigwedge U \Rightarrow \Phi$. If an interpretation $\mathcal{I}$ falsifies $\Phi$, then there exists some $p \in U$ such that $p(\mathcal{I}) = \textsf{false}$.
\end{lemma}

\begin{proof}
Assume for contradiction that $p(\mathcal{I}) = \mathsf{true}$ for all $p \in U$. Then $\mathcal{I} \models \bigwedge U$, and since $\bigwedge U \Rightarrow \Phi$, we have $\mathcal{I} \models \Phi$, contradicting the assumption that $\mathcal{I} \not\models \Phi$.
\end{proof}

\begin{lemma}
\label{lem:necessity}
Let $U$ be a finite set of atomic predicates such that $\Phi \Rightarrow p$ for all $p \in U$, and let $\Psi$ be any formula such that $\Psi \Rightarrow \Phi$. Then $\Psi \Rightarrow p$ for all $p \in U$.
\end{lemma}

\begin{proof}
Fix any $p \in U$. By assumption, $\Phi \Rightarrow p$ and $\Psi \Rightarrow \Phi$. By transitivity of implication, $\Psi \Rightarrow p$.
\end{proof}

\begin{lemma}
\label{lem:prune}
If $\Phi \not\Rightarrow p$, then $p$ is not included in the weakest implicant of $\Phi$ over $U$.
\end{lemma}

\begin{proof}
By definition, the weakest implicant $\Psi$ is the conjunction of exactly those $q\in U$ that appear in \emph{every} implicant of $\Phi$.  Let
\[
  C \;=\; \{\,q\in U \mid \forall\text{ implicants }R,\;R\models q\},
\]
so $\Psi = \bigwedge_{q\in C}q$.  We show $p\notin C$.

Since $\Phi\not\models p$, there exists a model $\mathcal{M}$ with $\mathcal{M}\models\Phi$ but $\mathcal{M}\not\models p$.  Define
\[
  R \;=\; \bigwedge\,\{\,q\in U \mid \mathcal{M}\models q\}.
\]
Then $R\models\Phi$ (so $R$ is an implicant of $\Phi$) and $R\not\models p$.  Hence, $p$ does not appear in every implicant, i.e., $p\notin C$, as required.
\end{proof}

\section{BMC for Identifying Pushdown Opportunities}\label{sec:bmc}

As discussed in Section~\ref{sec:eval}, we filter benchmarks with potential pushdown opportunities using  bounded model checking. In this appendix, we illustrate how this check is performed. Consider a symbolic dataframe containing two rows ($r_1, r_2$), and let a UDF $F = \mathsf{fold}(x, I, f)$. Given a candidate input filter $Q$, we construct the following symbolic program fragment representing the execution of the UDF on this bounded dataframe:

\begin{lstlisting}[style=pythonstyle,language=Haskell,frame=none,mathescape=true]
a11 = f(I, r1); a12 = f(a11, r2);
if Q(r1) then a21 = f(I, r1) else a21 = I;
if Q(r2) then a22 = f(a21, r2) else a22 = a21;
assert(P(a12) == P(a22) $\land$ (P(a12) $\Rightarrow$ a12 == a22));
\end{lstlisting}

This assertion encodes precisely the correctness condition required for predicate pushdown: it holds if and only if pre-filtering the bounded dataframe using $Q$ preserves the semantics of the original pipeline.

To check feasibility, we computed the weakest precondition (WP) $\phi$ of the above fragment with respect to the assertion, parameterized by the unknown predicate $Q$. To ensure $Q$ is nontrivial (i.e., not equivalent to \textsf{true}), we pose the following query to the SMT solver:
\[
\phi \land \exists r.\,\neg Q(r)
\]
If this query is satisfiable, it indicates that a meaningful predicate $Q$ may exist, warranting further manual verification. Conversely, unsatisfiability conclusively rules out any nontrivial pushdown opportunity, as it demonstrates that no predicate other than the trivial \textsf{true} predicate can satisfy correctness constraints, even in this minimal bounded scenario.

\section{Implementation Details}\label{sec:impl-details}

This appendix elaborates on the beginning of Section~\ref{sec:eval} and provides more information on the implementation of \toolname.

\subsection*{Domain-Specific Language}

Figure~\ref{fig:dsl} shows the grammar of our Python-like, framework-agnostic DSL. We provide a \texttt{fold} construct that takes in a lambda abstraction to express UDFs. The schema for the input dataframe is defined as a tuple of types, and the initializer for a UDF is defined as a tuple of expressions. For UDFs outputting a single scalar value, we express its initializer and accumulator as a 1-tuple. To model null/missing/sentinel values---common in real-world pipelines---the DSL includes an \texttt{Optional} type along with a \texttt{match} construct for case analysis, both mirroring Python's syntax. When compiling programs in the DSL into SMT-LIB, the input schema and UDF accumulator types are encoded as algebraic datatypes.

\begin{figure}[t]
  \small
  \begingroup
  \setlength{\jot}{2pt}
  \renewcommand{\arraystretch}{0.9}
  \[
    \begin{array}{l@{\ }l}
    \langle\mathit{program}\rangle
      &::=\ \langle\mathit{stmts}\rangle\;\texttt{EOF}\\
    \langle\mathit{stmts}\rangle
      &::=\ \langle\mathit{stmt}\rangle\;\langle\mathit{stmts}\rangle
      \;\mid\;\langle\mathit{stmt}\rangle\\
    \langle\mathit{stmt}\rangle
      &::=\ \texttt{IDENT}\;\texttt{=}\;\langle\mathit{expr}\rangle
      \;\mid\;\texttt{IDENT}\;\texttt{:}\;\texttt{(}\langle\mathit{types}\rangle\texttt{) =}\ \;\langle\mathit{expr}\rangle\\
    \langle\mathit{expr}\rangle
      &::=\ \langle\mathit{simple\_expr}\rangle\\
      &\mid\ \langle\mathit{expr}\rangle\;\langle\mathit{binop}\rangle\;\langle\mathit{expr}\rangle\\
      &\mid\ \texttt{not}\;\langle\mathit{expr}\rangle\\
      &\mid\ \texttt{(}\langle\mathit{types}\rangle\texttt{)}\\
      &\mid\ \texttt{(}\langle\mathit{exprs}\rangle\texttt{)}\\
      &\mid\ \texttt{[}\langle\mathit{exprs}\rangle\texttt{]}\\
      &\mid\ \langle\mathit{expr}\rangle\;\texttt{[}\texttt{INT}\texttt{]}\\
      &\mid\ \langle\mathit{expr}\rangle\;\texttt{[}\texttt{INT}\texttt{:]}\\
      &\mid\ \texttt{insert}\;\texttt{(}\langle\mathit{expr}\rangle\;\texttt{,}\;\langle\mathit{expr}\rangle\texttt{)}\\
      &\mid\ \langle\mathit{expr}\rangle\;\texttt{if}\;\langle\mathit{expr}\rangle\;\texttt{else}\;\langle\mathit{expr}\rangle\\
      &\mid\ \texttt{match}\;\langle\mathit{expr}\rangle\;\texttt{:}\;\langle\mathit{cases}\rangle\\
      &\mid\ \texttt{fold}\;\texttt{(}\langle\mathit{expr}\rangle\;\texttt{,}\;\langle\mathit{expr}\rangle\;\texttt{,}\;\langle\mathit{lambda}\rangle\texttt{)}\\
      &\mid\ \texttt{filter}\;\texttt{(}\langle\mathit{expr}\rangle\;\texttt{,}\;\langle\mathit{lambda}\rangle\texttt{)}\\
    \langle\mathit{simple\_expr}\rangle
      &::=\ \langle\mathit{value}\rangle
      \;\mid\;\texttt{(}\langle\mathit{expr}\rangle\texttt{)}\\
    \langle\mathit{binop}\rangle
      &::=\ \texttt{+}
      \;\mid\;\texttt{-}
      \;\mid\;\texttt{*}
      \;\mid\;\texttt{/}
      \;\mid\;\texttt{==}
      \;\mid\;\texttt{>=}
      \;\mid\;\texttt{>}
      \;\mid\;\texttt{<=}
      \;\mid\;\texttt{<}
      \;\mid\;\texttt{and}
      \;\mid\;\texttt{or}\\
    \langle\mathit{value}\rangle
      &::=\ \texttt{BOOL}
      \;\mid\;\texttt{INT}
      \;\mid\;\texttt{FLOAT}
      \;\mid\;\texttt{IDENT}
      \;\mid\;\texttt{None}\\
    \langle\mathit{cases}\rangle
      &::=\ \langle\mathit{case}\rangle\;\langle\mathit{cases}\rangle
      \;\mid\;\langle\mathit{case}\rangle\\
    \langle\mathit{case}\rangle
      &::=\ \texttt{case}\;\langle\mathit{value}\rangle\;\texttt{:}\;\langle\mathit{expr}\rangle\\
    \langle\mathit{types}\rangle
      &::=\ \langle\mathit{type}\rangle\;\texttt{,}\;\langle\mathit{types}\rangle
      \;\mid\;\langle\mathit{type}\rangle\;\texttt{,}\\
    \langle\mathit{type}\rangle
      &::=\ \texttt{bool}
      \;\mid\;\texttt{int}
      \;\mid\;\texttt{float}
      \;\mid\;\texttt{List[}\langle\mathit{type}\rangle\texttt{]}
      \;\mid\;\texttt{Optional[}\langle\mathit{type}\rangle\texttt{]}\\
    \langle\mathit{exprs}\rangle
      &::=\ \langle\mathit{expr}\rangle\;\texttt{,}\;\langle\mathit{exprs}\rangle
      \;\mid\;\langle\mathit{expr}\rangle
      \;\mid\;\varepsilon\\
    \langle\mathit{lambda}\rangle
      &::=\ \texttt{lambda}\;\langle\mathit{params}\rangle\;\texttt{:}\;\langle\mathit{expr}\rangle
      \;\mid\;\texttt{fix}\;\langle\mathit{params}\rangle\;\texttt{:}\;\langle\mathit{expr}\rangle\\
    \langle\mathit{params}\rangle
      &::=\ \texttt{IDENT}\;\texttt{,}\;\langle\mathit{params}\rangle
      \;\mid\;\texttt{IDENT}
      \;\mid\;\varepsilon\\
    \end{array}
  \]
  \endgroup
  \caption{DSL syntax.}
  \label{fig:dsl}
\end{figure}

\subsection*{Constructing Predicate Universes}

To construct the finite predicate universes, we first perform data-flow analyses over the UDF accumulator function $f$, the initializer $I$, and the post-UDF filter $P$ to infer the following:

\begin{itemize}[leftmargin=*]
  \item \emph{Dependencies between input columns and sub-aggregations and those between sub-aggregations.} Concretely, we infer a relation $\mathsf{IsDep}(i, j)$ that indicates the involvement of the intermediate state of sub-aggregation $i$ in the computation of sub-aggregation $j$. In particular, $\forall i.\ \mathsf{IsDep}(i, i)$. In addition, we infer a predicate $\mathsf{IsDepRow}(i)$ to indicate whether sub-aggregation $i$ depends on the current input tuple.

  \item \emph{The monotonicity (or lack thereof) and the lower or upper bound of each sub-aggregation (if monotonic).} Concretely, we compute $\mathsf{IsMono}(i)\in \{\mathsf{Inc}(\beta_{lo}),\mathsf{Dec}(\beta_{hi}),\mathsf{None}\}$ that indicates the monotonicity of sub-aggregation $i$ and its lower-bound (resp. upper-bound) if it monotonically increases (resp. decreases).
\end{itemize}

To construct $U_Q$, we first convert postcondition $P$ and if conditions in the UDF body into conjunctive normal form (CNF) and then extract conjuncts from $P$ as $U_{Q,P}$ and those from the if conditions as $U_{Q,f}$. Doing so allows our synthesis algorithm to make the most incremental decisions when dropping conjuncts before it finds the strongest valid $Q$. For conjuncts from $P$ that involve equalities, we convert an equality into $\ge$ (resp. $\le$) if the corresponding sub-aggregation monotonically increases (resp. decreases). With the collected conjuncts, we leverage the dependency information between input columns and sub-aggregations to substitute all accesses to fields of the UDF output tuple $a$ with accesses to fields of the input tuple $r$ that contribute to the corresponding sub-aggregation, which we informally denote as $U_{Q,P}' = U_{Q,P}[r/a]$ and $U_{Q,f}' = U_{Q,f}[r/a]$. Finally, we construct $U_Q$ as follows:
\[
  U_Q = U_{Q,P}' \cup \bigl\{\bigvee U_{Q,P}'\bigr\} \cup \bigl\{\bigvee U_{Q,f}'\bigr\} \cup \bigl\{\bigvee (U_{Q,P}' \cup U_{Q,f}') \mid U_{Q,P}' \neq \varnothing, U_{Q,f}' \neq \varnothing \bigr\}.
\]

To construct $U_{P'}$, we first convert $P$ into CNF and extract its conjuncts as $U_{P',P}$. Notably, we convert \texttt{match} not into an \texttt{ite} but instead a conjunction of implications: $(\mathsf{IsSome}(...) \Rightarrow ...) \land (\mathsf{IsNone}(...) \Rightarrow ...)$. Next, we construct a set of formulas $U_{P',I} = \{a_i \neq I_i \mid 0 \le i < |I|\}$. Finally, $U_{P'} = U_{P',P} \cup U_{P',I}$.

To construct $U_\bisim$, we observe that while bisimulation invariants generally involve disjunctions, they can be formulated as a conjunction of implications of the form $(\varphi_{k,i} \mid \neg \varphi_{k,i}) \Rightarrow (l \mid \circ)$, where $k$ is either $1$ or $2$ (corresponding to the execution before and after pushdown, respectively), $0 \le i < |I|$ indexes the sub-aggregations, and $l$ is a leaf formula in $\mathcal{L}_{k,i}$. $\circ$ refers to the implication itself, so the consequent may be a leaf formula or another such implication. For each pair of execution $k$ and sub-aggregation $i$, $\mathcal{L}_{k,i}$ includes a fixed set of atomic formulas, with $\mathcal{L}_{1,i} = \{\bot, a_{1,i} = a_{2,i}\}$ and $\mathcal{L}_{2,i} = \{\bot, a_{1,i} = a_{2,i}, a_{2,i} = I_i, a_{2,i} \neq I_i\}$, where $I$ is the initializer. Finally,
\[
  \varphi_{k,i} =
  \begin{cases}
     a_{k,i} \ge \beta_{lo}, & \text{if }\mathsf{IsMono}(i) = \mathsf{Inc}(\beta_{lo}) \text{ and } \neg \mathsf{IsDepRow}(i),\\
     a_{k,i} \le \beta_{hi}, & \text{if }\mathsf{IsMono}(i) = \mathsf{Dec}(\beta_{hi}) \text{ and } \neg \mathsf{IsDepRow}(i),\\
     a_{k,i} \ge c_{k,i},    & \text{if }\mathsf{IsMono}(i) = \mathsf{Inc}(\_) \text{, }\mathsf{IsDepRow}(i) \text{, and }P_{k,i} \text{ involves equality,} \\
     a_{k,i} \le c_{k,i},    & \text{if }\mathsf{IsMono}(i) = \mathsf{Dec}(\_) \text{, }\mathsf{IsDepRow}(i) \text{, and }P_{k,i} \text{ involves equality,} \\
     P_{k,i},                & \text{otherwise}.
  \end{cases}
\]
where $P_{k,i}$ is an atomic formula of $P$ involving $a_{k,i}$ and $c_{k,i}$ is the constant in $P_{k,i}$.

Intuitively, since the optimized execution processes a subset of the rows processed by the original execution, predicates on $a_2$ help constrain the behavior of the original execution, i.e., what values $a_1$ can take on. For this reason, we draw implications from predicates involving $a_2$ to formulas involving $a_1$, but not the opposite, to reduce the number of implications not relevant to establishing an optimal bisimulation invariant. Hence, we construct $U_\bisim$ as follows:
\[
  U_\bisim = \bigcup_{0 \le i < |I|}{\bigl(\{a_{1,i} = a_{2,i}\} \cup U_{\bisim,1,i} \cup U_{\bisim,2,i} \cup U_{\bisim,3,i}\bigr)}
\]
where
\begin{itemize}[leftmargin=*]
    \item $U_{\bisim,1,i} = \{\varphi_{1,i} \Rightarrow l, \neg \varphi_{1,i} \Rightarrow l \mid l \in \mathcal{L}_{1,i}\}$
    \item $U_{\bisim,2,i} = \{\varphi_{2,i} \Rightarrow l, \neg \varphi_{2,i} \Rightarrow l \mid l \in \mathcal{L}_{2,j}, 0 \le j < |I|, \mathsf{IsDep}(i, j) \}$
    \item $U_{\bisim,3,i} = \{\varphi_{2,i} \Rightarrow l, \neg \varphi_{2,i} \Rightarrow l \mid l \in U_{\bisim,1,j}, 0 \le j < |I|, j \neq i, \mathsf{IsDep}(j, i) \}$.
\end{itemize}

\section{Predicate Pushdown Case Studies}\label{sec:case-studies}

We present two case studies that illustrate the kind of pushdown optimizations synthesized by \toolname. Both UDFs are drawn from real-world data-processing applications and require nontrivial bisimulation invariants.

\subsection*{Case Study A: Event Aggregation}

This UDF, adapted from a financial-analytics streaming application\footnote{\url{https://github.com/Nemchinovrp/Fintech-Trading}} that tracks per-ticker trading activity in real time, tallies events by action type and tracks the temporal extent of each ticker group.
The original Spark pipeline (abridged for brevity) is shown in Figure~\ref{fig:eventagg-spark}: a custom \texttt{UserDefinedAggregateFunction} maintains two action counters (\texttt{time} and \texttt{price}) and the running minimum and maximum timestamps (\texttt{first}, \texttt{last}), initialized to zero and \texttt{null} respectively.
After grouping by ticker and applying the UDF, a post-aggregation filter selects groups with sufficient activity whose event window satisfies specific boundary conditions. This pattern is typical of compliance queries that identify instruments actively traded during particular regulatory periods or at known market-event dates.
The filter combines range checks with exact-match exceptions (e.g., tickers whose earliest event predates a regulatory deadline \emph{or} falls on a specific audit date), which makes pushdown nontrivial because the pre-filter must preserve rows relevant to both disjuncts simultaneously.
\begin{center}
\begin{minipage}{0.92\linewidth}
\begin{minted}[fontsize=\footnotesize, frame=lines, framesep=2mm]{scala}
class EventAggregationUdf extends UserDefinedAggregateFunction {
  // input:  (action: String, ts: Timestamp)
  // buffer: (time: Long, price: Long, first: Timestamp, last: Timestamp)
  def initialize(buf: MutableAggregationBuffer) = { buf(0) = 0L; buf(1) = 0L }
  def update(buf: MutableAggregationBuffer, in: Row) = {
    if (in(0) == "time") { buf(0) = buf.getLong(0) + 1L }
    else if (in(0) == "price") { buf(1) = buf.getLong(1) + 1L }
    val ts = in.getTimestamp(1)
    if (!buf.isNullAt(2) && !buf.isNullAt(3)) {
      if (ts.before(buf.getTimestamp(2))) buf(2) = ts
      if (ts.after(buf.getTimestamp(3)))  buf(3) = ts
    } else { buf(2) = ts; buf(3) = ts }
  }
  // ...
}
val Seq(tL, tM, tN, tR) =
  Seq("1980-07-30", "1990-07-30", "1995-07-30", "2001-07-30").map(lit(_).cast("timestamp"))
val udaf = new EventAggregationUdf()
val agg = df.groupBy($"ticker").agg(udaf($"action", $"ts").as("a")).select($"a.*")
val result = agg.filter(
  $"time" > 0 && $"price" > 5 && $"price" <= 18
  && !$"first".isNull && ($"first" === tM || $"first" <= tL)
  && !$"last".isNull  && ($"last" > tR   || $"last" === tN))
\end{minted}
\vspace{-3mm}
\captionof{figure}{Spark pipeline for the Event Aggregation UDF (Case Study~A).}
\label{fig:eventagg-spark}
\end{minipage}
\end{center}

Given the post-UDF filter
\begin{align*}
P(a) ={} & a.\mathsf{time} > 0 \;\wedge\; 5 < a.\mathsf{price} \le 18\\
& {}\wedge\; \neg\,a.\mathsf{first.isNull} \;\wedge\; (a.\mathsf{first} = t_M \vee a.\mathsf{first} \le t_L)\\
& {}\wedge\; \neg\,a.\mathsf{last.isNull} \;\wedge\; (a.\mathsf{last} > t_R \vee a.\mathsf{last} = t_N),
\shortintertext{where $t_L, t_M, t_N, t_R$ are the timestamp constants from Figure~\ref{fig:eventagg-spark}, \toolname synthesizes the split decomposition:}
Q(r) ={} & r.\mathsf{action} = \texttt{"time"} \;\vee\; r.\mathsf{action} = \texttt{"price"}\\
& {}\vee\; r.\mathsf{ts} \le t_M \;\vee\; r.\mathsf{ts} \ge t_N,\\
P'(a) ={} & a.\mathsf{time} \neq 0 \;\wedge\; 5 < a.\mathsf{price} \le 18\\
& {}\wedge\; (a.\mathsf{first} = t_M \vee a.\mathsf{first} \le t_L)\\
& {}\wedge\; (a.\mathsf{last} > t_R \vee a.\mathsf{last} = t_N).
\end{align*}
The pre-filter $Q$ removes only rows whose action is irrelevant ($r.\mathsf{action} \notin \{\texttt{"time"}, \texttt{"price"}\}$) \emph{and} whose timestamp falls strictly between $t_M$ and $t_N$ (i.e., the years 1990--1995), since such rows can affect neither the action counts nor the min/max timestamps in a filter-relevant way.
The residual $P'$ weakens $P$ in two ways: first, $a.\mathsf{time} > 0$ is relaxed to $a.\mathsf{time} \neq 0$, since $5 < a.\mathsf{price} \le 18$ already requires that at least one row be processed; second, $P'$ drops both $\mathsf{isNull}$ guards, because the remaining conjuncts already require at least one row to have been processed, and $\mathsf{first}$ and $\mathsf{last}$ are always initialized together on the first row, so both are necessarily non-null.

To reason about the correctness of this pushdown optimization, a data engineer would need to establish to themselves the following relationships between the original execution (with accumulator state $a_1$) and optimized execution (with accumulator state $a_2$):

\begin{enumerate}[label=\textbf{(\roman*)},nosep,leftmargin=*]

\item \textbf{Counter synchronization.}
At each step, the time count has to agree between the two executions ($a_1.\mathsf{time} = a_2.\mathsf{time}$), so does the price count ($a_1.\mathsf{price} = a_2.\mathsf{price}$), because $Q$ retains every row whose action is \texttt{"time"} or \texttt{"price"}, and only these actions increment the respective counters. Further, time counts must be non-negative, since the UDAF only increments from zero.

\item \textbf{Null synchronization.}
Since the optimized run folds over a subset of the rows seen by the original, a field that the original execution never initialized cannot have been initialized by the optimized run either:
\begin{itemize}[nosep, leftmargin=1.5em]
\item \emph{$a_1.\mathsf{first}$ is null.} Then $a_2.\mathsf{first}$ must also be null.
\item \emph{$a_1.\mathsf{last}$ is null.} Then $a_2.\mathsf{last}$ must also be null.
\end{itemize}

\item \textbf{Boundary synchronization for $\mathsf{first}$.}
$Q$ retains every row with $r.\mathsf{ts} \le t_M$, so when the minimum timestamp falls in this range, the row that achieves the minimum is available to both executions, forcing agreement on $\mathsf{first}$.
\begin{itemize}[nosep, leftmargin=1.5em]
\item \emph{$a_1.\mathsf{first}$ is non-null and ${\le}\, t_M$.} The minimum-achieving row must have $r.\mathsf{ts} \le t_M$ and is therefore retained by $Q$, so the two executions must agree on $\mathsf{first}$ ($a_1.\mathsf{first} = a_2.\mathsf{first}$). Otherwise, either $a_1.\mathsf{first}$ is null or $a_1.\mathsf{first} > t_M$.
\item \emph{$a_2.\mathsf{first}$ is non-null and ${\le}\, t_M$.} Since the original execution sees a superset, $a_1.\mathsf{first} \le a_2.\mathsf{first} \le t_M$ must hold, so the row achieving the original's minimum has $r.\mathsf{ts} \le t_M$ and is therefore retained by $Q$, meaning the optimized run must also see it, giving $a_2.\mathsf{first} \le a_1.\mathsf{first}$. The two executions therefore must agree ($a_1.\mathsf{first} = a_2.\mathsf{first}$). Otherwise, either $a_2.\mathsf{first}$ is null or $a_2.\mathsf{first} > t_M$.
\end{itemize}

\item \textbf{Boundary synchronization for $\mathsf{last}$.}
Symmetrically, $Q$ retains every row with $r.\mathsf{ts} \ge t_N$, so when the maximum timestamp falls in this range, the row achieving the maximum is seen by both executions, forcing agreement on $\mathsf{last}$.
\begin{itemize}[nosep, leftmargin=1.5em]
\item \emph{$a_1.\mathsf{last}$ is non-null and ${\ge}\, t_N$.} The maximum-achieving row must have $r.\mathsf{ts} \ge t_N$ and is therefore retained by $Q$, so the two executions must agree on $\mathsf{last}$ ($a_1.\mathsf{last} = a_2.\mathsf{last}$). Otherwise, either $a_1.\mathsf{last}$ is null or $a_1.\mathsf{last} < t_N$.
\item \emph{$a_2.\mathsf{last}$ is non-null and ${\ge}\, t_N$.} Since the original run sees a superset, $a_1.\mathsf{last} \ge a_2.\mathsf{last} \ge t_N$ must hold, so the row achieving the original's maximum has $r.\mathsf{ts} \ge t_N$ and is therefore retained by $Q$, meaning the optimized run must also see it, giving $a_2.\mathsf{last} \ge a_1.\mathsf{last}$. The two executions therefore must agree ($a_1.\mathsf{last} = a_2.\mathsf{last}$). Otherwise, either $a_2.\mathsf{last}$ is null or $a_2.\mathsf{last} < t_N$.
\end{itemize}

\item \textbf{Co-initialization, optimized to original.}
Because $\mathsf{first}$ and $\mathsf{last}$ are co-initialized on the first row processed, a non-null timestamp in $a_2$ means the optimized execution processed at least one row; since the original run sees a superset, it too must have processed at least one row, so co-initialization guarantees that both $a_1.\mathsf{first}$ and $a_1.\mathsf{last}$ are non-null.
\begin{itemize}[nosep, leftmargin=1.5em]
\item \emph{$a_2.\mathsf{first}$ is non-null.} The optimized execution must have processed at least one row, so $a_2.\mathsf{last}$ is non-null by co-initialization. Since the original run sees a superset, it too must have processed at least one row, so $a_1.\mathsf{last}$ is non-null. Otherwise, $a_2.\mathsf{first}$ is null.
\item \emph{$a_2.\mathsf{last}$ is non-null.} By the same reasoning, the optimized execution processed at least one row, giving $a_2.\mathsf{first}$ as non-null. The original run sees a superset and must also have processed at least one row, so $a_1.\mathsf{first}$ is non-null. Otherwise, $a_2.\mathsf{last}$ is null.
\end{itemize}
The crossed fields ($\mathsf{first} \Rightarrow \mathsf{last}$ and $\mathsf{last} \Rightarrow \mathsf{first}$) are essential: co-initialization holding across the two executions is integral to correctness, which is also what ensures that the boundary synchronization in~(iii) and~(iv) can assume both fields are defined.

\item \textbf{Co-initialization, original to optimized.}
When an extremal timestamp in $a_1$ falls at or beyond a $Q$ boundary ($a_1.\mathsf{first} \le t_M$ or $a_1.\mathsf{last} \ge t_N$), the input row achieving that extremum must itself satisfy $Q$, so the optimized execution processes at least one row, and co-initialization forces the cross-field counterpart to be non-null as well.
\begin{itemize}[nosep, leftmargin=1.5em]
\item \emph{$a_1.\mathsf{first}$ is non-null and ${\le}\, t_M$.} The minimum-achieving row must have $r.\mathsf{ts} \le t_M$, so $Q$ has to retain it; the optimized execution therefore processes at least this row, making $a_2.\mathsf{first}$ non-null, and it must in turn be the case that $a_2.\mathsf{last}$ is non-null as well, by co-initialization. Otherwise, either $a_1.\mathsf{first}$ is null or $a_1.\mathsf{first} > t_M$.
\item \emph{$a_1.\mathsf{last}$ is non-null and ${\ge}\, t_N$.} The maximum-achieving row must have $r.\mathsf{ts} \ge t_N$, so $Q$ has to retain it; the optimized execution therefore processes at least this row, ensuring $a_2.\mathsf{last}$ is non-null, and $a_2.\mathsf{first}$ must in turn be non-null as well, by co-initialization. Otherwise, either $a_1.\mathsf{last}$ is null or $a_1.\mathsf{last} < t_N$.
\end{itemize}

\end{enumerate}

\noindent The bisimulation invariant below captures all of these interdependencies (counter synchronization, null synchronization, boundary synchronization, and cross-execution co-initialization) in a single formula:
\[
\begin{aligned}
 &\underbrace{
   \orig{a_1.\mathsf{time}} = \opt{a_2.\mathsf{time}}
   \;\wedge\;
   \orig{a_1.\mathsf{price}} = \opt{a_2.\mathsf{price}}
   \;\wedge\;
   \opt{a_2.\mathsf{time}} \ge 0
 }_{\text{(i) counter synchronization}}\\
 \mathllap{\wedge}\;&\underbrace{
   (\orig{a_1.\mathsf{first}}.\mathsf{isNull} \Rightarrow \opt{a_2.\mathsf{first}}.\mathsf{isNull})
   \;\wedge\;
   (\orig{a_1.\mathsf{last}}.\mathsf{isNull} \Rightarrow \opt{a_2.\mathsf{last}}.\mathsf{isNull})
 }_{\text{(ii) null synchronization}}\\
 \mathllap{\wedge}\;&\underbrace{
   \begin{aligned}
     &(\neg\,\orig{a_1.\mathsf{first}}.\mathsf{isNull} \wedge \orig{a_1.\mathsf{first}} \le t_M \Rightarrow \orig{a_1.\mathsf{first}} = \opt{a_2.\mathsf{first}}) \;\wedge\\
     &(\neg\,\opt{a_2.\mathsf{first}}.\mathsf{isNull} \wedge \opt{a_2.\mathsf{first}} \le t_M \Rightarrow \orig{a_1.\mathsf{first}} = \opt{a_2.\mathsf{first}})
   \end{aligned}
 }_{\text{(iii) boundary synchronization for $\mathsf{first}$}}\\
 \mathllap{\wedge}\;&\underbrace{
   \begin{aligned}
     &(\neg\,\orig{a_1.\mathsf{last}}.\mathsf{isNull} \wedge \orig{a_1.\mathsf{last}} \ge t_N \Rightarrow \orig{a_1.\mathsf{last}} = \opt{a_2.\mathsf{last}}) \;\wedge\\
     &(\neg\,\opt{a_2.\mathsf{last}}.\mathsf{isNull} \wedge \opt{a_2.\mathsf{last}} \ge t_N \Rightarrow \orig{a_1.\mathsf{last}} = \opt{a_2.\mathsf{last}})
   \end{aligned}
 }_{\text{(iv) boundary synchronization for $\mathsf{last}$}}\\
 \mathllap{\wedge}\;&\underbrace{
   (\neg\,\opt{a_2.\mathsf{first}}.\mathsf{isNull} \Rightarrow \neg\,\orig{a_1.\mathsf{last}}.\mathsf{isNull})
   \;\wedge\;
   (\neg\,\opt{a_2.\mathsf{last}}.\mathsf{isNull} \Rightarrow \neg\,\orig{a_1.\mathsf{first}}.\mathsf{isNull})
 }_{\text{(v) co-initialization, optimized to original}}\\
 \mathllap{\wedge}\;&\underbrace{
   \begin{aligned}
     &(\neg\,\orig{a_1.\mathsf{first}}.\mathsf{isNull} \wedge \orig{a_1.\mathsf{first}} \le t_M \Rightarrow \neg\,\opt{a_2.\mathsf{last}}.\mathsf{isNull}) \;\wedge\\
     &(\neg\,\orig{a_1.\mathsf{last}}.\mathsf{isNull} \wedge \orig{a_1.\mathsf{last}} \ge t_N \Rightarrow \neg\,\opt{a_2.\mathsf{first}}.\mathsf{isNull})
   \end{aligned}
 }_{\text{(vi) co-initialization, original to optimized}}
\end{aligned}
\]
The reasoning above, spanning cross-execution, cross-field dependencies across six groups of interdependent conditions, is highly unintuitive to carry out by hand. This bisimulation invariant captures it all in a single formula that \toolname synthesizes and verifies automatically.

\subsection*{Case Study B: Return-Price Aggregation}

This UDF, adapted from a cryptocurrency high-frequency analysis framework\footnote{\url{https://github.com/QuantLet/Cryptocurrencies-and-Stablecoins-a-high-frequency-analysis}}, simultaneously tracks the \emph{opening} and \emph{closing} trades for each symbol.
As shown in Figure~\ref{fig:returnprice-spark}, the aggregation maintains four buffer fields: \texttt{openPrice} and \texttt{openEpoch} record the price and epoch of the earliest trade (argmin on epoch), while \texttt{closePrice} and \texttt{closeEpoch} record those of the latest trade (argmax on epoch), with prices initialized to~$0$ and epochs to~\texttt{null}.
The post-aggregation filter retains only symbols whose opening and closing prices fall within acceptable bounds and whose trading window spans particular epoch thresholds. This combination of price-range checks with temporal coverage requirements on both ends of the window is typical of return-computation pipelines.
Because the UDF interleaves an argmin and an argmax over the same input fields, the pre-filter must simultaneously preserve rows relevant to both earliest-epoch and latest-epoch tracking.
\begin{center}
\begin{minipage}{0.92\linewidth}
\begin{minted}[fontsize=\footnotesize, frame=lines, framesep=2mm]{scala}
class ReturnPriceAggregator extends UserDefinedAggregateFunction {
  // input:  (price: Double, epoch: Long)
  // buffer: (openPrice: Double, openEpoch: Long?, closePrice: Double, closeEpoch: Long?)
  def initialize(buf: MutableAggregationBuffer) =
    { buf(0) = 0.0; buf(1) = null; buf(2) = 0.0; buf(3) = null }
  def update(buf: MutableAggregationBuffer, row: Row) = {
    val price = row.getDouble(0); val epoch = row.getLong(1)
    if (buf.isNullAt(1) || epoch < buf.getLong(1)) { buf(0) = price; buf(1) = epoch }
    if (buf.isNullAt(3) || epoch > buf.getLong(3)) { buf(2) = price; buf(3) = epoch }
  }
  // ...
}
val result = df.groupBy($"symbol").agg(udaf($"price", $"epoch").as("row"))
  .select($"row.*").filter(
    $"openPrice" > 5 && $"openPrice" <= 100
    && !$"openEpoch".isNull && $"openEpoch" <= 38
    && $"closePrice" >= 10 && $"closePrice" <= 100
    && !$"closeEpoch".isNull && ($"closeEpoch" > 55 || $"closeEpoch" === 53))
\end{minted}
\vspace{-3mm}
\captionof{figure}{Spark pipeline for the Return-Price Aggregation UDF (Case Study~B).}
\label{fig:returnprice-spark}
\end{minipage}
\end{center}

Given the post-UDF filter
\begin{align*}
P(a) ={} & 5 < a.\mathsf{openPrice} \le 100\\
& {}\wedge\; \neg\,a.\mathsf{openEpoch.isNull} \;\wedge\; a.\mathsf{openEpoch} \le 38\\
& {}\wedge\; 10 \le a.\mathsf{closePrice} \le 100\\
& {}\wedge\; \neg\,a.\mathsf{closeEpoch.isNull}
  \;\wedge\; (a.\mathsf{closeEpoch} > 55 \vee a.\mathsf{closeEpoch} = 53),
\shortintertext{\toolname synthesizes the split decomposition:}
Q(r) ={} & r.\mathsf{epoch} \le 38 \;\vee\; r.\mathsf{epoch} \ge 53,\\[2pt]
P'(a) ={} & 5 < a.\mathsf{openPrice} \le 100
  \;\wedge\; (a.\mathsf{openEpoch.isNull} \vee a.\mathsf{openEpoch} \le 38)\\
& {}\wedge\; 10 \le a.\mathsf{closePrice} \le 100
  \;\wedge\; (a.\mathsf{closeEpoch.isNull} \vee a.\mathsf{closeEpoch} > 55 \vee a.\mathsf{closeEpoch} = 53).
\end{align*}
The pre-filter $Q$ retains rows with $\mathsf{epoch} \le 38$ or $\mathsf{epoch} \ge 53$, discarding only rows in the gap $(38,53)$ that can affect neither the opening trade nor the closing trade in a filter-relevant way.
Manually verifying that this decomposition is correct requires tracking several interrelated properties across the original execution (accumulator state $a_1$) and the optimized execution (accumulator state $a_2$):

\begin{enumerate}[label=\textbf{(\roman*)},nosep,leftmargin=*]

\item \textbf{Initialization synchronization.}
The price fields ($\mathsf{openPrice}$, $\mathsf{closePrice}$) are integers initialized to~$0$, while the epoch fields ($\mathsf{openEpoch}$, $\mathsf{closeEpoch}$) are optional and start as~\texttt{null}.
Since each price-epoch pair is overwritten atomically whenever a new extremum is found, a null epoch guarantees the corresponding price has never been touched and remains at~$0$:
\begin{itemize}[nosep, leftmargin=1.5em]
\item \emph{$a_2.\mathsf{openEpoch}$ is null.} Then $a_2.\mathsf{openPrice}$ must be~$0$.
\item \emph{$a_2.\mathsf{closeEpoch}$ is null.} Then $a_2.\mathsf{closePrice}$ must be~$0$.
\end{itemize}

\item \textbf{Null synchronization.}
The optimized execution processes a subset of the rows available to the original. Consequently, if the original has never observed a trade (epoch remains null), the optimized cannot have observed one either:
\begin{itemize}[nosep, leftmargin=1.5em]
\item \emph{$a_1.\mathsf{openEpoch}$ is null.} No row has been processed by the original, so $a_2.\mathsf{openEpoch}$ has to be null as well ($a_1.\mathsf{openEpoch} = a_2.\mathsf{openEpoch}$).
\item \emph{$a_1.\mathsf{closeEpoch}$ is null.} By the same reasoning, $a_2.\mathsf{closeEpoch}$ has to be null ($a_1.\mathsf{closeEpoch} = a_2.\mathsf{closeEpoch}$).
\end{itemize}

\item \textbf{Boundary synchronization, optimized to original.}
When the optimized execution's extremal epoch falls within $Q$'s range, the row achieving that extremum must have passed $Q$. Because the original processes a superset, it too must have seen that row, forcing both executions to agree on the entire price-epoch pair.
\begin{itemize}[nosep, leftmargin=1.5em]
\item \emph{$a_2.\mathsf{openEpoch}$ is non-null and ${\le}\,38$.} Because the original processes a superset of rows, $a_1.\mathsf{openEpoch} \le a_2.\mathsf{openEpoch} \le 38$ has to hold; the original's minimum-achieving row therefore also has $r.\mathsf{epoch} \le 38$ and passes~$Q$, so the optimized run sees it too, yielding $a_2.\mathsf{openEpoch} \le a_1.\mathsf{openEpoch}$. Agreement follows ($a_1.\mathsf{openEpoch} = a_2.\mathsf{openEpoch}$). Otherwise, either $a_2.\mathsf{openEpoch}$ is null or $a_2.\mathsf{openEpoch} > 38$.
\item \emph{$a_2.\mathsf{closeEpoch}$ is non-null and ${\ge}\,53$.} The original's superset property gives $a_1.\mathsf{closeEpoch} \ge a_2.\mathsf{closeEpoch} \ge 53$, so its maximum-achieving row also has $r.\mathsf{epoch} \ge 53$ and passes~$Q$; the optimized run must see it, giving $a_2.\mathsf{closeEpoch} \ge a_1.\mathsf{closeEpoch}$, and the two values must coincide ($a_1.\mathsf{closeEpoch} = a_2.\mathsf{closeEpoch}$). Otherwise, either $a_2.\mathsf{closeEpoch}$ is null or $a_2.\mathsf{closeEpoch} < 53$.
\end{itemize}
Since the UDAF overwrites each price-epoch pair atomically, agreement on the epoch forces agreement on the associated price as well.

\item \textbf{Boundary synchronization, original to optimized.}
Symmetrically, when the original execution's extremal epoch falls within $Q$'s range, $Q$ must have retained the row achieving that extremum, so the optimized execution also sees it and the epochs must agree.
\begin{itemize}[nosep, leftmargin=1.5em]
\item \emph{$a_1.\mathsf{openEpoch}$ is non-null and ${\le}\,38$.} The row that produced this minimum has $r.\mathsf{epoch} \le 38$, so $Q$ retains it; both runs see it, and the two executions must lock to the same $\mathsf{openEpoch}$ ($a_1.\mathsf{openEpoch} = a_2.\mathsf{openEpoch}$). Otherwise, either $a_1.\mathsf{openEpoch}$ is null or $a_1.\mathsf{openEpoch} > 38$.
\item \emph{$a_1.\mathsf{closeEpoch}$ is non-null and ${\ge}\,53$.} The row behind this maximum has $r.\mathsf{epoch} \ge 53$ and passes~$Q$; both executions process it and must arrive at the same $\mathsf{closeEpoch}$ ($a_1.\mathsf{closeEpoch} = a_2.\mathsf{closeEpoch}$). Otherwise, either $a_1.\mathsf{closeEpoch}$ is null or $a_1.\mathsf{closeEpoch} < 53$.
\end{itemize}

\end{enumerate}

\noindent The following bisimulation invariant encodes all of these properties (initialization synchronization, null synchronization, bidirectional boundary synchronization, and associated-price agreement) in a single verifiable formula:
\[
\begin{aligned}
 &\underbrace{
   \begin{aligned}
     &(\opt{a_2.\mathsf{openEpoch}}.\mathsf{isNull} \Rightarrow \opt{a_2.\mathsf{openPrice}} = 0) \;\wedge\\
     &(\opt{a_2.\mathsf{closeEpoch}}.\mathsf{isNull} \Rightarrow \opt{a_2.\mathsf{closePrice}} = 0)
   \end{aligned}
 }_{\text{(i) initialization synchronization}}\\
 \mathllap{\wedge}\;&\underbrace{
   \begin{aligned}
     &(\orig{a_1.\mathsf{openEpoch}}.\mathsf{isNull}
       \Rightarrow \orig{a_1.\mathsf{openEpoch}} = \opt{a_2.\mathsf{openEpoch}}) \;\wedge\\
     &(\orig{a_1.\mathsf{closeEpoch}}.\mathsf{isNull}
       \Rightarrow \orig{a_1.\mathsf{closeEpoch}} = \opt{a_2.\mathsf{closeEpoch}})
   \end{aligned}
 }_{\text{(ii) null synchronization}}\\
 \mathllap{\wedge}\;&\underbrace{
   \begin{aligned}
     &(\neg\,\opt{a_2.\mathsf{openEpoch}}.\mathsf{isNull} \wedge \opt{a_2.\mathsf{openEpoch}} \le 38\\
     &\quad{}\Rightarrow \orig{a_1.\mathsf{openPrice}} = \opt{a_2.\mathsf{openPrice}}
       \;\wedge\; \orig{a_1.\mathsf{openEpoch}} = \opt{a_2.\mathsf{openEpoch}}) \;\wedge\\
     &(\neg\,\opt{a_2.\mathsf{closeEpoch}}.\mathsf{isNull} \wedge \opt{a_2.\mathsf{closeEpoch}} \ge 53\\
     &\quad{}\Rightarrow \orig{a_1.\mathsf{closePrice}} = \opt{a_2.\mathsf{closePrice}}
       \;\wedge\; \orig{a_1.\mathsf{closeEpoch}} = \opt{a_2.\mathsf{closeEpoch}})
   \end{aligned}
 }_{\text{(iii) boundary sync, optimized to original}}\\
 \mathllap{\wedge}\;&\underbrace{
   \begin{aligned}
     &(\neg\,\orig{a_1.\mathsf{openEpoch}}.\mathsf{isNull} \wedge \orig{a_1.\mathsf{openEpoch}} \le 38
       \Rightarrow \orig{a_1.\mathsf{openEpoch}} = \opt{a_2.\mathsf{openEpoch}}) \;\wedge\\
     &(\neg\,\orig{a_1.\mathsf{closeEpoch}}.\mathsf{isNull} \wedge \orig{a_1.\mathsf{closeEpoch}} \ge 53
       \Rightarrow \orig{a_1.\mathsf{closeEpoch}} = \opt{a_2.\mathsf{closeEpoch}})
   \end{aligned}
 }_{\text{(iv) boundary sync, original to optimized}}
\end{aligned}
\]
This reasoning is tedious and error-prone to carry out by hand. \toolname derives the invariant automatically from a universe of 172 candidate atoms.